%% file: Qbipartite_2_NIP.tex
\documentclass[11pt]{amsart}

\makeatletter
\usepackage{amssymb}
\usepackage{enumitem}
\usepackage{array}
\usepackage{booktabs}
\usepackage{colortbl}
\newcolumntype{R}[1]{>{\raggedright\arraybackslash}p{#1}}
\usepackage{xcolor}
\usepackage[margin=1.2in]{geometry}
\usepackage[skins,breakable]{tcolorbox}
\usepackage{tikz}
\usetikzlibrary{arrows.meta}
\usepackage[mathcal]{euscript}

\definecolor{deepblue}{RGB}{47,62,85}
\definecolor{midblue}{RGB}{65,99,137}
\definecolor{urlpurple}{RGB}{98,74,131}
\definecolor{palepurple}{RGB}{236,230,246}
\colorlet{keyframe}{urlpurple!65!white}
\colorlet{keyfill}{palepurple!28!white}

\usepackage{hyperref}
\hypersetup{colorlinks=true, linkcolor=deepblue, citecolor=midblue, urlcolor=urlpurple,
  pdftitle={Classification of bipartite observables II: NIP observables and Todor\v{c}evi\'c's trichotomy},
  pdfauthor={Eduardo Due\~nez and Jos\'e Iovino}}

\RequirePackage{amssymb}
\RequirePackage{mathrsfs}

\newcommand{\CC}{\mathbb{C}}

\newcommand{\NN}{\mathbb{N}}

\newcommand{\QQ}{\mathbb{Q}}
\newcommand{\RR}{\mathbb{R}}

\newcommand{\ZZ}{\mathbb{Z}}

\newcommand{\cR}{\mathcal{R}}
\newcommand{\cS}{\mathcal{S}}
\newcommand{\cT}{\mathcal{T}}

\newcommand{\fA}{\mathfrak{A}}

\newcommand{\fD}{\mathfrak{D}}

\makeatletter
\RequirePackage{amsmath}
\RequirePackage{amssymb}
\RequirePackage{mathtools}

\DeclareMathOperator{\Tr}{Tr}

\newcommand{\one}{\mathbf{1}}

\newcommand{\oA}{\otimes_{\mathsf{A}}}
\newcommand{\oB}{\otimes_{\mathsf{B}}}

\makeatletter
\RequirePackage[nosort]{cleveref}

\newcommand{\di@crefnames}[3]{\crefname{#1}{#2}{#3}\Crefname{#1}{#2}{#3}}
\di@crefnames{section}{Section}{Sections}
\di@crefnames{subsection}{Section}{Sections}
\di@crefnames{subsubsection}{Section}{Sections}
\di@crefnames{chapter}{Chapter}{Chapters}
\di@crefnames{part}{Part}{Parts}
\di@crefnames{appendix}{Appendix}{Appendices}
\di@crefnames{table}{Table}{Tables}
\di@crefnames{figure}{Figure}{Figures}
\di@crefnames{theorem}{Theorem}{Theorems}
\di@crefnames{keytheorem}{Theorem}{Theorems}
\di@crefnames{proposition}{Proposition}{Propositions}
\di@crefnames{corollary}{Corollary}{Corollaries}
\di@crefnames{lemma}{Lemma}{Lemmas}
\di@crefnames{fact}{Fact}{Facts}
\di@crefnames{conjecture}{Conjecture}{Conjectures}
\di@crefnames{definition}{Definition}{Definitions}
\di@crefnames{example}{Example}{Examples}
\di@crefnames{remark}{Remark}{Remarks}
\di@crefnames{question}{Question}{Questions}
\di@crefnames{axiom}{Axiom}{Axioms}
\di@crefnames{notation}{Notation}{Notations}
\di@crefnames{exercise}{Exercise}{Exercises}

\AtBeginDocument{%
}
\makeatletter
\def\@currname{di-theorem-rudin}\def\@curroptions{}
\RequirePackage{amsthm}
\RequirePackage{aliascnt}

\def\di@rudin@parent{section}
\relax

\newtheoremstyle{rudin}%
  {\topsep}{\topsep}{\itshape}{}{\bfseries}{.}{ }%
  {\thmname{#1}\thmnumber{ #2}\thmnote{ (#3)}}
\newtheoremstyle{rudindef}%
  {\topsep}{\topsep}{\upshape}{}{\bfseries}{.}{ }%
  {\thmname{#1}\thmnumber{ #2}\thmnote{ (#3)}}

\theoremstyle{rudin}
\ifx\di@rudin@parent\@empty
  \newtheorem{theorem}{Theorem}
\else
  
\fi
\newcommand{\DiSharedTheorem}[2]{%
  \newaliascnt{#1}{theorem}%
  \newtheorem{#1}[#1]{#2}%
  \aliascntresetthe{#1}%
}
\DiSharedTheorem{proposition}{Proposition}
\DiSharedTheorem{corollary}{Corollary}
\DiSharedTheorem{lemma}{Lemma}
\DiSharedTheorem{fact}{Fact}
\theoremstyle{rudindef}
\DiSharedTheorem{definition}{Definition}
\DiSharedTheorem{example}{Example}
\DiSharedTheorem{remark}{Remark}
\DiSharedTheorem{question}{Question}
\makeatletter
\RequirePackage{xcolor}
\RequirePackage{tikz}
\usetikzlibrary{arrows.meta}

\providecolor{deepblue}{RGB}{30,78,140}
\providecolor{midblue}{RGB}{52,120,190}
\providecolor{urlpurple}{RGB}{120,81,169}
\providecolor{palepurple}{RGB}{236,230,246}

\newcommand{\VennRing}[3]{({(#1+#2)/2},0) ellipse ({(#2-#1)/2} and #3)}

\newcommand{\HeptaCol}[4]{%
  \fill[#3, rounded corners=4pt] (#1,0.85) rectangle ++(3.65,-9.05);
  \node[vtext, font=\footnotesize\bfseries, align=center, anchor=south] at ({#1+1.825},-0.08) {#2};
  \node[vtext, font=\scriptsize\itshape] at ({#1+1.825},-0.24) {#4};}
\newcommand{\HeptaPanel}[5]{%
  \draw[white, line width=1pt, rounded corners=3pt] (#1+0.1,#2-0.05) rectangle ++(3.45,-3.65);
  \node[vtext, font=\footnotesize\bfseries] at ({#1+1.825},#2-2.1) {#3};
  \node[vtext, text width=3.3cm, align=center, anchor=north] at ({#1+1.825},#2-2.28) {\textit{#4}\\[1pt] #5};}
\newcommand{\HeptaBump}[4]{\fill[bump!55] (#1,#4) rectangle ++(#2,#3); \draw[bump, line width=0.4pt] (#1,#4) rectangle ++(#2,#3);}
\makeatletter

\theoremstyle{rudin}
\DiSharedTheorem{keytheorem}{Theorem}
\tcolorboxenvironment{keytheorem}{%
  colback=keyfill, colframe=keyframe, boxrule=0.4pt, arc=1pt,
  left=6pt, right=6pt, top=2pt, bottom=2pt,
  before skip=10pt, after skip=10pt}

\theoremstyle{rudindef}

\crefname{enumi}{Question}{Questions}
\Crefname{enumi}{Question}{Questions}

\newcommand{\Cp}{C_p}
\newcommand{\Bone}{B_1}
\makeatother

\title[Classification of bipartite observables II]{Classification of bipartite observables II:\\ NIP observables and Todor\v{c}evi\'c's trichotomy}
\author{Eduardo Due\~nez}
\address{Department of Mathematics, The University of Texas at San Antonio, One UTSA
Circle, San Antonio, TX 78249, USA}
\email{Eduardo.Duenez@utsa.edu}

\author{Jos\'e Iovino}
\address{Department of Mathematics, The University of Texas at San Antonio, One UTSA
Circle, San Antonio, TX 78249, USA}
\email{jose.iovino@utsa.edu}
\date{\today}
\thanks{This version: \today. The most recent version of this paper is available at \url{https://utsa-quantum-math.github.io/resources/bipartite-observables-2-nip/}; it may supersede the version deposited on arXiv.}

\begin{document}

\begin{abstract}
  This is the second paper in a series that reads the dividing lines of Shelah’s classification theory as properties of bipartite quantum observables.
  For an observable of two independently prepared systems, Alice’s and Bob’s, we study its idealized responses: the expectation values of Alice’s limiting preparations, as functions of Bob’s preparations, which serve as probes, with Bob’s system idealized first.
  We show that the observable is NIP, that is, no sequence of Alice’s preparations lets Bob’s probes realize every finite pattern of high and low values at a fixed positive gap, exactly when all idealized responses are of the first Baire class, pointwise limits of sequences of continuous functions;
  they then form a separable Rosenthal compactum.
  Because Alice can mix preparations, this compactum is convex, and hereditary separability forces it to be metrizable, so that NIP observables occupy three of the four topological classes defined by metrizability and Todorčević’s trichotomy.
  Comparisons of photon numbers, of positions and of events in spacetime realize these three classes, and trees of preparations realize the seven minimal families of the Argyros--Dodos--Kanellopoulos heptachotomy.
  For the ferromagnetic transverse-field Ising model at positive temperature, the finite-volume magnetizations with offset longitudinal fields form a stable family when the spontaneous magnetization vanishes, and a family whose closure is first countable but not hereditarily separable when it does not.
\end{abstract}

\maketitle

\setcounter{tocdepth}{1}
\tableofcontents

%#
%## SECTION: Introduction
%#
\section{Introduction}\label{sec:intro}

In Part~I~\cite{Duenez-Iovino:2026a} we studied when two independently prepared quantum systems, Alice’s and Bob’s, can be idealized in either order without changing the expectation value of a bipartite observable.
We work, as there, with separable Hilbert spaces and bounded self-adjoint observables.
An actual preparation is described by a density operator, which defines a state of the algebra of bounded operators on its system’s Hilbert space;
the remaining states, which arise as limits of actual preparations in the topology of convergence of expectation values, serve as the mathematical model of idealized preparations.
An observable is \emph{stable} when its expectation value does not depend on which system is idealized first, that is, when, for all sequences of actual preparations of the two systems, the two iterated limits of its expectation values agree whenever both exist.
In Part~I we identify this condition with properties from analysis, topology, probability and model theory.
Stable examples include the swap test, which estimates the overlap of two prepared states, Hong--Ou--Mandel coincidences, in which two photons meeting at a beam splitter leave by different ports, and discrete equality tests, such as “equal photon numbers”;
these observables are even uniformly stable, with quantitative bounds independent of the dimensions of the Hilbert spaces.
In this paper we ask what regularity the idealized responses retain when the two orders of idealization give different answers.

Other observables, no less natural, are of a different kind:
for a quantity of Alice’s system, given by a self-adjoint operator
$A$, and a quantity of Bob’s, given by a self-adjoint operator $B$, the order comparison $[A<B]$ is the orthogonal projection that represents the event “Bob’s value of $B$ exceeds Alice’s value of $A$”.
The comparison $[A<B]$ is unstable whenever $A$ and
$B$ have pure point spectra, that is, orthonormal bases of eigenvectors, and their eigenvalues are unbounded above, as for the observables “Bob’s system has more photons than Alice’s” and “Bob’s oscillator has the higher energy”:
if Alice and Bob prepare eigenstates of ever higher energy, the probability that Bob’s energy is the higher is $1$ when Bob’s system is idealized first and $0$ when Alice’s system is.
The order of the idealizations matters because an order relation is built into the observable, and observables of this kind supply the principal examples of this paper.

An order is not yet independence.
Call a sequence of Alice’s preparations \emph{independent} for an observable if there are real numbers
$a<b$ such that, for any two disjoint finite sets of terms of the sequence, a single preparation of Bob’s system, a \emph{probe}, gives expectation values at most $a$ on the first set and at least $b$ on the second.
The existence of such a sequence is the \emph{independence property for Alice’s responses};
its absence is \emph{NIP for Alice’s responses}, where NIP abbreviates “no independence property”.
Unqualified NIP will refer to this direction, since exchanging Alice and Bob can change the property (see \cref{cor:registers}).
In \cref{sec:orders} we show that $[A<B]$ is uniformly NIP for all self-adjoint operators $A$ and $B$:
for each gap $b-a$, the size of a finite set of preparations on which probes realize every pattern is bounded by a number depending only on the gap.
An independent sequence occurs instead for the fringe observable, whose local detector response follows an interference pattern (see \cref{prop:fringes}).
If Alice’s preparations produce fringes of wavenumbers $1,8,8^2,\dots$ on Bob’s screen, in units in which the screen has length $2\pi$, then every finite choice of bright and dark responses is realized by a position of Bob’s detector.

\enlargethispage{3pt}
The division of first-order theories, sets of axioms in a formal language, by properties such as stability and NIP, which model theorists call dividing lines, is the subject of Shelah’s classification theory~\cite{Shelah:1990a}.
Part~I and the present paper study these properties for the expectation predicate of a fixed observable in its preparation structure, the structure formed by the preparations of the two systems together with this predicate:
stability in Part~I and the weaker condition NIP here.
Uniform stability and uniform NIP require this predicate to have the corresponding property in every model of the theory of the preparation structure, and stability or NIP of the entire theory requires it for every formula.
In model theory, an infinite linear order is the prototype of an unstable NIP structure, that is, of a set equipped with relations and operations whose theory is unstable and NIP.
The rational order $(\QQ,<)$ is unstable, yet it is NIP:
over a finite set of parameters it has only polynomially many types in the size of the set, a type being a complete description of a single element in terms of the parameters.
The ordered field of real numbers is likewise unstable and NIP, and so, more generally, is every o-minimal structure, that is, every ordered structure in which every set of elements defined by a formula with parameters
(a definable set) is a finite union of points and intervals;
the o-minimal structures are the best-known unstable NIP structures.\footnote{Other standard examples are the ordered group of integers
  $(\ZZ,+,<)$, whose theory is Presburger arithmetic, and, more generally, ordered abelian groups~\cite{Gurevich-Schmitt:1984};
  the fields of
  $p$-adic numbers, for every prime $p$;
  and algebraically closed valued fields
  (see~\cite[Chapter~2]{Simon:2015b}).
  For o-minimal
  structures see~\cite{Pillay-Steinhorn:1986a,Knight-Pillay-Steinhorn:1986};
  the real field with the exponential function is o-minimal by a theorem of Wilkie~\cite{Wilkie:1996}.}

In these model-theoretic examples the formulas are two-valued, even when the underlying set is a continuum.
For observables the formulas are real-valued, because they are expectation values, and a fixed positive gap separates the high and low values in the independence property.
Real-valued forms of NIP already occur in quantum theory through the fat-shattering dimension governing the learnability of quantum states and measurements in finite dimensions~\cite{Aaronson:2007,Cheng-Hsieh-Yeh:2015}.
Here the object of classification is the entire compactum of idealized responses in the topology of pointwise convergence.
The relevant structure theory is that of Rosenthal compacta, compact spaces homeomorphic to pointwise compact families of functions of the first Baire class on a Polish space (a separable, completely metrizable space).
The passage from finite patterns to the regularity of this compactum is the dictionary on which the classification rests;
we state it, with the results that follow from it, in \cref{sec:main}.
The classification retains both the real values of expectation predicates and their affine dependence on probabilistic mixtures, and we formulate it in Keisler’s model theory of real-valued structures~\cite{Keisler:2023} (see \cref{sec:logical-framework}).

\paragraph{Related work.}
Model theory has met operator algebras and quantum theory along other routes.
Farah, Hart and Sherman developed the continuous model theory of C$^*$-algebras,
the algebras of bounded operators on a Hilbert space that are closed under adjoints and in the norm, and of tracial von Neumann algebras, those that are closed under adjoints and in the weak operator topology
(the topology of convergence of all matrix elements) and carry a trace, a faithful normal state that vanishes on commutators;
in this setting they studied stability~\cite{Farah-Hart-Sherman:2013};
there the structures are the algebras themselves, and stability is a property of their theories, whereas here the object is a single observable together with its idealized responses.
Zilber gave a model-theoretic semantics of the canonical commutation relation~\cite{Zilber:2016b}.
On the side of real-valued NIP, Ben Yaacov showed that NIP is preserved under Keisler randomization,
the passage from a theory to a theory of random variables with values in its models~\cite{Yaacov:2009}, whence averages of NIP relations over random parameters are again NIP;
the bound of \cref{prop:uniform-orders} is an elementary, quantitative instance of this phenomenon for relations between quantities measured by Alice and by Bob.
Hrushovski and Pillay studied invariant measures in NIP theories~\cite{Hrushovski-Pillay:2011b},
Simon connected NIP with Rosenthal compacta~\cite{Simon:2015a}, and Glasner and Megrelishvili used Todorčević’s trichotomy to build a hierarchy of tame dynamical systems~\cite{Glasner-Megrelishvili:2022}.
To our knowledge, none of these works considers the orders of limits of a bipartite observable,
the compactum of its idealized responses, or the place of physical observables in Todorčević’s trichotomy and in the Argyros--Dodos--Kanellopoulos heptachotomy;
Part~I concerns the first of these topics, and the present paper concerns all three.

\paragraph{Organization of the paper.}
We state the main results in \cref{sec:main}.
The facts about Rosenthal compacta that we need, and a classification of countable NIP families of continuous functions that we derive from them (see \cref{thm:nipclasses}), occupy \cref{sec:prelim,sec:classification}.
In \cref{sec:responses,sec:levels} we attach to each observable its compactum of idealized responses, prove \cref{thm:levels}, and show that, for the compactum of idealized responses of an NIP observable, hereditary separability implies metrizability
(see \cref{thm:convex} and \cref{cor:nip2}).
The examples follow in \cref{sec:orders,sec:hepta};
spontaneous magnetization, classical and quantum, is the subject of \cref{sec:ssb}, and two further orders of limits are treated in \cref{sec:spectral}.
Open questions are collected in \cref{sec:questions}.
Standard arguments are cited, and sketched where the adaptation is not routine.

%#
%## SECTION: Main results
%#
\section{Main results}\label{sec:main}

\noindent The results rest on a dictionary between finite patterns of expectation values and the regularity of idealized responses, which we now describe.
Fix a bounded self-adjoint observable $C$ on the tensor product $H_1\otimes H_2$ of Alice’s and Bob’s separable Hilbert spaces.
For each density operator $\rho$ on $H_1$, the response on actual probes is the function $r_\rho\colon\fD(H_2)\to[-\|C\|,\|C\|]$ given by $r_\rho(\sigma)=\varphi_C(\rho,\sigma)=\Tr((\rho\otimes\sigma)C)$, where $\fD(H_2)$ denotes Bob’s density operators.
Averaging against Alice’s preparation gives an effective observable $T_C(\rho)$ on $H_2$, characterized by $\Tr(\sigma T_C(\rho))=\Tr((\rho\otimes\sigma)C)$ for every density operator $\sigma$ on $H_2$.
The probe algebra $\fA_C$ is the smallest unital norm-closed algebra containing these effective observables and their adjoints, and its state space $Y_C$ is the probe space.
An actual probe $\sigma\in\fD(H_2)$ determines the state $\mathfrak{p}_\sigma\colon\fA_C\to\CC$ given by $\mathfrak{p}_\sigma(a)=\Tr(\sigma a)$, an element of $Y_C$.
For fixed $C$ and $\rho$, the completed response is the function $\mathfrak{r}_\rho\colon Y_C\to[-\|C\|,\|C\|]$ given by $\mathfrak{r}_\rho(\omega_2)=\omega_2(T_C(\rho))$;
it satisfies $\mathfrak{r}_\rho(\mathfrak{p}_\sigma)=r_\rho(\sigma)=\varphi_C(\rho,\sigma)$ for every actual probe $\sigma$.
In \cref{lem:probe-space} we show that $Y_C$ is compact metrizable in the topology of pointwise convergence on $\fA_C$, that the records $\mathfrak{p}_\sigma$ of actual probes are dense in it, and that every completed response $\mathfrak{r}_\rho$ of an actual preparation is continuous.
We call the members of the pointwise closure of the completed responses of actual preparations the \emph{idealized responses}, including the completed responses themselves.
In \cref{lem:closure} we identify this closure with the responses represented by all states of $B(H_1)$, with Bob’s system idealized first;
distinct states may represent the same response.

The compactum of idealized responses of an NIP observable is a separable Rosenthal compactum (see \cref{lem:separable}).
A space is \emph{hereditarily separable} if every subspace has a countable dense subset, and \emph{first countable} if every point has a countable base of neighborhoods.
For Rosenthal compacta, metrizability implies hereditary separability, which implies first countability (see \cref{sec:classification}), and Todorčević’s trichotomy~\cite{Todorcevic:1999} supplies a canonical witness for the failure of each converse:
a non-metrizable separable Rosenthal compactum that is hereditarily separable contains a copy of the split interval;
one that is first countable but not hereditarily separable contains a copy of the Alexandroff duplicate of the Cantor set or of its analogue for the split Cantor set;
and one that is not first countable contains a copy of the one-point compactification of a discrete set of size continuum.
The split interval is the unit interval with each point doubled into two adjacent points, in the order topology, and the split Cantor set is the Cantor set doubled in the same way;
the Alexandroff duplicate of the Cantor set adds to the Cantor set an isolated twin of each of its points, in such a way that the twins of a sequence of points converging to a point distinct from all of them converge to the same point.
We call the four classes of separable Rosenthal compacta so obtained, the metrizable ones and the three classes of the trichotomy, the \emph{four topological classes}.
An NIP family with separable closure is NIP$^{\mathrm{fc}}$, NIP$^{\mathrm{hs}}$ or NIP$^{\mathrm{met}}$ when its closure is, respectively, first countable, hereditarily separable or metrizable (see \cref{sec:classification});
in particular, NIP$^{\mathrm{met}}$ means that one countable set of probes distinguishes all idealized responses (see \cref{prop:countable-probes,prop:levels-separation}).

Argyros, Dodos and Kanellopoulos refined the trichotomy to the level of families, in their heptachotomy~\cite{Argyros-Dodos-Kanellopoulos:2008b}.
Index a pointwise bounded family of pairwise distinct functions on a Polish space by the dyadic tree, the set of finite sequences of zeros and ones, and suppose that its pointwise closure consists of functions of the first Baire class.
Then there is a dyadic subtree, a set of nodes that,
ordered by extension, forms a copy of the dyadic tree, such that the subfamily indexed by it, reindexed by the dyadic tree through the isomorphism that also preserves the lexicographic order, is equivalent to one of exactly seven minimal families;
two families indexed by the same set are \emph{equivalent} when the correspondence between them extends to a homeomorphism of their pointwise closures,
and a family is \emph{minimal} when it is equivalent to its restriction to every dyadic subtree, reindexed in the same way.
For an NIP observable, a tree of Alice’s preparations with pairwise distinct responses gives such a family,
for instance a tree of successively sharper localizations of a particle, and the heptachotomy shows that, up to passing to a subtree, there are seven ways, and only seven, in which the responses of such a tree can approach their idealized limits.
Through the connection among NIP, the Vapnik--Chervonenkis dimension
(the two-valued form of the fat-shattering dimension) and Rosenthal compacta, the trichotomy also classifies NIP formulas whose parameters range over a compact metrizable space~\cite[Section~4.2]{Duenez-Iovino-MatosWiederhold-Salvetti-Tall:2026}.

The classification of observables and of families of responses is drawn in \cref{fig:venn},
with the physical systems that inhabit each class and the minimal families that each class contains, and in \cref{fig:hepta} each of the seven minimal families of the heptachotomy is followed to its idealized limits.
The classes are listed in \cref{tab:classes} and the minimal families in \cref{tab:families}, each with a physically natural example and with the feature of the example that the class or family isolates.
The two tables distinguish the classification of a whole response space from that of a selected tree of responses.

\begin{figure}[ht]
  \centering
  \makebox[\textwidth][c]{\input{venn_body.tex}}
  \caption{The classes of observables and of families of responses, nested from uniformly stable to NIP, with non-NIP examples outside the nested classes. Physical examples and the minimal families $D_1,\dots,D_7$ are placed according to the classes of their response closures.}
  \label{fig:venn}
\end{figure}

\begin{table}[ht]
  \centering
  \small
  \renewcommand{\arraystretch}{1.3}
  \setlength{\tabcolsep}{4pt}
  \arrayrulecolor{urlpurple!55!white}
  \begin{tabular}{R{1.9cm}R{1.95cm}R{4.0cm}R{4.4cm}r}
    \toprule
    \rowcolor{palepurple}
    Class                                   & Topological criterion or witness & Meaning for responses           & Physical setting                 & Result                                    \\
    \midrule
    stable                                  & continuous       & idealizations commute           & swap test; gapped transverse-current response        & \ref{thm:sc}                       \\
    NIP$^{\mathrm{met}}$                                 & metrizable       & one countable set of probes separates all responses   & photon-number comparison; spin-chain scanning             & \ref{thm:orders}, \ref{cor:registers}                   \\
    \rowcolor{black!4}
    $\mathrm{NIP}^{\mathrm{hs}}\setminus\mathrm{NIP}^{\mathrm{met}}$ & split interval   & hereditarily separable but nonmetrizable   & absent for full observable compacta              & \ref{thm:convex}                   \\
    $\mathrm{NIP}^{\mathrm{fc}}\setminus\mathrm{NIP}^{\mathrm{hs}}$ & $\hat{D}(2^\NN)$ or $\hat{D}(S(2^\NN))$ & each response has its own countable determining probes       & positions; quadratures; Ising with offsets    & \ref{thm:orders}, \ref{prop:cv-helly}, \ref{thm:qssb}                     \\
    $\mathrm{NIP}\setminus\mathrm{NIP}^{\mathrm{fc}}$   & $\hat{A}(2^\NN)$ & some response has no countable determining set of probes & causal order; phases on a circle & \ref{thm:causal}, \ref{thm:phases} \\
    not NIP                                 & $\beta\NN$       & one infinite preparation sequence admits every finite high/low pattern at a fixed positive gap          & interference fringes             & \ref{prop:fringes}                 \\
    \bottomrule
  \end{tabular}

  \vspace{8pt}
  \caption{From the regularity of idealized responses to their physical interpretation, class by class.}
  \label{tab:classes}
\end{table}

\begin{table}[ht]
  \centering
  \small
  \renewcommand{\arraystretch}{1.3}
  \setlength{\tabcolsep}{5pt}
  \arrayrulecolor{urlpurple!55!white}
  \begin{tabular}{R{1.3cm}R{2.45cm}R{4.4cm}R{5.93cm}}
    \toprule
    \rowcolor{palepurple}
    Family & Closure             & Pattern of idealization  & Selected preparations                            \\
    \midrule
    $D_1$  & $A(2^{<\NN})$       & escape to infinity       & photon number states                   \\
    $D_2$  & $2^{\le\NN}$        & infinite configurations  & a spin chain read site by site         \\
    \addlinespace[3pt]
    $D_3$  & $\hat{S}_+(2^\NN)$  & approach from the right  & a particle in a box                    \\
    $D_4$  & $\hat{S}_-(2^\NN)$  & approach from the left   & a particle in a box                    \\
    \addlinespace[3pt]
    $D_6$  & $\hat{D}(2^\NN)$    & twins over a Cantor set  & a spin chain, sites and templates      \\
    $D_7$  & $\hat{D}(S(2^\NN))$ & approach from both sides & a particle in a box, half on each side \\
    \addlinespace[3pt]
    $D_5$  & $\hat{A}(2^\NN)$    & a point-supported limit  & light cones; templates                 \\
    \bottomrule
  \end{tabular}

  \vspace{8pt}
  \caption{The seven minimal families $D_1,\dots,D_7$ of the heptachotomy and their closures, the canonical compacta of \cref{sec:classification}, grouped, as in \cref{fig:hepta}, by the four topological classes (see \cref{thm:hepta-natural}).}
  \label{tab:families}
\end{table}

\begin{figure}[ht]
  \centering
  \makebox[\textwidth][c]{\input{hepta_body.tex}}
  \caption{Schematic realizations of the seven minimal families of the heptachotomy, grouped by the class, metrizable or one of the three classes of Todorčević’s trichotomy, that contains their closures (see \cref{thm:hepta-natural}). The NIP level at each class is the largest level it reaches.}
  \label{fig:hepta}
\end{figure}

The results behind the tables are the following.
\begin{enumerate}[label=\textup{(\Alph*)},leftmargin=*]
\item
  \emph{The dictionary.}
  An observable is stable exactly when every idealized response is continuous, and NIP exactly when every idealized response is of the first Baire class, that is, a pointwise limit of a sequence of continuous functions;
  the idealized responses of an NIP observable form a Rosenthal compactum, and every one of them is the limit of a sequence of responses of actual preparations (see \cite[Theorem~3F]{Bourgain-Fremlin-Talagrand:1978}).
  When the observable is not NIP, some sequence of Alice’s actual preparations is independent, and the pointwise closure of the responses of a subsequence is a copy of $\beta\NN$, the space of ultrafilters on the natural numbers (see \cref{thm:levels}).
\item
  \emph{Mixtures force metrizability.}
  Alice can mix two preparations, that is, make one or the other at random with given probabilities, and the response of the mixture is the corresponding convex combination of the two responses;
  hence the idealized responses form a convex set.
  A separable convex Rosenthal compactum that is hereditarily separable is metrizable (see \cref{thm:convex});
  hence NIP$^{\mathrm{hs}}$ observables are NIP$^{\mathrm{met}}$ (see \cref{cor:nip2}), and NIP observables fall into three of the four topological classes: the metrizable compacta and the last two classes of the trichotomy (see \cref{thm:observable-classes}).
\item
  \emph{Natural examples.}
  The comparison of photon numbers is NIP$^{\mathrm{met}}$;
  the comparison of positions or of quadratures of light (the two field amplitudes of a mode, which behave like a position and a momentum) is NIP$^{\mathrm{fc}}$
  but not NIP$^{\mathrm{hs}}$;
  and the causal order of a spacetime with one dimension of space and the comparison of two phases on a circle are NIP but not NIP$^{\mathrm{fc}}$ (see \cref{thm:orders,thm:causal,thm:phases} and \cref{prop:cv-helly}).
  Each minimal family is realized by a tree of preparations of one of these observables or of an observable that probes a chain of spins, by reading it site by site, by matching it with a template, or both (see \cref{thm:hepta-natural}).
  The scanning observable, which reads the chain site by site, is NIP for Alice’s responses but not for Bob’s, so NIP is not symmetric for observables (see \cref{cor:registers}), although uniform NIP is (see \cref{rem:symmetry}).
\item
  \emph{The quantum Ising model.}
  For the Ising model in a transverse field, a lattice of spins coupled along one axis in a field perpendicular to it, consider the family of magnetizations of finite systems, as functions of the longitudinal field
  (the field along the axis of the couplings), together with an offset field that moves the transition point
  (the field at which the magnetization changes sign).
  At positive temperature this family is stable, with a metrizable closure, when the spontaneous magnetization vanishes, and it is NIP$^{\mathrm{fc}}$ but not NIP$^{\mathrm{hs}}$ when the spontaneous magnetization is positive;
  the family without offsets is NIP$^{\mathrm{met}}$ in both cases (see \cref{thm:qssb}).
  The proof uses the real analyticity of the pressure (the limit per site of the logarithm of the partition function, the normalizing constant of the thermal equilibrium state) away from zero longitudinal field, which we derive from a Lee--Yang theorem for this model:
  the partition function of each finite system has no zeros off the imaginary axis of the complex longitudinal field (see \cref{prop:q-ly,prop:qanalytic}).
  We include a proof of this statement and discuss its relation to the quantum Lee--Yang theorems of Asano~\cite{Asano:1970b}, Suzuki and Fisher~\cite{Suzuki-Fisher:1971b}, and their formulation by Fröhlich and Rodriguez~\cite{Frohlich-Rodriguez:2012b}.
\item
  {\sloppy\emph{Orders of limits.}
  In the classical and quantum Ising models at positive temperature, the finite-volume magnetizations, with or without offset fields, form an unstable family exactly when the spontaneous magnetization is positive:
  a magnetization can remain when the volume tends to infinity before the external field is switched off, while reversing these limits gives zero.
  For the spectral families considered here, the Banks--Casher relation~\cite{Banks-Casher:1980} identifies the obstruction to stability with the limiting density of eigenvalues near zero of the Dirac operator governing quarks in a gluon field,
  and the criterion of Scalapino, White and Zhang~\cite{Scalapino-White-Zhang:1993}, which distinguishes superconductors, metals and insulators by their current responses, identifies it with the dynamic current-response limit minus the static one (see \cref{thm:ssb,thm:qssb,thm:chiral,thm:sc});
  for these families, stability is equivalent to the vanishing of the limiting valence condensate of the chosen ensemble, or to the agreement of the static and dynamic current-response limits, respectively.\par}
\end{enumerate}

The examples explain where the classes come from.
The compactum of idealized responses of the causal order fails to be first countable because an event that occurs at the time of Bob’s detection can influence a detection only at its own position.
The comparison of two phases lies in the same class:
on a circle, for each pair of antipodal points there is an idealized response that differs from the constant response of the uniform distribution only on probes that give weight to one of the two points.
Trees of preparations realize all seven minimal families, although the full compactum is convex, because the closure of the responses of a tree is a closed subset of that compactum and need not be convex.
At each translated transition point of the Ising model with positive spontaneous magnetization, the responses of the two oppositely magnetized pure phases form a split pair, two functions that differ only at the transition point, and the response of their equal mixture is the midpoint of the pair;
these midpoints give the same obstruction to hereditary separability as in the proof of \cref{thm:convex}.
\Cref{thm:convex} itself is a special case of a known result on locally connected Rosenthal compacta (see \cref{rem:convex-literature}), of which we give a direct proof.
Throughout, the structure theorems supply canonical witnesses for the response compacta, while the location of each physical example in its class is proved directly.

%#
%## SECTION: A family of continuous functions is NIP exactly when its closure consists of Baire-1 functions
%#
\section{A family of continuous functions is NIP exactly when its closure consists of Baire-1 functions}\label{sec:prelim}

\noindent We begin with the topological dichotomy on which the paper rests:
a pointwise bounded family of continuous functions on a compact metrizable space either has the independence property or has its pointwise closure inside the functions of the first Baire class.
In most of the paper the functions will be responses of Alice’s preparations and the points will be Bob’s probes.

%#
%### SUBSECTION: Functions of the first Baire class
%#
\subsection{Functions of the first Baire class}
Throughout, the set $\NN$ of natural numbers contains $0$.
For a topological space $Y$, let $\Cp(Y)$ be the set of continuous functions $Y\to\RR$ with the topology of pointwise convergence.
A function $f\colon Y\to\RR$ is of the \emph{first Baire class}, or \emph{Baire-1}, if it is the pointwise limit of a sequence of continuous functions, and we write
$\Bone(Y)$ for the set of these functions, again with the topology of pointwise convergence.
For $\cR\subseteq\RR^Y$, we write $\overline{\cR}$ for the closure of $\cR$ in $\RR^Y$, and we call $\cR$ \emph{pointwise bounded} if $\sup_{r\in\cR}|r(y)|<\infty$ for every $y\in Y$.

We use the classical characterizations of Baire-1 functions on Polish spaces, the separable, completely metrizable spaces:
condition~(3) of \cref{fact:baire} below is Baire’s characterization, and condition~(2) concerns $F_\sigma$ sets, the countable unions of closed sets.

\begin{fact}[see {\cite[Section~10]{Todorcevic:1997}}]\label{fact:baire} Let $Y$ be a Polish space and $f\colon Y\to\RR$ a function.
  The following are equivalent:
  \begin{enumerate}[label=\textup{(\arabic*)}]
  \item
    the function $f$ is Baire-1;
  \item
    for every open $U\subseteq\RR$, the set $f^{-1}[U]$ is $F_\sigma$;
  \item
    for every nonempty closed $F\subseteq Y$, the restriction $f|_F$ has a point of continuity.
  \end{enumerate}
  Moreover, if $f$ is not Baire-1, then there are nonempty countable sets $Z_0,Z_1\subseteq Y$ with $\overline{Z_0}=\overline{Z_1}$ and reals $a<b$ such that $f\le a$ on $Z_0$ and $f\ge b$ on $Z_1$.
\end{fact}

The last clause of \cref{fact:baire} lets us test membership in $\Bone(Y)$ on countable sets, which is what makes pointwise closures manageable.

\begin{lemma}\label{lem:relcompact}
  Let $Y$ be a Polish space and $\cR\subseteq\Bone(Y)$ a pointwise bounded family.
  The following are equivalent:
  \begin{enumerate}[label=\textup{(\roman*)}]
  \item
    the family $\cR$ is relatively compact in $\Bone(Y)$;
  \item
    every sequence in $\cR$ has a cluster point in $\Bone(Y)$;
  \item
    we have $\overline{\cR}\subseteq\Bone(Y)$.
  \end{enumerate}
\end{lemma}

\begin{proof}[Sketch]
  This is the Bourgain--Fremlin--Talagrand argument~\cite{Bourgain-Fremlin-Talagrand:1978}.
  The implication (i)$\Rightarrow$(ii) is clear, and (iii)$\Rightarrow$(i) is Tychonov’s theorem.
  For (ii)$\Rightarrow$(iii), suppose that some $f\in\overline\cR$ is not Baire-1, and take $Z_0$, $Z_1$, $a$ and $b$ as in \cref{fact:baire}.
  A sequence in $\cR$ converging to $f$ on the countable set $Z_0\cup Z_1$ has no cluster point in $\Bone(Y)$, because such a point would agree with $f$ on $Z_0\cup Z_1$ and would therefore have no point of continuity on $\overline{Z_0}=\overline{Z_1}$.
\end{proof}

%#
%### SUBSECTION: The independence property
%#
\subsection{The independence property}
The obstruction to Baire-1 closures is combinatorial, as \cref{thm:dichotomy} below shows:
it is a sequence of functions whose sublevel and superlevel sets meet in every finite pattern.

\begin{definition}\label{def:ip}
  Let $Y$ be a set.
  A family $\cR\subseteq\RR^Y$ has the \emph{independence property} if there are a sequence
  $(r_n)_{n\in\NN}$ in $\cR$ and reals $a<b$ such that, for all disjoint finite sets $E,F\subseteq\NN$, there is $y\in Y$ with $r_n(y)\le a$ for $n\in E$ and $r_n(y)\ge b$ for $n\in F$.
  Otherwise, the family $\cR$ is \emph{NIP}.
\end{definition}

On the Cantor space $\{0,1\}^\NN$ of infinite sequences of zeros and ones, the family of coordinate functions
$p_n$, given by $p_n(y)=y(n)$ for $n\in\NN$, has the independence property, and the closure of $\{p_n:n\in\NN\}$ in $\RR^{\{0,1\}^\NN}$ is a copy of $\beta\NN$, the space of ultrafilters on $\NN$.
The next two facts say that this example is the only obstruction.

\begin{fact}[Rosenthal’s dichotomy;
    see {\cite{Rosenthal:1974,Todorcevic:1997}}]\label{fact:rosenthal}
  Let $Y$ be a Polish space and $(r_n)_{n\in\NN}$ a pointwise bounded sequence in $\Cp(Y)$.
  Then the sequence $(r_n)$ has either a pointwise convergent subsequence or a subsequence whose closure in $\RR^Y$ is homeomorphic to $\beta\NN$.
\end{fact}

\begin{fact}[Bourgain--Fremlin--Talagrand;
    see {\cite[Theorem~4G]{Bourgain-Fremlin-Talagrand:1978}}]\label{fact:bft} Let $Y$ be a Polish space and $\cR\subseteq\Cp(Y)$ a pointwise bounded family.
  Then $\overline{\cR}\subseteq\Bone(Y)$ if and only if, for every sequence $(r_n)_{n\in\NN}$ in $\cR$ and all reals $a<b$, there is $I\subseteq\NN$ such that no $y\in Y$ satisfies $r_n(y)\le a$ for all $n\in I$ and $r_n(y)\ge b$ for all $n\notin I$.
\end{fact}

When the space $Y$ is compact, the condition in \cref{fact:bft} reduces to a statement about finite patterns, and we obtain the dichotomy announced at the beginning of this section.

\begin{keytheorem}\label{thm:dichotomy}
  Let $Y$ be a compact metrizable space and $\cR\subseteq\Cp(Y)$ a pointwise bounded family.
  Then exactly one of the following holds:
  \begin{enumerate}[label=\textup{(\roman*)}]
  \item
    we have $\overline{\cR}\subseteq\Bone(Y)$;
  \item
    the family $\cR$ has the independence property.
  \end{enumerate}
\end{keytheorem}

The following proof is the classical passage from infinite to finite patterns by compactness (see~\cite[Section~1]{Duenez-Iovino-MatosWiederhold-Salvetti-Tall:2026}).

\begin{proof}
  For a sequence $(r_n)$ in $\cR$, reals $a<b$ and $I\subseteq\NN$, the sets $\{y:r_n(y)\le a\}$ for $n\in I$ and $\{y:r_n(y)\ge b\}$ for $n\notin I$ are closed, because each function $r_n$ is continuous.
  By compactness of $Y$, these closed sets have empty intersection if and only if finitely many of them do.
  Hence the condition in \cref{fact:bft} fails for a given sequence $(r_n)$ and given reals $a<b$ exactly when every finite pattern is realized, that is, exactly when $(r_n)$, $a$ and $b$ witness the independence property.
  By \cref{fact:bft}, alternative (i) therefore fails exactly when alternative (ii) holds.
\end{proof}

We shall also need to replace a space of probes by a dense subset, and the replacement does not affect the independence property.

\begin{lemma}\label{lem:nipclosure}
  Let $Y$ be a topological space, $\cR\subseteq\Cp(Y)$ a family and $L\subseteq Y$ a dense set.
  Then the family $\cR$ is NIP if and only if the family $\{r|_L:r\in\cR\}$ of restrictions is NIP.
\end{lemma}

\begin{proof}[Sketch]
  If the restrictions have the independence property, then so does $\cR$.
  Conversely, if $(r_n)$ and $a<b$ witness the independence property of $\cR$, and $a<a'<b'<b$, then each finite pattern holds strictly, with $a'$ and $b'$, on a nonempty open set, which meets $L$;
  hence the restrictions witness the independence property with the reals $a'<b'$.
\end{proof}

\begin{remark}[Symmetry]\label{rem:symmetry}
  For a formula of a first-order theory, NIP is symmetric in the two variables:
  the formula $\varphi(x,y)$ is NIP if and only if the formula $\varphi^{\mathrm{opp}}$, given by $\varphi^{\mathrm{opp}}(y,x)=\varphi(x,y)$, is NIP (see~\cite[Chapter~II, Section~4]{Shelah:1990a} and~\cite[Lemma~2.5]{Simon:2015b}).
  Assouad’s bound makes the symmetry quantitative.
  The Vapnik--Chervonenkis (VC) dimension of $\varphi$ is the largest size of a finite set of elements every subset of which is cut out by some instance
  $\varphi(x,c)$ of the formula, for a parameter $c$ that depends on the subset;
  if $\varphi$ has finite VC dimension $d$, then the formula $\varphi^{\mathrm{opp}}$ has VC dimension less than $2^{d+1}$~\cite{Assouad:1983}.
  Assouad’s argument is finite and combinatorial, and it applies, with the reals $a<b$ fixed, to real-valued functions.
  For uniform NIP, which bounds, for each gap $b-a$, the size of the finite sets of functions on which every pattern of values at most
  $a$ and at least $b$ is realized at some point (see \cref{rem:uniform} below), the roles of Alice and Bob are therefore interchangeable.
  The notion of \cref{def:ip}, which asks for an infinite sequence inside a single family, is not symmetric:
  the relation $n\in u$ between numbers $n$ and finite sets $u\subseteq\NN$ has the independence property in the variable $n$ but not in the variable $u$.
\end{remark}

%#
%### SUBSECTION: Rosenthal compacta
%#
\subsection{Rosenthal compacta}
A \emph{Rosenthal compactum} is a compact space homeomorphic to a compact subset of $\Bone(Y)$, with the topology of pointwise convergence, for some Polish space $Y$.
By \cref{thm:dichotomy} and \cref{lem:relcompact}, if $Y$ is compact metrizable and $\cR\subseteq\Cp(Y)$ is pointwise bounded and NIP, then the closure $\overline{\cR}$ is a Rosenthal compactum.
The classification of Rosenthal compacta (see \cref{sec:classification}) will therefore classify the NIP families that arise below.

%#
%### SUBSECTION: Historical remarks
%#
\subsection{Historical remarks}\label{sec:hist-prelim}
Rosenthal’s dichotomy grew out of his $\ell_1$ theorem~\cite{Rosenthal:1974};
Bourgain, Fremlin and Talagrand proved \cref{fact:bft} in~1978~\cite{Bourgain-Fremlin-Talagrand:1978}, and Godefroy named Rosenthal compacta~\cite{Godefroy:1980} (see also~\cite{Arkhangelskii:1992,Tkachuk:2011}).
The independence property was isolated by Shelah~\cite{Shelah:1971b,Shelah:1990a};
the polynomial bound behind its $\{0,1\}$-valued case,
the Sauer--Shelah lemma, was found independently by Sauer~\cite{Sauer:1972} and Shelah~\cite{Shelah:1972b},
and, a little earlier, by Vapnik and Chervonenkis~\cite{Vapnik-Chervonenkis:1971,Vapnik-Chervonenkis:1974}, in their work on the uniform convergence of frequencies to probabilities, which became a foundation of learning theory~\cite{Valiant:1984}.
The connection between NIP and Rosenthal compacta is due to Simon~\cite{Simon:2015a}, and its use for limits of computations goes back to our joint work with Matos-Wiederhold, Salvetti and Tall~\cite{Duenez-Iovino-MatosWiederhold-Salvetti-Tall:2026}.

%#
%## SECTION: Separable Rosenthal compacta fall into finitely many canonical classes
%#
\section{Separable Rosenthal compacta fall into finitely many canonical classes}\label{sec:classification}

\noindent The Baire-1 characterization reduces the analysis of NIP families to the topology of their pointwise closures.
For a countable, pointwise bounded NIP family of continuous functions on a compact metrizable space, this closure is a separable Rosenthal compactum.
Todorčević’s trichotomy distinguishes the nonmetrizable cases by hereditary separability and first countability;
the canonical form of the trichotomy identifies minimal families whose closures embed into the corresponding compacta.
We state these results together with the Argyros--Dodos--Kanellopoulos heptachotomy, following their notation~\cite{Argyros-Dodos-Kanellopoulos:2008b}, and formulate the resulting levels of NIP for use with response families.

%#
%### SUBSECTION: Three topological invariants
%#
\subsection{Three topological invariants}
Recall that a space is hereditarily separable if each of its subspaces has a countable dense subset, and first countable if each of its points has a countable base of neighborhoods.
Every compact metrizable space is hereditarily separable, and by a theorem of Pol every hereditarily separable Rosenthal compactum is first countable (see~\cite[Section~10]{Debs:2013}).
For a Rosenthal compactum $K$ we therefore have:
\[
  K\text{ metrizable}\ \Longrightarrow\ K\text{ hereditarily separable}\
  \Longrightarrow\ K\text{ first countable}.
\]
For separable Rosenthal compacta, the failure of each converse is witnessed by a canonical subspace, and these canonical subspaces are pointwise closures of families indexed by the dyadic tree, the set of finite sequences of zeros and ones
(see \cref{fact:tricho} below).

%#
%### SUBSECTION: The canonical compacta
%#
\subsection{The canonical compacta}
The minimal families encode localization, one-sided limits and paired observations by functions on the dyadic tree and its boundary.
Let $\{0,1\}^{<\NN}$ be the set of finite sequences of zeros and ones, and $\{0,1\}^\NN$ the Cantor space of infinite such sequences.
For $t\in\{0,1\}^{<\NN}$, write $|t|$ for its length and $t\sqsubset\sigma$ when $\sigma\in\{0,1\}^\NN$ extends $t$.
The function $v_t\colon\{0,1\}^\NN\to\{0,1\}$ is $1$ on the clopen set $V_t=\{\sigma\in\{0,1\}^\NN:t\sqsubset\sigma\}$ and $0$ elsewhere.
The points $t^\frown0^\infty$ and $t^\frown1^\infty$ extend $t$ by zeros and by ones, respectively.
Let $\prec$ be the lexicographic order on $\{0,1\}^\NN$, with $\preceq$ its reflexive version.
For $\sigma\in\{0,1\}^\NN$, define $x^+_\sigma,x^-_\sigma\colon\{0,1\}^\NN\to\{0,1\}$ as the characteristic functions of $\{\zeta:\sigma\preceq\zeta\}$ and $\{\zeta:\sigma\prec\zeta\}$, respectively.
For topological spaces $X_1,X_2$, with disjoint sum $X_1\oplus X_2$, and real functions $f\colon X_1\to\RR$ and $g\colon X_2\to\RR$, let $(f,g)\colon X_1\oplus X_2\to\RR$ take the value $f(x)$ at the copy of $x\in X_1$ and $g(x)$ at the copy of $x\in X_2$.
Since a point of $\{0,1\}^\NN$ is a function from $\NN$ to $\{0,1\}$, the Cantor space is a subset of $\RR^\NN$.
The compacta below keep the names under which Argyros, Dodos and Kanellopoulos introduced them, such as $A(2^\NN)$, $\hat{D}(2^\NN)$ and $A(2^{<\NN})$, in which $2^\NN$ stands for the Cantor space and $2^{<\NN}$ for the set of finite sequences.

The following compacta appear in the classification:
\begin{itemize}[leftmargin=*]
\item
  the \emph{one-point compactification} $A(2^\NN)$ of a discrete set of size continuum, realized as
  $\{\delta_\sigma:\sigma\in\{0,1\}^\NN\}\cup\{0\}$, where
  $\delta_\sigma$ is the characteristic function of $\{\sigma\}$, and its separable extension $\hat{A}(2^\NN)$, the pointwise closure of $\{v_t:t\in\{0,1\}^{<\NN}\}$, which equals $\{v_t:t\in\{0,1\}^{<\NN}\}\cup A(2^\NN)$;
\item
  the \emph{split Cantor set} $S(2^\NN)=\{x^+_\sigma,x^-_\sigma:\sigma\in\{0,1\}^\NN\}$, the \emph{split interval}
  $S(I)$, that is, $I\times\{0,1\}$ for the unit interval $I=[0,1]$, ordered lexicographically, with the order topology, and the extended split Cantor sets $\hat{S}_+(2^\NN)$ and $\hat{S}_-(2^\NN)$ defined below;
  the space $\hat{S}_+(2^\NN)$ is homeomorphic to a closed subspace of $S(I)$;
\item
  the \emph{Alexandroff duplicate} $D(2^\NN)$ of the Cantor set and its separable extension $\hat{D}(2^\NN)$, and the extended duplicate $\hat{D}(S(2^\NN))$ of the split Cantor set, the last two being defined below;
  the duplicate $D(2^\NN)$ is the set $\{0,1\}^\NN\times\{0,1\}$
  in which the points $(\sigma,1)$ are isolated and the basic neighborhoods of a point $(\sigma,0)$ are the sets $(U\times\{0,1\})\setminus\{(\sigma,1)\}$, for open neighborhoods $U$ of $\sigma$ in $\{0,1\}^\NN$.
\end{itemize}
Among these, the space $S(2^\NN)$ is hereditarily separable but not metrizable, the space $D(2^\NN)$ is first countable but not separable, and the space $A(2^\NN)$ is not first countable (see~\cite{Argyros-Dodos-Kanellopoulos:2008b}).

%#
%### SUBSECTION: The heptachotomy
%#
\subsection{The heptachotomy}
A \emph{dyadic subtree} of $\{0,1\}^{<\NN}$ is a subset of $\{0,1\}^{<\NN}$ that, ordered by extension, is isomorphic to $\{0,1\}^{<\NN}$;
we write it as $(s_t)_{t\in\{0,1\}^{<\NN}}$, using its canonical enumeration, the unique such isomorphism $t\mapsto s_t$ that also preserves the lexicographic order (see~\cite[Section~2]{Argyros-Dodos-Kanellopoulos:2008b}).
A dyadic subtree is \emph{regular} if each of its levels is contained in a level of $\{0,1\}^{<\NN}$.
Fix a regular dyadic subtree $\Theta=(s_t)_{t\in\{0,1\}^{<\NN}}$ such that the sequences $s^\frown0^\infty$, for $s\in\Theta$, are pairwise distinct, and so are the sequences $s^\frown1^\infty$;
a family indexed by the nodes $s$ of $\Theta$ is regarded as indexed by $\{0,1\}^{<\NN}$ through the canonical enumeration $t\mapsto s_t$.

\begin{definition}\label{def:equivalent}
  Let $X$ be a Polish space.
  \begin{enumerate}[label=\textup{(\arabic*)}]
  \item
    For topological spaces $X_1,X_2$ and a countable index set $J$, let $f_j\colon X_1\to\RR$ and $g_j\colon X_2\to\RR$ be functions for $j\in J$, with each family pointwise bounded.
    The families $\{f_j:j\in J\}$ and $\{g_j:j\in J\}$ are \emph{equivalent} if the assignment $f_j\mapsto g_j$ is well-defined and extends to a homeomorphism from the closure of $\{f_j:j\in J\}$ in $\RR^{X_1}$ onto the closure of $\{g_j:j\in J\}$ in $\RR^{X_2}$, both with the product topology.
  \item
    A family of pairwise distinct functions $\{f_t:t\in\{0,1\}^{<\NN}\}$ that is relatively compact in $\Bone(X)$ is \emph{minimal} if it is equivalent to $\{f_{s_t}:t\in\{0,1\}^{<\NN}\}$ for every dyadic subtree $(s_t)$ of $\{0,1\}^{<\NN}$.
  \end{enumerate}
\end{definition}

The seven minimal families are the following, and the pointwise closure of each of them is a separable Rosenthal compactum that contains the family as a discrete set (see~\cite{Argyros-Dodos-Kanellopoulos:2008b}):
\begin{enumerate}[label=\textup{(\arabic*)},leftmargin=*]
\item
  $D_1=\{\frac1{|t|+1}v_t:t\in\{0,1\}^{<\NN}\}$, whose closure $A(2^{<\NN})$ is the one-point compactification of a countable discrete set;
\item
  $D_2=\{s^\frown0^\infty:s\in\Theta\}$, a family of points of $\{0,1\}^\NN\subseteq\RR^\NN$, whose closure is denoted $2^{\le\NN}$;
\item
  $D_3=\{x^+_{s^\frown0^\infty}:s\in\Theta\}$, whose closure is $\hat{S}_+(2^\NN)$;
\item
  $D_4=\{x^-_{s^\frown1^\infty}:s\in\Theta\}$, whose closure is $\hat{S}_-(2^\NN)$;
\item
  $D_5=\{v_t:t\in\{0,1\}^{<\NN}\}$, whose closure is $\hat{A}(2^\NN)$;
\item
  $D_6=\{(v_t,t^\frown0^\infty):t\in\{0,1\}^{<\NN}\}$, whose closure is $\hat{D}(2^\NN)$;
\item
  $D_7=\{(v_s,x^+_{s^\frown0^\infty}):s\in\Theta\}$, whose closure is $\hat{D}(S(2^\NN))$.
\end{enumerate}
The closures in (1) and (2) are metrizable, those in
(3) and (4) are hereditarily separable but not metrizable, those in (6) and (7) are first countable but not hereditarily separable, and the closure in (5) is not first countable (see~\cite{Argyros-Dodos-Kanellopoulos:2008b}).

\begin{fact}[Heptachotomy;
    see {\cite[Theorem~2]{Argyros-Dodos-Kanellopoulos:2008b}}]\label{fact:hepta} \leavevmode
  \begin{enumerate}[label=\textup{(\alph*)}]
  \item
    Up to equivalence, there are exactly seven minimal families, namely $D_1,\dots,D_7$, and they are pairwise non-equivalent.
  \item
    Let $X$ be a Polish space and $\{f_t:t\in\{0,1\}^{<\NN}\}$ a family of pairwise distinct functions that is relatively compact in $\Bone(X)$.
    Then there is a regular dyadic subtree $(s_t)$ of $\{0,1\}^{<\NN}$ such that $\{f_{s_t}:t\in\{0,1\}^{<\NN}\}$ is equivalent to one of $D_1,\dots,D_7$.
  \end{enumerate}
\end{fact}

%#
%### SUBSECTION: Todorčević’s trichotomy
%#
\subsection{Todorčević’s trichotomy}
Let $X$ be a Polish space, $K\subseteq\Bone(X)$ a separable Rosenthal compactum and $\{f_n:n\in\NN\}$ a countable dense subset of $K$.
An injection $\iota\colon\{0,1\}^{<\NN}\to\NN$ is \emph{increasing} if $\iota(s)<\iota(t)$ whenever $s$ precedes $t$ in the order of $\{0,1\}^{<\NN}$ by length and, within each length, by the lexicographic order.
A minimal family $D_i$ \emph{canonically embeds} into
$K$ (with respect to $\{f_n\}$) if there is an increasing injection $\iota\colon\{0,1\}^{<\NN}\to\NN$ such that $\{f_{\iota(t)}:t\in\{0,1\}^{<\NN}\}$ is equivalent to $D_i$.
In that case the closure of $D_i$ is homeomorphic to a closed subspace of $K$.

\begin{fact}[Todorčević’s trichotomy in canonical form;
    see {\cite{Todorcevic:1999}} and {\cite[Theorem~3]{Argyros-Dodos-Kanellopoulos:2008b}}]\label{fact:tricho} Let $X$ be a Polish space, $K\subseteq\Bone(X)$ a separable Rosenthal compactum and $\{f_n:n\in\NN\}$ a countable dense subset of $K$.
  \begin{enumerate}[label=\textup{(\alph*)}]
  \item
    If $K$ consists of bounded functions, is metrizable and is not separable in the supremum norm, then the family $D_2$ canonically embeds into $K$.
  \item
    If $K$ is hereditarily separable but not metrizable, then the family $D_3$ or the family $D_4$ canonically embeds into $K$.
  \item
    If $K$ is first countable but not hereditarily separable, then the family $D_6$ or the family $D_7$ canonically embeds into $K$.
  \item
    If $K$ is not first countable, then the family $D_5$ canonically embeds into $K$.
  \end{enumerate}
\end{fact}

In particular, a non-metrizable separable Rosenthal compactum contains a copy of
$\hat{S}_+(2^\NN)$ or $\hat{S}_-(2^\NN)$, of $\hat{D}(2^\NN)$ or $\hat{D}(S(2^\NN))$, or of $\hat{A}(2^\NN)$, according as it is hereditarily separable, first countable but not hereditarily separable, or not first countable.
If $K$ is metrizable and infinite, then the family $D_1$ canonically embeds into $K$, as recorded in \cref{fig:tree}.
Indeed, an infinite compact metrizable space has a non-isolated point $g$.
Since the enumerated family is dense, we can choose indices $n_0<n_1<\cdots$ such that the functions $f_{n_k}$ are pairwise distinct, differ from $g$ and converge to $g$.
Enumerate $\{0,1\}^{<\NN}$ as $t_0,t_1,\dots$ by length and, within each length, lexicographically, and let $\iota(t_k)=n_k$, which defines an increasing injection.
The closure of $\{f_{\iota(t)}:t\in\{0,1\}^{<\NN}\}$ and the closure $D_1\cup\{0\}$ of $D_1$ each consist of a sequence of pairwise distinct isolated points together with its limit, and every bijection between the isolated points of two such spaces extends, by sending limit to limit, to a homeomorphism.
Hence the family $\{f_{\iota(t)}:t\in\{0,1\}^{<\NN}\}$ is equivalent to $D_1$.

\begin{figure}[ht]
  \centering
  \makebox[\textwidth][c]{\input{tree_body.tex}}
  \caption{Todorčević’s trichotomy as a decision tree, with the canonical compacta and the minimal families of the heptachotomy at its leaves (compare \cite[Section~4.2]{Duenez-Iovino-MatosWiederhold-Salvetti-Tall:2026}), where the dashed leaf is absent for full observable compacta (see \cref{thm:convex} and \cref{cor:nip2}). The NIP level at each class is the largest level it reaches.}
  \label{fig:tree}
\end{figure}

%#
%### SUBSECTION: Levels of NIP
%#
\subsection{Levels of NIP}
We now transfer the classification to NIP families, as in our work on deep computations~\cite[Section~4.2]{Duenez-Iovino-MatosWiederhold-Salvetti-Tall:2026}.

\begin{definition}\label{def:niplevels}
  Let $Y$ be a compact metrizable space and $\cR\subseteq\Cp(Y)$ a pointwise bounded NIP family whose closure $\overline\cR$ is separable.
  The family $\cR$ is \emph{NIP$^{\mathrm{fc}}$}, \emph{NIP$^{\mathrm{hs}}$} or \emph{NIP$^{\mathrm{met}}$} if $\overline\cR$ is, respectively, first countable, hereditarily separable or metrizable.
\end{definition}

We write the levels with superscripts, rather than as NIP$_1$, NIP$_2$, NIP$_3$ as in~\cite[Section~4.2]{Duenez-Iovino-MatosWiederhold-Salvetti-Tall:2026}, because in model theory NIP$_k$ denotes Shelah’s $k$-dependence.

By the implications among the three invariants, we have NIP$^{\mathrm{met}}\Rightarrow$NIP$^{\mathrm{hs}}\Rightarrow$NIP$^{\mathrm{fc}}\Rightarrow$NIP.

The top level has a concrete meaning.
A set $Z\subseteq Y$ \emph{separates} a set $K$ of functions on $Y$ if any two distinct members of $K$ differ at some point of $Z$.

\begin{proposition}\label{prop:countable-probes}
  Let $Y$ be a set and $K\subseteq\RR^Y$ a set that is compact in the topology of pointwise convergence.
  Then the space $K$ is metrizable if and only if some countable set $Z\subseteq Y$ separates $K$.
\end{proposition}

\begin{proof}
  Suppose that a countable set $Z$ separates $K$.
  Restriction to $Z$ is a continuous injection of $K$ into $\RR^Z$, which is metrizable because $Z$ is countable.
  A continuous injection of a compact space into a Hausdorff space is a homeomorphism onto its image, so the space $K$ is metrizable.

  Conversely, suppose that $K$ is metrizable.
  Being compact, it is second countable, and so is $K\times K$.
  Hence the open set $W=\{(f,g)\in K\times K:f\ne g\}$ is Lindelöf.
  For $y\in Y$, the set $U_y=\{(f,g)\in K\times K:f(y)\ne g(y)\}$ is open, and these sets cover $W$.
  If $\{U_y:y\in Z\}$ is a countable subcover, then the set $Z$ separates $K$.
\end{proof}

Thus an NIP family with separable closure is NIP$^{\mathrm{met}}$ exactly when countably many points of $Y$ distinguish all functions in its closure.
When the functions are responses and the points are probes, this says that countably many probes distinguish all idealized responses.
The other two levels have readings of the same kind, for which we need two notions.
Let $K$ be a set of real functions on a set $Y$.
A member $f$ of $K$ is \emph{countably determined} if there is a countable set $Z_f\subseteq Y$ such that $f$ is the only member of $K$ that agrees with $f$ on $Z_f$.
A set $E\subseteq K$ is \emph{finitely separated} if, for each $f\in E$, there are a finite set $F_f\subseteq Y$ and a real number $\varepsilon_f>0$ such that every $g\in E\setminus\{f\}$ satisfies $|g(y)-f(y)|\ge\varepsilon_f$ for some $y\in F_f$.

\begin{proposition}\label{prop:levels-separation}
  Let $Y$ be a compact metrizable space and $\cR\subseteq\Cp(Y)$ a pointwise bounded NIP family whose closure $K=\overline\cR$ is separable.
  \begin{enumerate}[label=\textup{(\alph*)}]
  \item
    The family is NIP$^{\mathrm{met}}$ if and only if some countable set $Z\subseteq Y$ separates $K$.
  \item
    The family is NIP$^{\mathrm{fc}}$ if and only if every member of $K$ is countably determined.
  \item
    The family is NIP$^{\mathrm{hs}}$ if and only if no subset of $K$ of cardinality $2^{\aleph_0}$ is finitely separated.
  \end{enumerate}
\end{proposition}

\begin{proof}
  Part (a) is \cref{prop:countable-probes}.
  For (b), fix $f\in K$ and choose a countable determining set $Z_f\subseteq Y$ when one exists.
  The singleton $\{f\}$ is the intersection of the open sets $\{g\in K:|g(y)-f(y)|<1/m\}$, indexed by $y\in Z_f$ and positive integers $m$.
  In a compact Hausdorff space, a point whose singleton is a countable intersection of open sets has a countable neighborhood base.
  Conversely, let $(U_k)_{k\in\NN}$ be a countable neighborhood base at $f$.
  For each $k\in\NN$, choose a basic product neighborhood of $f$ contained in $U_k$, defined by a finite set $F_k\subseteq Y$ of coordinates.
  The set $Z_f=\bigcup_{k\in\NN}F_k$ determines $f$, since any member of $K$ agreeing with $f$ on these coordinates belongs to every $U_k$ and hence equals $f$.
  (c) A set is finitely separated exactly when it is a discrete subspace, and a separable Rosenthal compactum is hereditarily separable if and only if it has no discrete subspace of cardinality
  $2^{\aleph_0}$, by Pol’s theorem, which also follows from Todorčević’s trichotomy (see~\cite{Argyros-Dodos-Kanellopoulos:2008b} and~\cite[Theorem~0.4]{Glasner-Megrelishvili:2022}).
\end{proof}

The three levels thus measure how many points of $Y$ are needed to determine a limit in the closure $K$:
one countable set for all limits (NIP$^{\mathrm{met}}$), a countable set for each limit (NIP$^{\mathrm{fc}}$), or no family of size continuum in which each member can be isolated by finitely many points at a positive tolerance (NIP$^{\mathrm{hs}}$).
In the last condition, both the finite set and the tolerance may depend on the member being isolated.
When limits are taken as types for a fixed formula and points as parameters, the levels describe how instances of the formula determine a type.
This partly answers a question, raised in our work on deep computations~\cite{Duenez-Iovino-MatosWiederhold-Salvetti-Tall:2026}, that asks for a non-topological characterization of the levels;
a finitary characterization, in the spirit of the fat-shattering dimension, remains open (see \cref{ques:finitary-levels}).
The fat-shattering dimension at a given gap is the supremum of the sizes of the finite sets of points at which the functions realize every pattern of high and low values, around thresholds that may vary from point to point.

\begin{keytheorem}\label{thm:nipclasses}
  Let $Y$ be a compact metrizable space and $\{r_n:n\in\NN\}\subseteq\Cp(Y)$ a pointwise bounded NIP family.
  Then exactly one of the following holds:
  \begin{enumerate}[label=\textup{(\roman*)}]
  \item
    the family is NIP$^{\mathrm{met}}$;
  \item
    the family is NIP$^{\mathrm{hs}}$ but not NIP$^{\mathrm{met}}$, and there is an increasing injection $\iota\colon\{0,1\}^{<\NN}\to\NN$ such that $\{r_{\iota(t)}:t\in\{0,1\}^{<\NN}\}$ is equivalent to $D_3$ or to $D_4$;
  \item
    the family is NIP$^{\mathrm{fc}}$ but not NIP$^{\mathrm{hs}}$, and there is an increasing injection $\iota\colon\{0,1\}^{<\NN}\to\NN$ such that $\{r_{\iota(t)}:t\in\{0,1\}^{<\NN}\}$ is equivalent to $D_6$ or to $D_7$;
  \item
    the family is not NIP$^{\mathrm{fc}}$, and there is an increasing injection $\iota\colon\{0,1\}^{<\NN}\to\NN$ such that $\{r_{\iota(t)}:t\in\{0,1\}^{<\NN}\}$ is equivalent to $D_5$.
  \end{enumerate}
\end{keytheorem}

The following proof merely combines the facts above.

\begin{proof}
  By \cref{thm:dichotomy} and \cref{lem:relcompact}, the closure $K$ of the family is a Rosenthal compactum, and the family is a countable dense subset of it.
  The four cases are mutually exclusive by the definitions and cover all possibilities.
  In cases (ii)--(iv), the required injection is provided by \cref{fact:tricho}(b)--(d), applied to $K$ and the dense subset $\{r_n\}$.
\end{proof}

Thus a non-metrizable NIP closure is always explained by a subfamily of the original functions, arranged along the dyadic tree and equivalent to one of five minimal families.
Applied to realized types, the complete descriptions of elements of the structure in terms of the parameters, the classification of \cref{thm:nipclasses} yields that of NIP formulas recalled in \cref{sec:main}
(see~\cite[Section~4.2]{Duenez-Iovino-MatosWiederhold-Salvetti-Tall:2026}).
There the realized types are viewed as continuous functions on a compact metrizable space of parameters, and this proviso matters:
the space of types of a formula over a countable set of parameters, with no topology on the parameters, is always metrizable, since the countable set of parameters separates the types (see \cref{prop:countable-probes}).

%#
%### SUBSECTION: Historical remarks
%#
\subsection{Historical remarks}
Todorčević proved his trichotomy in~1999~\cite{Todorcevic:1999}, using forcing and absoluteness in the step for points without a countable neighborhood base;
Argyros, Dodos and Kanellopoulos extracted the canonical form and the heptachotomy from the methods of his paper in~2008~\cite{Argyros-Dodos-Kanellopoulos:2008b},
using infinite-dimensional Ramsey theory for trees and perfect sets~\cite{Halpern-Lauchli:1966,Milliken:1981,Louveau-Shelah-Velickovic:1993}.
Glasner and Megrelishvili used the trichotomy to build a hierarchy of tame dynamical systems~\cite{Glasner-Megrelishvili:2022}, and the levels NIP$^{\mathrm{fc}}$, NIP$^{\mathrm{hs}}$
and NIP$^{\mathrm{met}}$ were introduced for limits of computations in our work on deep computations~\cite{Duenez-Iovino-MatosWiederhold-Salvetti-Tall:2026}.

In both settings, Todorčević’s trichotomy is applied to families of functions without linear structure, namely enveloping semigroups and sets of computations, and all its classes can occur.
The idealized responses of an observable (see \cref{sec:responses}), by contrast, form a convex set (see the proof of \cref{cor:nip2}), and convexity collapses NIP$^{\mathrm{hs}}$ onto NIP$^{\mathrm{met}}$ (see \cref{thm:convex} and \cref{rem:convex-literature});
the shorter hierarchy is a feature of quantum, and more generally of probabilistic, response spaces, although all seven minimal families still occur along the dyadic tree (see \cref{thm:hepta-natural}).

%#
%## SECTION: The idealized responses of an observable form a compact set of functions
%#
\section{The idealized responses of an observable form a compact set of functions}\label{sec:responses}

\noindent To apply the classification to a bipartite observable, we retain from Bob’s operator algebra the effective observables produced by Alice’s preparations.
The states of the smallest unital C$^*$-algebra containing them form a compact metrizable probe space, on which each actual preparation has a continuous response.
We identify the pointwise closure of these responses with the responses obtained from all states of Alice’s operator algebra, with Bob’s system idealized first.
The compactum to be classified thus records preparations through their responses to the fixed observable.

%#
%### SUBSECTION: Observables and responses
%#
\subsection{Observables and responses}
We work in the setting of Part~I~\cite{Duenez-Iovino:2026a} and recall its operator theory briefly.
Alice’s and Bob’s systems are described by separable complex Hilbert spaces $H_1$ and $H_2$, whose inner products are linear in the second slot.
For a Hilbert space $H$, we write $B(H)$ for the algebra of bounded operators on $H$;
the trace-class operators, those $x\in B(H)$ with $\Tr|x|<\infty$, where $|x|=(x^*x)^{1/2}$, form a Banach space $\cT(H)$ under the trace norm $\|x\|_1=\Tr|x|$.
A \emph{density operator} on $H$ is a positive trace-class operator of trace~$1$;
density operators are the \emph{actual states} of the system, and we write $\fD(H)$ for the set they form.
A bounded self-adjoint operator $C$ on $H_1\otimes H_2$ is an \emph{observable}, and its expectation when the two systems are prepared independently in actual states $\rho\in\fD(H_1)$ and $\sigma\in\fD(H_2)$ is
\[
  \varphi_C(\rho,\sigma)=\Tr\bigl((\rho\otimes\sigma)C\bigr).
\]
For fixed $C$ and $\rho\in\fD(H_1)$, Bob’s preparation is the probe variable: the response on actual probes is $r_\rho\colon\fD(H_2)\to[-\|C\|,\|C\|]$, defined by $r_\rho(\sigma)=\varphi_C(\rho,\sigma)$.
Fixing Alice’s state $\rho$ produces an observable on Bob’s system alone, the \emph{effective observable} $T_C(\rho)=\Tr_1((\rho\otimes I)C)$, obtained by tracing out Alice’s system.
It is characterized by $\varphi_C(\rho,\sigma)=\Tr(\sigma\,T_C(\rho))$ for all $\sigma\in\fD(H_2)$.
The map $\rho\mapsto T_C(\rho)$ extends linearly to a bounded operator $T_C\colon\cT(H_1)\to B(H_2)$ with $\|T_C\|\le\|C\|$, the \emph{slice map}.

Idealized states are described through C$^*$-algebras.
A \emph{C$^*$-algebra} is a norm-closed algebra of bounded operators closed under adjoints;
we use only unital ones, which contain the identity.
A \emph{state} of a unital C$^*$-algebra $\fA$ is a linear functional $\omega$ on $\fA$ that is positive, in the sense that $\omega(a^*a)\ge0$ for all $a\in\fA$, and satisfies $\omega(I)=1$.
The set $\cS(\fA)$ of states, with the weak$^*$ topology (pointwise convergence on $\fA$), is compact;
it is metrizable when $\fA$ is norm-separable.
Every density operator $\rho$ defines a state $a\mapsto\Tr(\rho a)$;
the remaining states of $B(H)$, which no density operator represents, are the \emph{idealized} states, and each of them is the weak$^*$ limit of some net of actual states (see~\cite[Lemma~2.2]{Duenez-Iovino:2026a}).

\subsection{The space of probes}
The responses of Alice’s preparations involve Bob’s system only through the effective observables $T_C(\rho)$.
The natural probes are therefore the states of the smallest unital C$^*$-algebra containing them.

\begin{definition}\label{def:probe-space}
  Let $C$ be an observable on $H_1\otimes H_2$.
  The \emph{probe algebra} $\fA_C\subseteq B(H_2)$ is the unital C$^*$-algebra generated by $T_C(\cT(H_1))$, and the \emph{probe space} is $Y_C=\cS(\fA_C)$.
  The \emph{record} of an actual probe $\sigma\in\fD(H_2)$ is the functional $\mathfrak{p}_\sigma\colon\fA_C\to\CC$ given by $\mathfrak{p}_\sigma(a)=\Tr(\sigma a)$.
  For $\rho\in\fD(H_1)$, the \emph{completed response} of $\rho$ is the function $\mathfrak{r}_\rho\colon Y_C\to[-\|C\|,\|C\|]$ given by $\mathfrak{r}_\rho(\omega_2)=\omega_2(T_C(\rho))$, and the \emph{response family} is $\cR_C=\{\mathfrak{r}_\rho:\rho\in\fD(H_1)\}$.
\end{definition}

The record $\mathfrak{p}_\sigma$ is a state of $\fA_C$, and hence a point of $Y_C$, rather than a density operator;
different density operators may have the same record.
In a completed response the preparation $\rho$ remains actual: only the domain of probes is completed.

\begin{lemma}\label{lem:probe-space}
  Let $C$ be an observable on $H_1\otimes H_2$.
  \begin{enumerate}[label=\textup{(\alph*)}]
  \item
    The probe space $Y_C$ is compact and metrizable.
  \item
    For every $\rho\in\fD(H_1)$, the completed response $\mathfrak{r}_\rho$ is continuous on $Y_C$ and bounded by $\|C\|$, so the response family $\cR_C$ is a pointwise bounded subset of $\Cp(Y_C)$.
  \item
    The records $\mathfrak{p}_\sigma$ of actual probes $\sigma\in\fD(H_2)$ form a dense subset of $Y_C$, and for $\rho\in\fD(H_1)$ and $\sigma\in\fD(H_2)$ we have $\mathfrak{r}_\rho(\mathfrak{p}_\sigma)=r_\rho(\sigma)=\varphi_C(\rho,\sigma)$.
  \end{enumerate}
\end{lemma}

\begin{proof}
  (a) Because $H_1$ is separable, so is $\cT(H_1)$;
  its image under the norm-continuous slice map is therefore separable, and so is the C$^*$-algebra $\fA_C$ that this image generates.
  The compact space $Y_C$ is thus metrizable.

  (b) The response $\mathfrak{r}_\rho$ is evaluation at the element $T_C(\rho)$ of $\fA_C$, and evaluation at a fixed element is weak$^*$ continuous.
  Since $\rho\ge0$ and $C$ is self-adjoint, the operator $T_C(\rho)$ is self-adjoint, whence $\mathfrak{r}_\rho$ is real-valued;
  moreover, $|\mathfrak{r}_\rho(\omega_2)|\le\|T_C(\rho)\|\le\|C\|$.

  (c) By the Hahn--Banach theorem, every state of $\fA_C$ extends to a state of $B(H_2)$
  (see \cite[Theorem~4.3.13]{Kadison-Ringrose:1983a});
  restriction to $\fA_C$ is thus a weak$^*$-continuous map of the state space of $B(H_2)$ onto $Y_C$.
  The states given by density operators are weak$^*$ dense in the state space of $B(H_2)$ (see \cite[Lemma~2.2]{Duenez-Iovino:2026a}), and their restrictions are therefore dense in $Y_C$.
  The final identity restates the characterization of the effective observable.
\end{proof}

By \cref{lem:probe-space}(c) and \cref{lem:nipclosure}, Bob’s actual probes alone decide whether $\cR_C$ has the independence property.
The family $\cR_C$ has it exactly when there are actual states $\rho_1,\rho_2,\ldots$ of Alice’s system and reals $a<b$ with the following property:
for all disjoint finite $E,F\subseteq\NN$, some actual state $\sigma$ of Bob’s system satisfies $\varphi_C(\rho_n,\sigma)\le a$ for $n\in E$ and $\varphi_C(\rho_n,\sigma)\ge b$ for $n\in F$.
A single probe then lets Bob read any finite pattern off Alice’s sequence.

\begin{definition}\label{def:nip-observable}
  An observable $C$ on $H_1\otimes H_2$ is \emph{NIP for Alice’s responses} if the response family $\cR_C$ is NIP.
\end{definition}

%#
%### SUBSECTION: Idealized responses
%#
\subsection{Idealized responses}
To describe the pointwise closure of $\cR_C$, we extend the slice map from trace-class operators to functionals on $B(H_1)$.
For a bounded operator $T\colon V\to W$ between Banach spaces, its transpose is $T'\colon W^*\to V^*$, given by $T'(\lambda)=\lambda\circ T$, and its bitranspose is $T''\colon V^{**}\to W^{**}$.
Under the canonical embeddings of each Banach space into its bidual, the bitranspose extends $T$ and is weak$^*$-to-weak$^*$ continuous.
We regard the slice map as a map into $\fA_C$, so that the relevant maps are $T_C'\colon\fA_C^*\to B(H_1)$ and $T_C''\colon B(H_1)^*\to\fA_C^{**}$, where $\cT(H_1)^*$ is identified with $B(H_1)$.
The canonical image of a density operator $\rho$ in $B(H_1)^*$ is the functional $a\mapsto\Tr(\rho a)$; the same domain also contains every idealized state of $B(H_1)$.

For a fixed observable $C$ on $H_1\otimes H_2$ and a state $\omega_1\in\cS(B(H_1))$, define $\mathfrak{r}_{\omega_1}\colon Y_C\to[-\|C\|,\|C\|]$ by
\[
  \mathfrak{r}_{\omega_1}(\omega_2)=\langle T_C''(\omega_1),\omega_2\rangle
  =\omega_1(T_C'\omega_2).
\]
For each $\omega_2\in Y_C$, the operator $T_C'\omega_2$ is self-adjoint with norm at most $\|C\|$, so the displayed value is real and satisfies the stated bound.
The response depends affinely on each state separately, and when $\omega_1$ is the state represented by a density operator $\rho$ it agrees with $\mathfrak{r}_\rho$ of \cref{def:probe-space}.
We recall from Part~I~\cite{Duenez-Iovino:2026a} the two products of $\omega_1\in\cS(B(H_1))$ and $\widetilde{\omega}_2\in\cS(B(H_2))$, evaluated at $C$.
For nets $(\rho_i)_{i\in I}$ and $(\sigma_j)_{j\in J}$ of actual states, indexed by directed sets $I$ and $J$ and converging weak$^*$ to $\omega_1$ and $\widetilde{\omega}_2$, respectively, the products satisfy
\[
  (\omega_1\oA\widetilde{\omega}_2)(C)=\lim_j\lim_i\varphi_C(\rho_i,\sigma_j),
  \qquad
  (\omega_1\oB\widetilde{\omega}_2)(C)=\lim_i\lim_j\varphi_C(\rho_i,\sigma_j).
\]
The subscript names the system idealized first, that is, the one whose limit is taken inside.

\begin{lemma}\label{lem:closure}
  Let $C$ be an observable on $H_1\otimes H_2$.
  Then
  \[
    \overline{\cR_C}=\{\mathfrak{r}_{\omega_1}:\omega_1\in\cS(B(H_1))\}.
  \]
  Moreover, if $\omega_1$ is a state of $B(H_1)$, $\omega_2\in Y_C$ and $\widetilde{\omega}_2$ is a state of $B(H_2)$ extending $\omega_2$, then $\mathfrak{r}_{\omega_1}(\omega_2)=(\omega_1\oB\widetilde{\omega}_2)(C)$.
\end{lemma}

\begin{proof}
For each probe $\omega_2\in Y_C$, the coordinate $\omega_1\mapsto \mathfrak{r}_{\omega_1}(\omega_2)=\omega_1(T_C'\omega_2)$ is weak$^*$ continuous on $\cS(B(H_1))$.
  Hence the response map $\omega_1\mapsto \mathfrak{r}_{\omega_1}$ is continuous into $\RR^{Y_C}$ with its pointwise topology, and its image $K$ is compact.
  Since the bitranspose extends the slice map, the image contains $\cR_C$ and therefore contains its pointwise closure.
  Conversely, actual states are weak$^*$ dense in $\cS(B(H_1))$ (see \cite[Lemma~2.2]{Duenez-Iovino:2026a}), so continuity gives $K\subseteq\overline{\cR_C}$.
  For the final identity, if $\widetilde{\omega}_2\in\cS(B(H_2))$ restricts to $\omega_2\in Y_C$, then
  \[
    \mathfrak{r}_{\omega_1}(\omega_2)=\omega_1(T_C'\omega_2)
    =(\omega_1\oB\widetilde{\omega}_2)(C),
  \]
  by the product formula of \cite[Lemma~2.4]{Duenez-Iovino:2026a}: the functional $x\mapsto\widetilde{\omega}_2(T_Cx)$ on $\cT(H_1)$ equals $x\mapsto \omega_2(T_Cx)$.
  Thus the value is independent of the extension $\widetilde{\omega}_2$.
\end{proof}

We call the members of $\overline{\cR_C}$ the \emph{idealized responses};
they include the completed responses of actual preparations.
By \cref{lem:closure}, the response map $\omega_1\mapsto \mathfrak{r}_{\omega_1}$ sends $\cS(B(H_1))$ continuously onto this compactum; different states may have the same response.
A probe in $Y_C$, the restriction of a state of $B(H_2)$, need not be represented by a density operator, and an idealized response evaluates it with Bob’s system idealized first, whichever extension of the probe is chosen.
On actual probes the two orders agree (see \cite[Corollary~2.5]{Duenez-Iovino:2026a}): for every state $\omega_1$ of $B(H_1)$ and every $\sigma\in\fD(H_2)$, we have $\mathfrak{r}_{\omega_1}(\mathfrak{p}_\sigma)=(\omega_1\oB\sigma)(C)=(\omega_1\oA\sigma)(C)$.

%#
%### SUBSECTION: Real-valued predicates and bounded confinements
%#
\subsection{Real-valued predicates and bounded confinements}\label{sec:logical-framework}
The operator-algebraic setting determines the real-valued predicates to which we apply the independence property.
For expectation predicates, the fixed positive gap $b-a$ in \cref{def:ip} separates responses at a prescribed resolution, while affine dependence on the state records probabilistic mixing.
We work in Keisler’s model theory of real-valued structures~\cite{Keisler:2023},
in which no metric is given a priori and each formula carries its own bound,
rather than in the continuous logic of Ben Yaacov and Usvyatsov~\cite{BenYaacov-Usvyatsov:2010}, in which each sort, that is, each kind of object, carries a metric given in advance and all formulas take values in $[0,1]$.
For the bounded observables of this paper either setting would serve;
the distinction matters for unbounded quantities.

Unbounded quantities are unavoidable in quantum mechanics:
the position $Q$ and the momentum $P$ of a particle satisfy the canonical commutation relation
$QP-PQ=i\hbar I$ on a common dense invariant domain, where $\hbar$ is the reduced Planck constant and $I$ is the identity operator, and no pair of bounded operators satisfies this relation on the whole Hilbert space~\cite{Wintner:1947,Wielandt:1949}.
For comparisons, unboundedness is harmless, because $[A<B]$ is a projection, unchanged when $A$ and $B$ are replaced by the same strictly increasing function of themselves (see \cref{prop:cv-helly}).
For expectation values it is not:
the energy $\rho\mapsto\Tr(\rho H)$, where the Hamiltonian $H$ of the system is an unbounded self-adjoint operator and the state $\rho$ ranges over the density operators for which this trace is defined, has no bounded range on its full domain, and so cannot be normalized to $[0,1]$ by an invertible affine change of scale.
A nonlinear rescaling would destroy the affine dependence on mixtures;
another approach adds a point at infinity~\cite{Yaacov:2008}.
In Keisler’s setting each formula $\varphi$ takes its values in an interval $[-r_\varphi,r_\varphi]$, whose radius $r_\varphi$ depends on the formula, and an unbounded quantity enters through its restrictions to confinements, sets of states on each of which its expectation is bounded:
for example, states supported in a bounded spectral interval of the quantity or, for the energy when $H\ge0$, states satisfying the constraint $\Tr(\rho H)\le E$ for a fixed energy $E\ge0$, in the units of $H$.
The confinements play the role of the shards of our joint work with Matos-Wiederhold, Salvetti and Tall~\cite{Duenez-Iovino-MatosWiederhold-Salvetti-Tall:2026}, the parts of a state space on which every predicate stays within prescribed bounds.

Idealized states, and with them the instabilities studied in this series, can arise only from confinements that are large in the trace norm, the norm that measures how well two states can be distinguished.
Indeed, a set of density operators that is relatively compact in the trace norm has no idealized state in its weak$^*$ closure, since its trace-norm closure is weak$^*$ compact.
Moreover, when Bob’s states are restricted to such a set, the expectation of every bounded observable is stable on the restricted domain, because every sequence of Bob’s states then has a subsequence that converges in the trace norm, and the expectation is Lipschitz in Bob’s state, uniformly in Alice’s:
we have $|\varphi_C(\rho,\sigma)-\varphi_C(\rho,\sigma')|\le\|C\|\,\|\sigma-\sigma'\|_1$ for all density operators $\rho$, $\sigma$ and $\sigma'$.

%#
%## SECTION: Stability, NIP and the independence property are continuity, membership in the first Baire class and $\beta\NN$
%#
\section{Stability, NIP and the independence property are continuity, membership in the first Baire class and $\beta\NN$}\label{sec:levels}

\noindent Stability and NIP admit a common description in terms of the regularity of idealized responses.
Recall from Part~I~\cite{Duenez-Iovino:2026a} that an observable is stable when, for all sequences of actual preparations of the two systems, the two iterated limits of its expectation values, one taken first over Alice’s sequence and the other first over Bob’s, agree whenever both exist.
In this sense its expectation value is independent of the order in which the two systems are idealized.
Here an idealized response is the expectation of a preparation of Alice’s system, actual or idealized, as a function of Bob’s probe, with Bob’s system idealized first.
In \cref{thm:levels} below, we characterize stability by continuity of every idealized response and NIP by membership of every idealized response in the first Baire class;
when NIP fails, a sequence of actual responses has a closure homeomorphic to the space $\beta\NN$ of ultrafilters on the natural numbers.
Within NIP, the finer classification concerns the topology of the compactum formed by these functions, and the affine dependence on preparations imposes a further restriction through convexity.

To state the result, we recall the notation of \cref{sec:responses}:
an observable $C$ is a bounded self-adjoint operator on the tensor product
$H_1\otimes H_2$ of the Hilbert spaces of Alice’s and Bob’s systems, its probe space
$Y_C$ is the compact metrizable state space of its probe algebra (see \cref{def:probe-space} and \cref{lem:probe-space}), and its response family $\cR_C$ consists of the responses of Alice’s actual preparations;
the closure $\overline{\cR_C}$ of this family is the set of idealized responses (see \cref{lem:closure}).

For stability we need one more fact from operator theory, Kadison’s function representation:
on the state space of a unital C$^*$-algebra $\fA$, the continuous affine real functions are exactly the evaluations $\omega_2\mapsto \omega_2(z)$ at self-adjoint elements $z\in\fA$ (see~\cite{Kadison:1951a}).

\begin{keytheorem}\label{thm:levels}
  Let $C$ be an observable on $H_1\otimes H_2$.
  \begin{enumerate}[label=\textup{(\alph*)}]
  \item
    The observable $C$ is stable if and only if every idealized response is continuous, that is, $\overline{\cR_C}\subseteq\Cp(Y_C)$.
  \item
    The observable $C$ is NIP for Alice’s responses if and only if every idealized response is Baire-1, that is, $\overline{\cR_C}\subseteq\Bone(Y_C)$.
    In that case, the set $\overline{\cR_C}$ of idealized responses is a Rosenthal compactum.
  \item
    If the observable $C$ is not NIP for Alice’s responses, then there is a sequence $(\rho_n)_{n\in\NN}$ of actual states of Alice’s system such that the closure of $\{\mathfrak{r}_{\rho_n}\}$ is homeomorphic to $\beta\NN$.
  \end{enumerate}
\end{keytheorem}

\begin{proof}
  Part (a) rests on the characterization of stability by weak compactness (see \cite[Theorem~5.1]{Duenez-Iovino:2026a}), and parts (b) and (c) rest on the dichotomy of \cref{thm:dichotomy}.

  For part (a), recall that a bounded operator between Banach spaces is \emph{weakly compact} if it maps bounded sets to relatively weakly compact sets,
  and that, by Gantmacher’s theorem~\cite{Gantmacher:1940}, this holds exactly when its bitranspose takes values in the original target space.
  If $C$ is stable, the slice map is weakly compact as a map into $B(H_2)$, and therefore also as a map into the closed subspace $\fA_C$, whose weak topology is the restriction of that of $B(H_2)$.
  The bitranspose $T_C''$ then maps $B(H_1)^*$ into $\fA_C$, whence every idealized response $\mathfrak{r}_{\omega_1}$ is evaluation at an element of $\fA_C$ and is continuous on $Y_C$.

  Conversely, suppose that every idealized response is continuous, and let $\omega_1\in\cS(B(H_1))$ be given.
  The response $\mathfrak{r}_{\omega_1}$ is real and affine on $Y_C$ by its defining formula, so Kadison’s function representation gives a self-adjoint $z\in\fA_C$ with $\mathfrak{r}_{\omega_1}(\omega_2)=\omega_2(z)$ for every $\omega_2\in Y_C$.
  Thus the elements $T_C''(\omega_1)$ and $z$ of $\fA_C^{**}$ agree on all states of $\fA_C$; since these states span $\fA_C^*$, the elements are equal.
  Every functional in $B(H_1)^*$ is a complex linear combination of states (see \cite[Lemma~2.2]{Duenez-Iovino:2026a}), so linearity implies $T_C''(B(H_1)^*)\subseteq\fA_C$.
  Gantmacher’s theorem therefore makes $T_C\colon\cT(H_1)\to\fA_C$ weakly compact.
  Composing with the bounded inclusion $\fA_C\hookrightarrow B(H_2)$ preserves weak compactness, and the characterization in \cite[Theorem~5.1]{Duenez-Iovino:2026a} implies that $C$ is stable.

  By \cref{lem:probe-space}, the dichotomy of \cref{thm:dichotomy} applies to the family $\cR_C\subseteq\Cp(Y_C)$, and part (b) follows from it together with \cref{lem:relcompact}.
  For part (c), let actual states $\rho_n\in\fD(H_1)$, for $n\in\NN$, and reals $a<b$ witness the independence property of $\cR_C$.
  By compactness of $Y_C$ (see the proof of \cref{thm:dichotomy}), for every $I\subseteq\NN$ some probe $\omega_2$ satisfies $\mathfrak{r}_{\rho_n}(\omega_2)\le a$ for $n\in I$ and $\mathfrak{r}_{\rho_n}(\omega_2)\ge b$ for $n\notin I$.
  No subsequence of $(\mathfrak{r}_{\rho_n})$ can therefore converge pointwise, and \cref{fact:rosenthal} yields a subsequence whose closure is homeomorphic to $\beta\NN$;
  its states form the required sequence.
\end{proof}

In terms of preparations and probes, \cref{thm:levels} describes how the responses of Alice’s preparations depend on Bob’s probe at each level:
\begin{itemize}[leftmargin=*]
\item
  for a stable observable, the response of each idealized preparation of Alice’s system depends continuously on Bob’s probe, idealized probes included;
\item
  for an observable that is NIP but unstable, such as the comparison of photon numbers
  (see \cref{sec:orders}), these responses can be discontinuous, but each is a pointwise limit of a sequence of continuous functions, and together they form a Rosenthal compactum;
\item
  for an observable with the independence property, a sequence of Alice’s actual preparations lets Bob read every finite pattern with a single probe,
  and every such sequence has a subsequence whose closure is a copy of $\beta\NN$ (see the proof of \cref{thm:levels}(c)).
\end{itemize}

\begin{remark}\label{rem:uniform}
  In Part~I~\cite{Duenez-Iovino:2026a} we distinguished stability (no infinite ladder of positive margin) from uniform stability (bounded ladders at every margin, equivalently stability on ultrapowers).
  A ladder is a pair of sequences of Alice’s and of Bob’s preparations, of the same length, finite or infinite, whose expectation values are high when Alice’s index does not exceed Bob’s and low otherwise, or the reverse, and it has margin
  $\delta$ when every high value exceeds every low value by at least $\delta$.
  The same distinction appears here.
  For a real number $\delta>0$, a finite set of functions on a set
  $Y$ is \emph{$\delta$-shattered} if there are reals $a<b$ with $b-a\ge\delta$ such that, for any two disjoint subsets $E$ and $F$ of it, some point $y\in Y$ satisfies $r(y)\le a$ for $r\in E$ and $r(y)\ge b$ for $r\in F$.
  The notion of NIP in \cref{def:ip} forbids an infinite sequence all of whose finite patterns are realized with fixed reals $a<b$, which we call an \emph{independent} sequence;
  \emph{uniform NIP} bounds, for each $\delta>0$, the size of the $\delta$-shattered subsets of the family, and so implies NIP.
  For a uniformly bounded family, such as a family of responses, uniform NIP amounts, by the pigeonhole principle applied to the thresholds and by Assouad’s duality argument (see \cref{rem:symmetry}), to the finiteness of the fat-shattering dimension at every gap.

  Comparisons $[A<B]$, the projections that represent the event “Bob’s value of $B$ exceeds Alice’s value of $A$” for quantities $A$ of Alice’s system and $B$ of Bob’s, are uniformly NIP, with bounds independent of the Hilbert spaces
  (see \cref{cor:comparisons-uniform}), and they are unstable when both quantities have pure point spectra unbounded above, by the ladder argument of \cref{sec:orders}.
  By \cref{thm:levels}, their idealized responses form a Rosenthal compactum of Baire-1 functions, not all continuous.

  As in Part~I, the uniform notion is the one that survives ultrapowers.
  Recall that, for an ultrafilter $\mathcal{U}$ on $\NN$, every bounded sequence
  $(x_k)_{k\in\NN}$ of complex numbers has a unique ultralimit $\lim_{k,\mathcal{U}}x_k$, the complex number $x$ such that $\{k:|x_k-x|<\varepsilon\}\in\mathcal{U}$ for every $\varepsilon>0$.
  The extension of a family of functions on $Y$
  to the ultrapower by $\mathcal{U}$ consists of the functions that send a sequence $(y_k)_{k\in\NN}$ of points of $Y$ to the ultralimit $\lim_{k,\mathcal{U}}r_k(y_k)$, one function for each sequence $(r_k)_{k\in\NN}$ of members of the family;
  each value is a single ultralimit along the common index $k$, not the value of a limit function at a limit point.
  A uniformly bounded family is uniformly NIP if and only if its extension to the ultrapower by every nonprincipal ultrafilter on $\NN$ is NIP.
  Indeed, by Łoś’s theorem, which says that the value of a formula at elements of the ultrapower, given by sequences, is the ultralimit of its values at the terms of the sequences, a
  $\delta$-shattered set of size $m$ in the ultrapower yields, for every $\delta'<\delta$, a $\delta'$-shattered set of size $m$ in the original family.
  Hence, if the sizes of the $\delta$-shattered subsets of the original family are bounded for every $\delta>0$, so are those of its extension, which then has no infinite independent sequence.
  Conversely, suppose that these sizes are unbounded for some $\delta>0$.
  The thresholds of the shattered sets lie in a bounded interval, because the family is uniformly bounded, and rounding each lower threshold upwards and each upper threshold downwards to a grid of mesh $\delta/3$ leaves a gap of at least $\delta/3$.
  Only finitely many rounded pairs occur, and the pigeonhole principle gives fixed reals $a<b$ with $b-a\ge\delta/3$ with which subsets of arbitrarily large finite size are shattered, and hence subsets of every finite size.
  In the ultrapower by any nonprincipal ultrafilter on $\NN$, which is countably incomplete, these finite shattered sets combine, by $\aleph_1$-saturation (every countable set of conditions, each finite subset of which is realized in the ultrapower, is realized in it), into an infinite independent sequence.

  Uniform NIP also has a statistical interpretation.
  After normalizing responses to $[0,1]$, consider independently sampled probe--outcome pairs from an unknown joint distribution.
  In agnostic learning with absolute loss, one seeks a predictor whose mean absolute error is within a prescribed tolerance of the infimum over the response class, with a prescribed probability of success.
  For the response classes considered here, distribution-free sample bounds of this kind exist exactly when the fat-shattering dimension is finite at every positive gap~\cite{Alon-BenDavid-CesaBianchi-Haussler:1997}, that is, exactly when the class is uniformly NIP.
  Part~I instead studies prediction along an adaptive sequence of probes, with a threshold fixed within each protocol and mistake bounds uniform in the threshold.
\end{remark}

%#
%### SUBSECTION: The finer classification of NIP observables
%#
\subsection{The finer classification of NIP observables}
For NIP observables, the classification of \cref{sec:classification} applies to the idealized responses.
We first check that they form a separable compactum.

\begin{lemma}\label{lem:separable}
  Let $C$ be an observable on $H_1\otimes H_2$ and $(\rho_n)_{n\in\NN}$ a sequence of density operators that is dense in $\fD(H_1)$ for the trace norm.
  Then the responses $\mathfrak{r}_{\rho_n}$ are dense in $\overline{\cR_C}$.
\end{lemma}

\begin{proof}
  For density operators $\rho$ and $\rho'$, we have $\sup_{\omega_2\in Y_C}|\mathfrak{r}_\rho(\omega_2)-\mathfrak{r}_{\rho'}(\omega_2)|\le\|T_C(\rho-\rho')\|\le\|C\|\,\|\rho-\rho'\|_1$.
  Every response is therefore a uniform limit, and hence a pointwise limit, of a sequence of responses $\mathfrak{r}_{\rho_n}$;
  the closure of $\{\mathfrak{r}_{\rho_n}\}$ thus contains $\cR_C$, and with it $\overline{\cR_C}$.
  Sequences such as $(\rho_n)$ exist, because $\cT(H_1)$ is separable.
\end{proof}

\begin{definition}\label{def:observable-levels}
  An observable $C$ that is NIP for Alice’s responses is \emph{NIP$^{\mathrm{fc}}$}, \emph{NIP$^{\mathrm{hs}}$} or \emph{NIP$^{\mathrm{met}}$} if the response family $\cR_C$ is, respectively, NIP$^{\mathrm{fc}}$, NIP$^{\mathrm{hs}}$ or NIP$^{\mathrm{met}}$ (see \cref{def:niplevels}).
\end{definition}

\begin{proposition}\label{prop:stable-metrizable}
  Every stable observable is NIP$^{\mathrm{met}}$:
  its idealized responses form a compact metrizable space.
\end{proposition}

\begin{proof}
  By \cref{thm:levels}(a), every idealized response is continuous on the compact metrizable space $Y_C$, hence Baire-1, and the observable $C$ is therefore NIP for Alice’s responses by \cref{thm:levels}(b);
  the set of idealized responses is separable by \cref{lem:separable}.
  A countable dense set of probes separates continuous functions, and \cref{prop:countable-probes} shows that the idealized responses form a metrizable space.
\end{proof}

%#
%### SUBSECTION: Mixtures force metrizability
%#
\subsection{Mixtures force metrizability}
The idealized responses of an observable form a convex set, because responses depend affinely on Alice’s state
(see the proof of \cref{cor:nip2}).
For separable convex Rosenthal compacta, hereditary separability implies metrizability (see \cref{thm:convex} below).
This is a special case of a known result on locally connected Rosenthal compacta (see \cref{rem:convex-literature}), for which we give a direct proof.
The mechanism is visible in the split interval $S(I)$ of \cref{sec:classification}, which is homeomorphic to the set of indicator functions $\one_{(a,1]}$ and $\one_{[a,1]}$ on $[0,1]$, for $0\le a\le1$.
For each such $a$, the two indicators form a split pair, two functions that differ only at $a$, and a convex set of functions that contains both of them contains their midpoint, which takes the value $\frac12$ at $a$.
These midpoints form an uncountable discrete set:
the midpoint for a point $b\ne a$ takes the value $0$ or $1$ at $a$.

\begin{keytheorem}[see {\cite[Theorems~5 and~7]{Todorcevic:1999}}, {\cite[Lemma~2.4]{Gruenhage:1988b}} and \cref{rem:convex-literature}]\label{thm:convex}
  Let $Y$ be a set, and let $K\subseteq\RR^Y$ be a separable compact convex set which, with the topology of pointwise convergence, is a Rosenthal compactum.
  If $K$ is hereditarily separable, then $K$ is metrizable.
\end{keytheorem}

\begin{proof}
  The proof proceeds by contradiction: if $K$ is hereditarily separable but not metrizable, the trichotomy supplies a topological copy in $K$ of one of the extended split Cantor sets.
  We take midpoints of its split pairs using the convex structure of $K\subseteq\RR^Y$; the embedding need not preserve affine combinations.
  The argument below uses convergence from either side of the split pairs to prove that these midpoints contain a nonseparable subspace, contradicting hereditary separability.

  Suppose, then, that $K$ is hereditarily separable but not metrizable.
  By \cref{fact:tricho}(b), there is a homeomorphism $h$ from $\hat{S}_+(2^\NN)$ or $\hat{S}_-(2^\NN)$, the closures of the minimal families $D_3$ and $D_4$ of \cref{sec:classification}, respectively, onto a subspace of $K$.
  We treat the case $\hat{S}=\hat{S}_+(2^\NN)$.
  The other case is symmetric:
  the map $c$ that exchanges the digits $0$ and $1$ reverses the lexicographic order
  $\prec$ of $\{0,1\}^\NN$, and the map $f\mapsto1-f\circ c$ carries $D_4$ to a family of the form of $D_3$ for the tree $c(\Theta)$, the image under $c$ of the regular dyadic subtree $\Theta$ fixed in \cref{sec:classification}.
  Write $Q=\{s^\frown0^\infty:s\in\Theta\}$ for the countable set of points indexing $D_3$, and $\Sigma$ for the closure of $Q$ in $\{0,1\}^\NN$.
  Recall that, for $\sigma\in\{0,1\}^\NN$, the functions $x^+_\sigma$ and $x^-_\sigma$ are the characteristic functions of $\{\zeta\in\{0,1\}^\NN:\sigma\preceq\zeta\}$ and $\{\zeta\in\{0,1\}^\NN:\sigma\prec\zeta\}$;
  thus $D_3$ consists of the functions $x^+_q$ for $q\in Q$.

  We first record the structure of $\hat{S}$.
  If $(\sigma_i)_{i\in\NN}$ is a strictly increasing sequence in $\{0,1\}^\NN$ converging to $\sigma$, then $x^\pm_{\sigma_i}\to x^+_\sigma$ pointwise;
  if it is strictly decreasing, then $x^\pm_{\sigma_i}\to x^-_\sigma$.
  Indeed, for $\zeta\in\{0,1\}^\NN$, the value of $x^\pm_{\sigma_i}$ at $\zeta$ is eventually $1$ if $\zeta\succ\sigma$, eventually $0$ if $\zeta\prec\sigma$, and at $\zeta=\sigma$ it is eventually $1$ in the increasing case and $0$ in the decreasing case.
  Since $\hat{S}$ is Fréchet--Urysohn (every point in the closure of a subset is the limit of a sequence from the subset;
  see~\cite[Theorem~3F]{Bourgain-Fremlin-Talagrand:1978})
  and every sequence of distinct points of $\{0,1\}^\NN$ has a strictly monotone subsequence, every point of $\hat{S}\setminus D_3$ has the form $x^+_\eta$ or $x^-_\eta$ with $\eta\in\Sigma$;
  and if $x^{e_i}_{\sigma_i}\in\hat{S}$ for a strictly monotone sequence $(\sigma_i)$ and signs $e_i\in\{+,-\}$, then the limit computed above lies in $\hat{S}$, because $\hat{S}$ is closed.
  Note also that $x^+_\sigma=x^-_\eta$ exactly when $\sigma$ is the immediate successor of $\eta$ in $\prec$, and that two of these functions with the same sign are equal only when their indices are.

  We use the following fact repeatedly:
  if $A\subseteq\{0,1\}^\NN$ is uncountable, then all but countably many points of $A$ are limits of strictly increasing and of strictly decreasing sequences from $A$.
  Indeed, if $\sigma\in A$ is neither the least element
  $0^\infty$ of $\{0,1\}^\NN$ nor a limit of a strictly increasing sequence from $A$, there is $\zeta_\sigma\prec\sigma$ with no point of $A$ strictly between them, and the intervals $(\zeta_\sigma,\sigma]$ for distinct such $\sigma$ are disjoint;
  each of them either contains a basic clopen set or makes $\zeta_\sigma$ the immediate predecessor of $\sigma$.
  The first alternative occurs for only countably many $\sigma$, because the intervals are disjoint and the basic clopen sets are countable, and so does the second, because only countably many points of $\{0,1\}^\NN$ have an immediate predecessor.
  The decreasing case is the same, with the greatest element $1^\infty$ in place of $0^\infty$.
  Let $P$ be the set of $\sigma\in\Sigma$ that are limits of strictly increasing and of strictly decreasing sequences from $Q$.
  The set $\Sigma$ is uncountable, because distinct branches of the dyadic tree
  $\Theta$ have distinct limits.
  By the same argument applied to $Q$, all but countably many points of $\Sigma$ belong to $P$ (an open interval that misses $Q$ also misses $\Sigma$);
  hence $P$ is uncountable.
  For $\sigma\in P$, both $x^+_\sigma$ and $x^-_\sigma$ belong to $\hat{S}$, they differ at $\sigma$, and neither equals $x^\pm_\eta$ for $\eta\ne\sigma$, because $\sigma$ has neither an immediate predecessor nor an immediate successor in $\prec$.

  The split pairs indexed by $P$ now supply the midpoints.
  For $\sigma\in P$, let $F^\pm_\sigma=h(x^\pm_\sigma)$ and $M_\sigma=\frac12(F^+_\sigma+F^-_\sigma)$, which lies in $K$ by convexity.
  Averaging is continuous for the pointwise topology, and so is $h$;
  consequently, if $(\sigma_i)$ is a strictly increasing (respectively, decreasing) sequence in $P$ converging to $\sigma$, then $M_{\sigma_i}\to h(x^+_\sigma)$ (respectively, $h(x^-_\sigma)$).
  The proof that the midpoints form a non-separable subspace takes three steps.

  \emph{Step 1: each value is taken by at most countably many midpoints.}
  Let $m\in K$ and suppose that $A=\{\sigma\in P:M_\sigma=m\}$ is uncountable.
  Choose $\sigma\in A$ that is a limit of strictly increasing and of strictly decreasing sequences from $A$.
  Then $h(x^+_\sigma)=m=h(x^-_\sigma)$, and since $h$ is injective, $x^+_\sigma=x^-_\sigma$, a contradiction.

  \emph{Step 2: only countably many midpoints lie on the copy of $\hat{S}$.}
  Let $E=\{\sigma\in P:M_\sigma\in h(\hat{S})\}$.
  By Step 1 and the countability of $D_3$, all but countably many
  $\sigma\in E$ satisfy $M_\sigma=h(x^{e(\sigma)}_{\eta(\sigma)})$ for some sign $e(\sigma)$ and some $\eta(\sigma)\in\Sigma$, and $\eta(\sigma)\ne\sigma$, since a midpoint of two distinct points differs from both.
  Suppose that uncountably many such $\sigma$ exist.
  Fix a metric $d$ on $\{0,1\}^\NN$ that induces its topology;
  then for some positive integer $n$ the set $E_n$ of these $\sigma$ with $d(\sigma,\eta(\sigma))\ge1/n$ is uncountable.
  Choose $\sigma\in E_n$ and a strictly monotone sequence $(\sigma_i)$ in $E_n$ converging to $\sigma$, say increasing.
  Passing to a subsequence, we may assume that the signs $e(\sigma_i)$ are constant and that the points $\eta(\sigma_i)$ are either constant or strictly monotone, converging to some $\eta^*$.
  Then the midpoints $M_{\sigma_i}$ converge both to $h(x^+_\sigma)$ and to $h(x^{\pm}_{\eta^*})$.
  Since $h$ is injective, $x^+_\sigma=x^\pm_{\eta^*}$, and since $\sigma\in P$ has no immediate predecessor, $\eta^*=\sigma$.
  Thus $d(\sigma_i,\eta(\sigma_i))\to0$, contrary to $\sigma_i\in E_n$.
  Consequently, the set $E$ is countable.

  \emph{Step 3: the midpoints form a non-separable subspace.}
  Let $P'=P\setminus E$, an uncountable set, and let $L=\{M_\sigma:\sigma\in P'\}$.
  Let $Z\subseteq L$ be countable.
  By Step 1, there is $\sigma\in P'$ with $M_\sigma\notin Z$.
  If $M_\sigma$ were in the closure of $Z$, then, since Rosenthal compacta are Fréchet--Urysohn (see~\cite[Theorem~3F]{Bourgain-Fremlin-Talagrand:1978}), there would be a sequence $(\sigma_i)$ in $P'$ with $M_{\sigma_i}\in Z$ and $M_{\sigma_i}\to M_\sigma$.
  As $M_\sigma\notin Z$, no point occurs infinitely often in $(\sigma_i)$, so a subsequence is strictly monotone and converges to some $\eta\in\Sigma$, and then the midpoints $M_{\sigma_i}$ converge to $h(x^+_\eta)$ or $h(x^-_\eta)$.
  Thus $M_\sigma\in h(\hat{S})$, contrary to $\sigma\notin E$.
  The subspace $L$ is therefore not separable, and $K$ is not hereditarily separable, a contradiction.
\end{proof}

\begin{remark}\label{rem:convex-literature}
  The proof uses convexity only to know that midpoints of points of $K$ lie in $K$.
  The same midpoints decide a claim about the split interval $S(I)=[0,1]\times\{0,1\}$, ordered lexicographically, with the order topology.
  In a remark following his Proposition~10~\cite{Godefroy:1980}, Godefroy states that the space $P(S(I))$ of Borel probability measures on $S(I)$, with the weak$^*$ topology, is hereditarily separable.
  This is not the case:
  for $0<a<1$, let $a^-=(a,0)$ and $a^+=(a,1)$, and let $m_a$ be the midpoint $\frac12(\delta_{a^-}+\delta_{a^+})$ of the point masses at $a^-$ and $a^+$.
  The set $\{s\in S(I):s\le a^-\}$ is clopen, so its indicator $f_a$ is continuous, and $m_a(f_a)=\frac12$, while $m_b(f_a)\in\{0,1\}$ for $b\ne a$.
  Hence the set $\{m_a:0<a<1\}$ is an uncountable discrete subspace of $P(S(I))$, as in the split-pair construction used for \cref{thm:convex}.
  Todorčević has pointed out to us that \cref{thm:convex} is known, and that it holds in the larger class of locally connected Rosenthal compacta, which contains the convex ones:
  it follows from his Theorems~5 and~7~\cite{Todorcevic:1999} and from an argument of Gruenhage~\cite[Lemma~2.4]{Gruenhage:1988b}.
  The proof above avoids local connectedness: it combines the canonical form of the trichotomy (see \cref{fact:tricho}(b)) and the Fréchet--Urysohn property of Rosenthal compacta with the midpoints of split pairs, the same midpoints that refute Godefroy’s remark.
  The examples of \cref{sec:orders} realize the other three of the four topological classes, the metrizable compacta and the three classes of Todorčević’s trichotomy.
  Together with these examples, \cref{thm:convex} determines which classes occur for convex response compacta, in the spirit of the Glasner--Megrelishvili hierarchy of tame systems~\cite{Glasner-Megrelishvili:2022}.
\end{remark}

\begin{corollary}\label{cor:nip2}
  Every observable that is NIP$^{\mathrm{hs}}$ is NIP$^{\mathrm{met}}$.
\end{corollary}

\begin{proof}
  Let $C$ be an observable on $H_1\otimes H_2$ that is NIP$^{\mathrm{hs}}$.
  We show that the set $\overline{\cR_C}$ of idealized responses satisfies the hypotheses of \cref{thm:convex}.
  Since the slice map and evaluation at a probe are linear and $\fD(H_1)$ is convex, the identity $\mathfrak{r}_{t\rho+(1-t)\rho'}=t\,\mathfrak{r}_\rho+(1-t)\,\mathfrak{r}_{\rho'}$, for $\rho,\rho'\in\fD(H_1)$ and $0\le t\le1$, shows that the family $\cR_C$ is convex.
  For fixed $t$, the map $(f,g)\mapsto tf+(1-t)g$ is continuous on $\RR^{Y_C}\times\RR^{Y_C}$ for the pointwise topology and sends $\cR_C\times\cR_C$ into $\cR_C$;
  it therefore sends $\overline{\cR_C}\times\overline{\cR_C}$, the closure of $\cR_C\times\cR_C$, into $\overline{\cR_C}$, which is thus convex.
  By \cref{lem:separable} and \cref{thm:levels}(b), the set $\overline{\cR_C}$ is a separable Rosenthal compactum, and it is hereditarily separable because $C$ is NIP$^{\mathrm{hs}}$.
  By \cref{thm:convex}, it is metrizable, that is, the observable $C$ is NIP$^{\mathrm{met}}$.
\end{proof}

For observables, then, at most three of the four topological classes can occur, and in \cref{sec:orders} we show that all three do.
For families of continuous functions all four occur, as the minimal families show, and the class that observables miss is realized by $D_3$:
the members of $D_3$, the characteristic functions of the sets
$\{\zeta\in\{0,1\}^\NN:\zeta\succeq s^\frown0^\infty\}$ for
$s$ in the regular dyadic subtree $\Theta$, are continuous on the Cantor set, since each of these sets is clopen, and their closure
$\hat{S}_+(2^\NN)$ is hereditarily separable but not metrizable, so the family $D_3$ is NIP$^{\mathrm{hs}}$ but not NIP$^{\mathrm{met}}$.

This example comes from the lexicographic order of the Cantor set, which is isomorphic to the order that the real line induces on the middle-thirds Cantor set;
the members of $D_3$ are realized types of this order, and they take only the values $0$ and $1$.
Expectations bring in mixtures, and the midpoints of the split pairs turn the split interval into a space with an uncountable discrete subset;
for the order of positions, the resulting space is the space of nondecreasing functions from
$[0,1]$ to $[0,1]$ that take the value $0$ at $0$ and $1$ at $1$, a closed subspace of the Helly space of all nondecreasing functions from $[0,1]$ to $[0,1]$ (see \cref{sec:orders}).
This is why, for averages over Alice’s mixed preparations, hereditary separability forces metrizability (see \cref{thm:convex}).

In the next theorem, increasing injections, equivalence of families and the minimal families $D_5$, $D_6$ and $D_7$ are those of \cref{sec:classification}.

\begin{keytheorem}\label{thm:observable-classes}
  Let $C$ be an observable on $H_1\otimes H_2$ that is NIP for Alice’s responses, and let $(\rho_n)_{n\in\NN}$ be a trace-norm dense sequence in $\fD(H_1)$.
  Then exactly one of the following holds:
  \begin{enumerate}[label=\textup{(\roman*)}]
  \item
    the observable $C$ is NIP$^{\mathrm{met}}$;
  \item
    the observable $C$ is NIP$^{\mathrm{fc}}$ but not NIP$^{\mathrm{hs}}$, and there is an increasing injection $\iota\colon\{0,1\}^{<\NN}\to\NN$ such that the responses $\{\mathfrak{r}_{\rho_{\iota(t)}}:t\in\{0,1\}^{<\NN}\}$ are equivalent to $D_6$ or to $D_7$;
  \item
    the observable $C$ is not NIP$^{\mathrm{fc}}$, and there is an increasing injection $\iota\colon\{0,1\}^{<\NN}\to\NN$ such that the responses $\{\mathfrak{r}_{\rho_{\iota(t)}}:t\in\{0,1\}^{<\NN}\}$ are equivalent to $D_5$.
  \end{enumerate}
\end{keytheorem}

\begin{proof}
  By \cref{lem:probe-space}, the family $\{\mathfrak{r}_{\rho_n}\}$ is a pointwise bounded family of continuous functions on the compact metrizable space $Y_C$, and it is NIP because it is contained in $\cR_C$.
  By \cref{lem:separable}, its closure is $\overline{\cR_C}$.
  Hence \cref{thm:nipclasses} applies, and its case (ii) is excluded by \cref{cor:nip2}.
\end{proof}

Physically, when the idealized responses of an NIP observable do not form a metrizable space, the reason is always visible in Alice’s actual preparations:
a family of them, arranged along the dyadic tree, has idealized limits that reproduce one of three canonical compacta, the closures of the minimal families $D_5$, $D_6$ and $D_7$.
In the next section we locate the comparisons $[A<B]$ in this classification.

%#
%### SUBSECTION: Historical remarks
%#
\subsection{Historical remarks}
Gantmacher characterized weakly compact operators by their bitransposes in~1940~\cite{Gantmacher:1940}, and Kadison proved the function representation used in \cref{thm:levels} in~1951~\cite{Kadison:1951a}.
Grothendieck’s theorem~\cite{Grothendieck:1952}, on which the characterizations of stability in Part~I~\cite{Duenez-Iovino:2026a} rest, lies behind the characterization of stability by the continuity of the idealized responses in \cref{thm:levels}(a);
the NIP version of this characterization in \cref{thm:levels}(b), with Baire-1 functions in place of continuous ones, follows our work on deep computations~\cite{Duenez-Iovino-MatosWiederhold-Salvetti-Tall:2026}.

%#
%## SECTION: Three physical orders realize the three classes that occur for observables
%#
\section{Three physical orders realize the three classes that occur for observables}\label{sec:orders}

\noindent Each of the three classes of \cref{thm:observable-classes} contains a physically natural order, and the three orders are unstable and NIP.
The comparison of photon numbers has a metrizable compactum of idealized responses and is thus NIP$^{\mathrm{met}}$ (see \cref{thm:orders}).
The comparison of the positions of two particles, in a box or on a line, or of the quadratures of two modes of light
(a quadrature is one of the two field amplitudes of a mode, which behave like a position and a momentum),
has idealized responses that form a copy of the space of nondecreasing functions from $[0,1]$ to $[0,1]$ that take the value $0$ at $0$ and $1$ at $1$, a closed subspace of the Helly space;
this space is first countable but not hereditarily separable, and the comparison is therefore NIP$^{\mathrm{fc}}$ but not NIP$^{\mathrm{hs}}$ (see \cref{thm:orders} and \cref{prop:cv-helly}).
The causal order of a spacetime with one dimension of space,
“Bob’s detection lies in the causal future of Alice’s event”, is NIP but not NIP$^{\mathrm{fc}}$:
since a light cone can shrink to a point, its compactum of idealized responses is not first countable (see \cref{thm:causal}).
The comparison of two phases, “Bob’s angle is ahead of Alice’s by less than half a turn”, a cyclic relation rather than an order, falls in the same class, because on a circle the jumps at the two ends of a half-turn can cancel (see \cref{thm:phases}).

%#
%### SUBSECTION: Relations between classical quantities
%#
\subsection{Relations between classical quantities}
All three orders relate a classical quantity measured by Alice to one measured by Bob, and their responses can be computed from classical probability distributions.
Let $(X_1,m_1)$ and $(X_2,m_2)$ be $\sigma$-finite measure spaces such that the Hilbert spaces
$H_1=L^2(X_1,m_1)$ and $H_2=L^2(X_2,m_2)$ are separable, fix an integer $d\ge1$, let $\alpha\colon X_1\to\RR^d$ and $\beta\colon X_2\to\RR$ be measurable, and let $R\subseteq\RR^d\times\RR$ be a Borel set.
The \emph{relation observable} $C_R$ is multiplication on $L^2(X_1\times X_2,m_1\otimes m_2)=H_1\otimes H_2$ by the indicator of $\{(s,t):(\alpha(s),\beta(t))\in R\}$;
it represents the event “Alice’s value and Bob’s value are related by $R$”.
In this section, the letters $\rho$ and $\sigma$ denote actual states of Alice’s and of Bob’s systems, that is, density operators on $H_1$ and on $H_2$.
For a state $\rho$ of Alice’s system, the \emph{distribution} of Alice’s value in
$\rho$ is the Borel probability measure on $\RR^d$
given by $\mu_\rho(U)=\Tr\bigl(\rho\,M_{\one_{\alpha^{-1}(U)}}\bigr)$
for Borel sets $U\subseteq\RR^d$, where $M_f$ denotes multiplication by $f$, and the distribution $\mu_\sigma$ of Bob’s value in a state $\sigma$ of Bob’s system is defined in the same way.
For $y\in\RR$, let $R^y=\{x\in\RR^d:(x,y)\in R\}$, and, for a state $\rho$ of Alice’s system, let $g_\rho\colon\RR\to[0,1]$ be given by $g_\rho(y)=\mu_\rho(R^y)$, the probability that Alice’s value is related by $R$ to the value $y$;
it is Borel measurable, by Tonelli’s theorem.
Recall from \cref{sec:responses} that the effective observable $T_C(\rho)$ of an observable $C$ is the observable on Bob’s system obtained by fixing Alice’s state $\rho$.

By the spectral theorem, every self-adjoint operator on a separable Hilbert space is unitarily equivalent to multiplication by a real measurable function on a space $L^2(X,m)$ of this kind.
Thus, for self-adjoint operators $A$ on $H_1$ and $B$ on $H_2$, realized as multiplication by $\alpha$ and by $\beta$, the \emph{comparison} $[A<B]$ is the relation observable $C_R$ for $R=\{(x,y):x<y\}$;
it is the projection $(E_A\otimes E_B)(R)$, where $E_A$ and $E_B$ are the spectral measures, so it does not depend on the realization.
When $A$ and $B$ are bounded, the operator $I\otimes B-A\otimes I$
is multiplication by $\beta(t)-\alpha(s)$, and $[A<B]$ is its spectral projection for the interval $(0,\infty)$, or equivalently the spectral projection of $A\otimes I-I\otimes B$ for the interval $(-\infty,0)$.

\begin{lemma}\label{lem:product}
  Let $C_R$ be a relation observable, $\rho$ a state of Alice’s system and $\sigma$ a state of Bob’s system.
  Then the expectation value of $C_R$ when the two systems are prepared independently in these states is
  \[
    \varphi_{C_R}(\rho,\sigma)=(\mu_\rho\otimes\mu_\sigma)(R)=\int g_\rho\,d\mu_\sigma,
  \]
  and the effective observable $T_{C_R}(\rho)$ is multiplication by $g_\rho\circ\beta$.
\end{lemma}

\begin{proof}
  For unit vectors $\psi\in H_1$ and $\chi\in H_2$, with distributions $\mu_\psi$ and $\mu_\chi$ in the corresponding pure states, Tonelli’s theorem and a change of variables give
  \[
    \langle\psi\otimes\chi,C_R(\psi\otimes\chi)\rangle
    =\iint\one_R\bigl(\alpha(s),\beta(t)\bigr)\,|\psi(s)|^2|\chi(t)|^2\,dm_1(s)\,dm_2(t)
    =(\mu_\psi\otimes\mu_\chi)(R).
  \]
  A density operator is a countable convex combination of pure states, $\rho=\sum_j\lambda_j|\psi_j\rangle\langle\psi_j|$, the series converging in trace norm;
  here $j$ ranges over a countable set, the weights $\lambda_j\ge0$ have sum $1$, and $|\psi_j\rangle\langle\psi_j|$ is the orthogonal projection onto the line spanned by a unit vector $\psi_j\in H_1$.
  Accordingly, the distribution of Alice’s value in $\rho$ is $\mu_\rho=\sum_j\lambda_j\mu_{\psi_j}$.
  Both sides of the first identity are affine in each state and respect such countable combinations, so the first identity holds for all states;
  the second is Tonelli’s theorem.
  Finally, for every unit vector $\chi\in H_2$,
  \[
    \langle\chi,T_{C_R}(\rho)\chi\rangle=\varphi_{C_R}(\rho,|\chi\rangle\langle\chi|)
    =\int g_\rho\bigl(\beta(t)\bigr)|\chi(t)|^2\,dm_2(t),
  \]
  and a bounded operator on a complex Hilbert space is determined by these values.
\end{proof}

When the functions $g_\rho$ are monotone, or differences of monotone functions, the responses cannot shatter large sets (see \cref{prop:uniform-orders} below).
Recall from \cref{rem:uniform} that a finite set of responses is $\delta$-shattered if every pattern on it is realized with reals $a<b$ such that $b-a\ge\delta$.
For a real number $\delta>0$, let $m_\delta=\lceil9/\delta\rceil$, and let $N(\delta)$ be the largest positive integer $n$ with $2^n\le(2nm_\delta+1)^{m_\delta}$;
it exists because $2^n$ grows faster than any polynomial in $n$.

\begin{lemma}\label{lem:monotone-shatter}
  Let $\delta>0$ be a real number and $n$ a positive integer, let $h_1,\dots,h_n\colon\RR\to\RR$ be functions of the form $h_i=u_i-v_i$, where $u_i,v_i\colon\RR\to[0,1]$ are nondecreasing, and let $a<b$ be reals with $b-a\ge\delta$.
  Suppose that for every $S\subseteq\{1,\dots,n\}$ there is a Borel probability measure $\nu_S$ on $\RR$ with $\int h_i\,d\nu_S\ge b$ for $i\in S$ and $\int h_i\,d\nu_S\le a$ for $i\notin S$.
  Then $2^n\le(2nm_\delta+1)^{m_\delta}$, and hence $n\le N(\delta)$.
\end{lemma}

\begin{proof}
  Write $m=m_\delta$.
  We first approximate a measure by $m$ equally weighted points.
  Let $\nu$ be a Borel probability measure on $\RR$, with distribution function $F(t)=\nu((-\infty,t])$.
  For $k=1,\dots,m$, let $y_k=\min\{t:F(t)\ge(k-\frac12)/m\}$;
  the minimum exists because $F$ is nondecreasing and right continuous, with limits $0$ and $1$ at $-\infty$ and $+\infty$, and $0<(k-\frac12)/m<1$.
  Let $\nu'$ be the average of the $m$ point masses at the points $y_1,\dots,y_m$.
  Since $y_k\le t$ exactly when $F(t)\ge(k-\frac12)/m$, we have
  $\nu'((-\infty,t])=\lfloor mF(t)+\frac12\rfloor/m$, whence $|\nu((-\infty,t])-\nu'((-\infty,t])|\le1/m$ for every $t$;
  by taking limits from the left, we obtain the same bound for the half-lines $(-\infty,t)$.
  Hence $|\nu(J)-\nu'(J)|\le1/m$ for every set $J$ of the form $\emptyset$, $\RR$, $[t,\infty)$ or $(t,\infty)$.
  For a nondecreasing $u\colon\RR\to[0,1]$, every set $\{u>s\}$ has this form, and $\int u\,d\nu=\int_0^1\nu(\{u>s\})\,ds$;
  therefore
  \begin{equation}\label{eq:quantile}
  \Bigl|\int u\,d\nu-\int u\,d\nu'\Bigr|\le\frac1m.
  \end{equation}

  Next we discretize the functions.
  Let $\tilde{u}_i=\lfloor mu_i\rfloor/m$ and $\tilde{v}_i=\lfloor mv_i\rfloor/m$.
  These functions are nondecreasing, take values in $\{0,\frac1m,\dots,1\}$, and differ from $u_i$ and $v_i$ by at most $\frac1m$.
  Each of the $2n$ functions $\tilde{u}_i,\tilde{v}_i$ changes value at most $m$ times along $\RR$, so the map $Q(y)=(\tilde{u}_1(y),\dots,\tilde{u}_n(y),\tilde{v}_1(y),\dots,\tilde{v}_n(y))$ takes at most $2nm+1$ values.

  For each $S$, let $y^S_1,\dots,y^S_m$ be the points constructed from $\nu_S$ as above, and let $w(S)=(Q(y^S_1),\dots,Q(y^S_m))$.
  By~\eqref{eq:quantile}, applied to $u_i$ and to $v_i$, and by the bound on the discretization,
  \[
    \Bigl|\int h_i\,d\nu_S-\frac1m\sum_{k=1}^m\bigl(\tilde{u}_i(y^S_k)-\tilde{v}_i(y^S_k)\bigr)\Bigr|
    \le\frac4m\le\frac{4\delta}9 ,
  \]
  and the average on the left is determined by $w(S)$.
  If $S\ne S'$, choose $i$ in one of the two sets but not the other, say $i\in S\setminus S'$.
  Then $\int h_i\,d\nu_S-\int h_i\,d\nu_{S'}\ge\delta>\frac{8\delta}9$, so the two averages differ, and $w(S)\ne w(S')$.
  Thus the map $S\mapsto w(S)$ is injective, and $2^n\le(2nm+1)^m$.
\end{proof}

\begin{proposition}\label{prop:uniform-orders}
  Let $C_R$ be a relation observable and $c\ge1$ a constant such that, for every state $\rho$ of Alice’s system, $g_\rho=c\,(u_\rho-v_\rho)$ for two nondecreasing functions $u_\rho,v_\rho\colon\RR\to[0,1]$.
  Then, for every $\delta>0$, no set of more than $N(\delta/(2c))$ responses of Alice’s states is $\delta$-shattered by points of the probe space.
  Hence the observable $C_R$ is uniformly NIP, with bounds that do not depend on the Hilbert spaces, and in particular it is NIP.
\end{proposition}

\begin{proof}
  Suppose that the responses $\mathfrak{r}_{\rho_1},\dots,\mathfrak{r}_{\rho_n}$ are $\delta$-shattered with reals $a<b$, the pattern $S$ being realized by a probe $\omega_{2,S}$ in the probe space $Y_{C_R}$ of $C_R$.
  Let $a'=a+\delta/4$ and $b'=b-\delta/4$.
  The set of probes $\omega_2$ with $\mathfrak{r}_{\rho_i}(\omega_2)>b'$
  for $i\in S$ and $\mathfrak{r}_{\rho_i}(\omega_2)<a'$ for $i\notin S$ is open, by \cref{lem:probe-space}(b), and contains $\omega_{2,S}$, so, by \cref{lem:probe-space}(c), it contains the record $\mathfrak{p}_{\sigma_S}$ of an actual state $\sigma_S$ of Bob’s system.
  By \cref{lem:product}, the value of $\mathfrak{r}_{\rho_i}$
  there is $\int g_{\rho_i}\,d\mu_{\sigma_S}$, and \cref{lem:monotone-shatter}, applied with $h_i=g_{\rho_i}/c=u_{\rho_i}-v_{\rho_i}$, $\nu_S=\mu_{\sigma_S}$ and the reals $a'/c<b'/c$, yields $n\le N(\delta/(2c))$.
  An infinite independent sequence with reals $a<b$ would contain $(b-a)$-shattered sets of every finite size, so the observable $C_R$ is NIP.
\end{proof}

\begin{corollary}\label{cor:comparisons-uniform}
  Let $A$ and $B$ be self-adjoint operators on $H_1$ and on $H_2$.
  Then, for every $\delta>0$, no set of more than $N(\delta/2)$ responses of Alice’s states for the comparison $[A<B]$ is $\delta$-shattered by points of its probe space.
  In particular, the comparison $[A<B]$ is uniformly NIP, with bounds independent of the Hilbert spaces, and it is NIP.
\end{corollary}

\begin{proof}
  Here $R^y=(-\infty,y)$, so the function $g_\rho(y)=\mu_\rho((-\infty,y))$ is nondecreasing, with values in $[0,1]$.
  The hypothesis of \cref{prop:uniform-orders} therefore holds, with $c=1$, $u_\rho=g_\rho$ and $v_\rho=0$.
\end{proof}

\begin{remark}[Commuting observables and the general framework]\label{rem:commutative}
  The general results and the commuting examples play complementary roles.
  The definitions of \cref{sec:responses} and the results of \cref{thm:levels,thm:observable-classes} and \cref{cor:nip2} apply to every bounded self-adjoint bipartite observable;
  for the swap test, for example, the probe algebra contains all compact operators.
  For every such observable the full compactum of idealized responses is convex;
  when the observable is NIP, hereditary separability therefore implies metrizability.
  The orders considered here have commutative probe algebras and realize all three resulting classes;
  the label-controlled observables of \cref{sec:hepta}, in which Alice’s label selects a test on Bob’s system, also have commutative probe algebras.
  For these realizations, integration against classical probability distributions makes the responses explicit.
  For the relation observables that satisfy the monotonicity hypothesis of \cref{prop:uniform-orders}, the proposition supplies bounds independent of the Hilbert spaces.
  In \cref{sec:ssb} we pass to parameterized magnetization families, including those of the Ising model in a transverse field, whose Hamiltonian has noncommuting terms;
  their classification follows from the same analysis of pointwise closures (see \cref{thm:qssb}).
  The question of which further phenomena arise from noncommutative probe algebras is stated in \cref{ques:noncommutative}.
\end{remark}

%#
%### SUBSECTION: Order observables
%#
\subsection{Order observables}
We now fix two model comparisons, one of discrete and one of continuous quantities, whose values lie in a compact set $\Xi\subseteq[0,1]$ containing $0$ and $1$.
The two cases are the following:
\begin{itemize}[leftmargin=*]
\item
  \emph{discrete}:
  the set $\Xi$ is $\{1-\frac1{l+1}:l\in\NN\}\cup\{1\}$, the Hilbert spaces are $H_1=H_2=\ell^2(\NN)$, and the operator $A=B$ is the diagonal operator with eigenvalue $1-\frac1{l+1}$ on the number ket $|l\rangle$;
  this operator is a bounded relabeling of the number operator, which counts photons and has eigenvalue $l$ on $|l\rangle$, and the order of its eigenvalues is that of the photon numbers;
\item
  \emph{continuous}:
  the set $\Xi$ is $[0,1]$, the Hilbert spaces are $H_1=H_2=L^2[0,1]$, and the operator $A=B$ is multiplication by $x$, the position of a particle confined to a box, measured from one end of the box in units of its length, so that it ranges over the interval $[0,1]$.
\end{itemize}
In both cases, let $C=[A<B]$ be the spectral projection of $I\otimes B-A\otimes I$ for the interval $(0,\infty)$, which represents the event “Bob’s value exceeds Alice’s”.
For a state $\rho$, let $\mu_\rho$ be the distribution of $A$ in $\rho$, a Borel probability measure on $\Xi$.
In the discrete case every probability measure on
$\Xi\setminus\{1\}$ occurs, and in the continuous case the measures that occur are exactly the absolutely continuous ones (for a probability density $p$, take the pure state given by the unit vector $\psi=\sqrt{p}$).
By \cref{lem:product}, the expectation value of $C$ for Alice’s state $\rho$ and Bob’s state $\sigma$ is
\[
  \varphi_C(\rho,\sigma)=(\mu_\rho\otimes\mu_\sigma)\bigl(\{(x,y):x<y\}\bigr)
  =\int g_\rho\,d\mu_\sigma,\qquad g_\rho(y)=\mu_\rho\bigl([0,y)\bigr).
\]
The function $g_\rho\colon \Xi\to[0,1]$ is nondecreasing, with $g_\rho(0)=0$ and $g_\rho(1)=1$ in both cases, since $\mu_\rho$ assigns no mass to the point $1$;
in the continuous case it is also continuous.
Accordingly, the effective observable is $T_C(\rho)=g_\rho(B)$.
By \cref{cor:comparisons-uniform}, the observable
$C$ is uniformly NIP, and in particular NIP.
It is unstable:
in the discrete case by the interleaving argument of Part~I~\cite[Proposition~4.6(3)]{Duenez-Iovino:2026a},
which uses only the order of the eigenvalues, and in the continuous case by the same ladder argument, with states of Alice’s and of Bob’s particles supported in small disjoint intervals around positions
$a_k$ and $b_l$, respectively, for $k,l\ge1$, interleaved as $a_1<b_1<a_2<b_2<\cdots$.
Writing $\rho_k$ and $\sigma_l$ for these preparations of Alice and Bob, respectively, we obtain $\varphi_C(\rho_k,\sigma_l)=1$ for $k\le l$ and $0$ for $k>l$;
the iterated limit taken first in $k$ is $0$, whereas the one taken first in $l$ is $1$.

%#
%### SUBSECTION: The probe space and the idealized responses
%#
\subsection{The probe space and the idealized responses}
For both comparisons of the previous subsection, the discrete and the continuous one, and for every state
$\rho$ of Alice’s system, the function $g_\rho$, whose value at a point $y\in\Xi$ is the probability that Alice’s value lies below $y$, is nondecreasing, and the pointwise limits of these functions describe all idealized responses (see \cref{prop:order-responses} below).
Recall that the probe algebra $\fA_C$ is the unital C$^*$-algebra generated by the effective observables, and that the probe space $Y_C$ is its state space.
We write $C(\Xi)$ for the algebra of continuous complex functions on $\Xi$, with the supremum norm, and $M_{\Xi}$ for the set of nondecreasing functions $g\colon \Xi\to[0,1]$ with $g(0)=0$ and $g(1)=1$, with the topology of pointwise convergence;
it is a closed subset of $[0,1]^{\Xi}$, hence compact.
Although $M_{\Xi}$ contains discontinuous functions, integration against a fixed probability measure is continuous on it, by the following lemma;
on the compact metrizable space $\Xi$, Radon and Borel probability measures coincide, so the lemma applies to every point of the probe space described in \cref{prop:order-responses}.

\begin{lemma}\label{lem:monotone-nets}
  Let $\nu$ be a Borel probability measure on $\Xi$, and let $(g_i)$ be a net in $M_{\Xi}$, indexed by a directed set, that converges pointwise to a function $g\in M_{\Xi}$.
  Then $\int g_i\,d\nu\to\int g\,d\nu$.
\end{lemma}

\begin{proof}
  Let $\varepsilon>0$ be given.
  Choose a positive integer $q$ with $1/q<\varepsilon$.
  For each $j=1,\dots,q-1$, take the least point $x\in\Xi$ with $\nu(\Xi\cap[0,x])\ge j/q$, and adjoin $0$ and $1$.
  List the distinct points obtained as $0=p_0<p_1<\cdots<p_m=1$.
  Each open cell $c_k=\Xi\cap(p_{k-1},p_k)$ has measure at most $1/q<\varepsilon$:
  a cell of greater measure would contain a point at which one of the intervening levels $j/q$ is first reached, contrary to the choice of consecutive partition points.
  Let $\eta_i=\max_k|g_i(p_k)-g(p_k)|$, which tends to $0$.
  On $c_k$ both $g_i$ and $g$ take values between their values at $p_{k-1}$ and at $p_k$, because they are nondecreasing.
  Hence
  \[
    \Bigl|\int_{c_k}g_i\,d\nu-\int_{c_k}g\,d\nu\Bigr|
    \le\nu(c_k)\bigl(g(p_k)-g(p_{k-1})+2\eta_i\bigr),
  \]
  and summing over $k$, together with the contribution of the points $p_k$, shows that $|\int g_i\,d\nu-\int g\,d\nu|\le\varepsilon+3\eta_i$, since the differences $g(p_k)-g(p_{k-1})$ add up to at most $1$.
\end{proof}

\begin{proposition}\label{prop:order-responses}
  For each of the two comparisons $C=[A<B]$, the discrete and the continuous one, the probe algebra
  $\fA_C$ is $\{f(B):f\in C(\Xi)\}$, so its state space $Y_C$ is identified with the space $P(\Xi)$ of Radon probability measures on $\Xi$, with the weak$^*$ topology, by $\nu\mapsto[f(B)\mapsto\int_\Xi f\,d\nu]$.
  Under this identification, the record $\mathfrak{p}_\sigma$ of an actual state $\sigma$ of Bob’s system corresponds to $\mu_\sigma$:
  these are exactly the measures with $\mu_\sigma(\{1\})=0$ in the discrete case and the absolutely continuous measures in the continuous case.
  The map $\Phi\colon \mathfrak{r}\mapsto\bigl(\mathfrak{r}(\delta_y)\bigr)_{y\in \Xi}$, where
  $\delta_y$ is the point mass at $y$, is a homeomorphism from the compactum $\overline{\cR_C}$ of idealized responses onto $M_{\Xi}$, and every idealized response is given by $\mathfrak{r}(\nu)=\int\Phi(\mathfrak{r})\,d\nu$.
\end{proposition}

\begin{proof}
  Each function $g_\rho$ is continuous on $\Xi$:
  in the continuous case this follows from absolute continuity of the distribution, and in the discrete case continuity at the only accumulation point follows from $g_\rho(1-1/(l+1))\to1=g_\rho(1)$ as $l\to\infty$.
  The functions $g_\rho$ for states $\rho$ whose distribution has full support in $\Xi$ separate the points of $\Xi$, and the map $f\mapsto f(B)$ is an isometric $*$-isomorphism from $C(\Xi)$ onto its image, because the spectrum of $B$ is $\Xi$.
  By the Stone--Weierstrass theorem, the unital C$^*$-algebra generated by the operators $g_\rho(B)$ is $\{f(B):f\in C(\Xi)\}$, and the slice map takes values in it.
  By the Riesz representation theorem, its state space is $P(\Xi)$, and the response of $\rho$ is $\mathfrak{r}_\rho(\nu)=\int g_\rho\,d\nu$.

  The map $\Phi$ is continuous, because its coordinates are evaluations.
  Let $\mathfrak{r}\in\overline{\cR_C}$ be the pointwise limit of a net $(\mathfrak{r}_{\rho_i})$.
  Then the functions $g_{\rho_i}=\Phi(\mathfrak{r}_{\rho_i})$ converge pointwise to $g=\Phi(\mathfrak{r})\in M_{\Xi}$, and, by \cref{lem:monotone-nets}, $\mathfrak{r}(\nu)=\lim_i\int g_{\rho_i}\,d\nu=\int g\,d\nu$.
  Thus the response $\mathfrak{r}$ is determined by $\Phi(\mathfrak{r})$, the map $\Phi$ is injective, and, as a continuous injection from a compact space, it is a homeomorphism onto its image.
  The image is compact and contains the functions $g_\rho$.

  It remains to be shown that the functions $g_\rho$ are dense in $M_{\Xi}$.
  Let $g\in M_{\Xi}$ and a finite set $F\subseteq \Xi$ containing $0$ and $1$ be given.
  In the continuous case, the piecewise linear interpolation $\tilde{g}$ of $g$ on $F$ is continuous and nondecreasing, with values $0$ at $0$ and $1$ at $1$;
  hence $\tilde{g}=g_\rho$ for a state $\rho$ whose distribution has $\tilde{g}$ as its distribution function, which is absolutely continuous, and $\tilde{g}$ agrees with $g$ on $F$.
  In the discrete case, write $F=\{\xi_0<\xi_1<\cdots<\xi_m\}$, with $\xi_0=0$ and $\xi_m=1$.
  The probability measure with mass $g(\xi_{j+1})-g(\xi_j)$
  at $\xi_j$ for $0\le j<m$ is the distribution of a state $\rho$ with finitely many photons, and $g_\rho(\xi_j)=g(\xi_j)-g(\xi_0)=g(\xi_j)$ for every $j$, so the function $g_\rho$ agrees with $g$ on $F$.
  Every neighborhood of $g$ therefore contains some $g_\rho$.
\end{proof}

%#
%### SUBSECTION: Location in the classification
%#
\subsection{Location in the classification}
We can now place both comparisons in the classification of \cref{thm:observable-classes}.
Recall from \cref{sec:classification} that $D_6$ and $D_7$ are the two minimal families of the heptachotomy whose closures are first countable but not hereditarily separable.

\begin{keytheorem}\label{thm:orders}
  Let $C=[A<B]$ be the comparison of Alice’s and Bob’s photon numbers, in the discrete case, or of the positions of their particles in a box, in the continuous case.
  \begin{enumerate}[label=\textup{(\alph*)}]
  \item
    In the discrete case, the observable $C$ is unstable and NIP$^{\mathrm{met}}$.
  \item
    In the continuous case, the observable $C$ is unstable and NIP$^{\mathrm{fc}}$, but not NIP$^{\mathrm{hs}}$.
    Hence, for every trace-norm dense sequence $(\rho_n)_{n\in\NN}$
    of states of Alice’s particle, there is an increasing injection
    $\iota\colon\{0,1\}^{<\NN}\to\NN$ such that the responses $\{\mathfrak{r}_{\rho_{\iota(t)}}:t\in\{0,1\}^{<\NN}\}$ are equivalent to $D_6$ or to $D_7$.
  \end{enumerate}
\end{keytheorem}

\begin{proof}
  (a) By \cref{prop:order-responses}, the compactum of idealized responses is homeomorphic to a subspace of $[0,1]^{\Xi}$ with $\Xi$ countable, which is metrizable.

  (b) By \cref{prop:order-responses}, the compactum of idealized responses is homeomorphic to $M_{[0,1]}$, a copy of the closed subspace of the Helly space described at the start of this section.
  It is first countable.
  Given $g\in M_{[0,1]}$, let $Z_g$ be the countable set consisting of the rational points, the jump points of $g$ and the endpoints;
  we claim that the sets $\{h:|h(p)-g(p)|<1/n\text{ for }p\in F\}$, with $F\subseteq Z_g$ finite and $n\ge1$, form a neighborhood base at $g$.
  It suffices to show that each set $\{h:|h(y)-g(y)|<\varepsilon\}$, with $y\in[0,1]$ and $\varepsilon>0$, contains one of them.
  If $y\in Z_g$, this is clear.
  If $y\notin Z_g$, then $g$ is continuous at $y$, and there are rationals $p<y<q$ with $g(q)-g(p)<\varepsilon/2$.
  Every $h\in M_{[0,1]}$ with $|h-g|<\varepsilon/4$ at $p$ and $q$ satisfies $h(p)\le h(y)\le h(q)$, because $h$ is nondecreasing, and hence $|h(y)-g(y)|<\varepsilon$.

  The compactum is not hereditarily separable.
  For $0<a<1$, let $h_a$ be $0$ on $[0,a)$, $\frac12$ at $a$ and $1$ on $(a,1]$.
  The set $\{h\in M_{[0,1]}:\frac14<h(a)<\frac34\}$ is a neighborhood of $h_a$ that contains no other $h_b$, so the set $\{h_a:0<a<1\}$ is an uncountable discrete subspace.

  The observable is therefore NIP$^{\mathrm{fc}}$ but not NIP$^{\mathrm{hs}}$, and the last statement is \cref{thm:observable-classes}(ii).
\end{proof}

The idealized responses of the continuous comparison have a direct physical reading, which is also a probabilistic one.
They are described by all the functions $g$ in $M_{[0,1]}$, where $g(y)$ is the limiting probability that Alice’s value lies below Bob’s fixed value $y$.
Such a function need not be the distribution function of a single countably additive position measure:
its value at a jump also records how the particle approaches the jump point, from the left, from the right, or both.
The function $h_a$ describes a half-and-half mixture of a particle concentrating at $a$ from the left and one concentrating at $a$ from the right, and Bob’s idealized probe concentrated at $a$ reads the value $\frac12$.
In the discrete case no such splitting occurs at an eigenvalue, because each eigenvalue is isolated;
the additional limiting behavior is escape to infinite photon number, the point $1$ of $\Xi$, which may be mixed with an ordinary distribution on the eigenvalues.

%#
%### SUBSECTION: Unconfined quantities
%#
\subsection{Unconfined quantities}
Confinement to a box is not what produces this subspace of the Helly space:
a comparison depends only on the order of the two outcomes, so it is unchanged when both quantities are replaced by the same strictly increasing function of themselves, and a suitable function turns unconfined quantities, such as the positions of two particles on a line or the quadratures of two modes of light, into positions in a box
(see \cref{prop:cv-helly} below).

\begin{proposition}\label{prop:cv-helly}
  Let $A$ be a self-adjoint operator on the Hilbert space $H_1$ of Alice’s system and $B$ one on the Hilbert space $H_2$ of Bob’s, each unitarily equivalent to multiplication by $x$ on $L^2(\RR)$.
  Then the comparison $[A<B]$ is unstable, and its idealized responses form a copy of the space of nondecreasing functions $[0,1]\to[0,1]$ taking the value $0$ at $0$ and $1$ at $1$, a closed subspace of the Helly space.
  Hence the comparison is NIP$^{\mathrm{fc}}$ but not NIP$^{\mathrm{hs}}$.
\end{proposition}

\begin{proof}
  We may assume that $A$ and $B$ are multiplication by $x$ on $L^2(\RR)$.
  For physical quantities, $x$ denotes the numerical value in the fixed common unit used for the comparison.
  Let $\theta\colon\RR\to(0,1)$ be the increasing diffeomorphism $\theta(x)=\frac12(1+\tanh x)$.
  The operator $U$ given by $(U\psi)(u)=\psi(\theta^{-1}(u))\,|(\theta^{-1})'(u)|^{1/2}$ is a unitary from $L^2(\RR)$ onto $L^2[0,1]$ that carries multiplication by $\theta(x)$ to multiplication by $u$.
  Since $\theta$ is strictly increasing, $x<y$ exactly when $\theta(x)<\theta(y)$;
  hence the spectral projection $[A<B]$ equals $[\theta(A)<\theta(B)]$, and $U\otimes U$ carries it to the comparison of the continuous case.
  The unitary $U\otimes U$ preserves expectation values and carries actual and idealized states to actual and idealized states;
  it therefore induces a homeomorphism between the two compacta of idealized responses.
  The statement now follows from \cref{prop:order-responses} and \cref{thm:orders}(b).
\end{proof}

For the comparisons covered by \cref{prop:cv-helly}, confinement does not affect the compactum of idealized responses.
For quantities whose expectation values matter, such as energies, the rescaling by $\theta$ is not available, because it is not affine (see \cref{sec:logical-framework}).

%#
%### SUBSECTION: The causal order is not first countable
%#
\subsection{The causal order is not NIP$^{\mathrm{fc}}$}
The compactum of idealized responses of a comparison is first countable because its members are monotone, and a monotone function has at most countably many jumps, each with two sides (see the proof of \cref{thm:orders}(b)).
A relation between Bob’s value and two values of Alice’s, one on each side, removes this restriction (see \cref{thm:causal} below), and the natural example is the causal order of spacetime.
We work in a $(1+1)$-dimensional spacetime, with one dimension of time and one of space, in units in which the speed of light is $1$.
An event $(t,\xi)$, at time $t$ and position $\xi$, can influence a detection by Bob at the point
$y$ at time $0$ exactly when $|y-\xi|\le-t$, that is, when $y$ lies between the points $a=\xi+t$ and $b=\xi-t$ at which the two light rays through the event cross Bob’s time slice.
We describe Alice’s event by these null coordinates and confine everything to the causal diamond over Bob’s interval $[0,1]$, the set of events whose two light rays cross Bob’s slice inside the interval:
Alice’s Hilbert space is $H_1=L^2([0,1]^2)$, with variables $(a,b)$, Bob’s is $H_2=L^2[0,1]$, with variable $y$, and the \emph{causal-order observable} $C$ is multiplication by the indicator of
\[
  R=\{(a,b,y):a\le y\le b\}.
\]
Events with $a>b$ lie in the future of Bob’s slice and influence none of his detections.
Recall from \cref{sec:classification} that $D_5$ is the minimal family of the heptachotomy whose closure, $\hat{A}(2^\NN)$, is not first countable.

For a state $\rho$ of Alice’s system, let $\mu_\rho$ be the distribution of $(a,b)$ in $\rho$, an absolutely continuous probability measure on $[0,1]^2$, and similarly for Bob.
The observable $C$ is the relation observable $C_R$
for this relation, so, by \cref{lem:product}, its expectation value in states
$\rho$ and $\sigma$ of the two systems is $\varphi_C(\rho,\sigma)=(\mu_\rho\otimes\mu_\sigma)(R)=\int g_\rho\,d\mu_\sigma$, where $g_\rho(y)=\mu_\rho(\{(a,b):a\le y\le b\})$ is the probability that Alice’s event can influence a detection at $y$.
The function $g_\rho$ is continuous:
the sections of $R$ at $y$ and at $y'$ differ by a set contained in two strips, each of area at most $|y-y'|$, and the $\mu_\rho$-measure of such strips tends to $0$ with $|y-y'|$, because $\mu_\rho$ is absolutely continuous.
The effective observable $T_C(\rho)$ is multiplication by $g_\rho$.

\begin{keytheorem}\label{thm:causal}
  The causal-order observable is NIP but not NIP$^{\mathrm{fc}}$.
  Hence, for every trace-norm dense sequence $(\rho_n)_{n\in\NN}$
  of states of Alice’s system, there is an increasing injection
  $\iota\colon\{0,1\}^{<\NN}\to\NN$ such that the responses $\{\mathfrak{r}_{\rho_{\iota(t)}}:t\in\{0,1\}^{<\NN}\}$ are equivalent to $D_5$.
\end{keytheorem}

\begin{proof}
  For every state $\rho$, the event $\{a\le y\le b\}$ is the event $\{a\le y\}$ with the event $\{a\le y,\ b<y\}$ removed, so
  \[
    g_\rho(y)=\mu_\rho\bigl(\{(a,b):a\le y\}\bigr)-\mu_\rho\bigl(\{(a,b):a\le y,\ b<y\}\bigr),
  \]
  a difference of two nondecreasing functions with values in $[0,1]$.
  By \cref{prop:uniform-orders}, whose constant is here $1$, the observable $C$ is uniformly NIP, and in particular NIP.
  By \cref{lem:separable} the compactum of idealized responses is separable, so its level is defined.

  The probe algebra consists of operators $M_f$ of multiplication by continuous functions $f$;
  it is therefore a C$^*$-subalgebra of $\{M_f:f\in C[0,1]\}$, which is isometrically isomorphic to $C[0,1]$, because $\|M_f\|=\sup|f|$.
  Every state of $\fA_C$ extends to a state of this larger algebra
  (see~\cite[Theorem~4.3.13]{Kadison-Ringrose:1983a}), and, by the Riesz representation theorem, every state $\omega_2$ of $\fA_C$ is thus represented by a probability measure $\nu$ on $[0,1]$, in the sense that $\omega_2(M_f)=\int f\,d\nu$ for $M_f\in\fA_C$;
  in particular, $\mathfrak{r}_\rho(\omega_2)=\int g_\rho\,d\nu$.
  For each point $c\in(0,1)$, the evaluation $\omega_{2,c}(M_f)=f(c)$ is a state of $\fA_C$, Bob’s probe localized at $c$.

  Let $\rho_0$ be a state whose distribution is uniform on $[\frac12,1]\times[0,\frac12]$, an event after Bob’s slice.
  Then $g_{\rho_0}=0$, and the zero function is a response.
  For $c\in(0,1)$ and large $n$, let $\rho_{c,n}$ be a state whose distribution is uniform on $[c-\frac2n,c-\frac1n]\times[c+\frac1n,c+\frac2n]$, an event at position near $c$ shortly before Bob’s slice.
  Then $g_{\rho_{c,n}}(c)=1$, and $g_{\rho_{c,n}}(y)=0$ when $|y-c|>\frac2n$.
  Hence the function $g_{\rho_{c,n}}$ converges pointwise to the indicator of $\{c\}$, and, by dominated convergence, $\mathfrak{r}_{\rho_{c,n}}(\omega_2)\to\nu(\{c\})$ for every probe $\omega_2$ and every measure $\nu$ representing it.
  The limit $\mathfrak{r}_c$ is an idealized response, with $\mathfrak{r}_c(\omega_{2,c})=1$.

  Suppose that the zero function had a countable base of neighborhoods $(U_k)$ in the compactum of idealized responses.
  Each $U_k$ contains a set $\{\mathfrak{r}:|\mathfrak{r}(\omega_2)|<\varepsilon_k\text{ for }\omega_2\in F_k\}$ with $F_k$ finite.
  Choose a representing measure for each probe in $F=\bigcup_kF_k$;
  together these measures have countably many atoms, so some $c\in(0,1)$ is an atom of none of them.
  Then $\mathfrak{r}_c(\omega_2)=0$ for $\omega_2\in F$, so the response $\mathfrak{r}_c$ lies in every $U_k$.
  Since the space is Hausdorff, the intersection of a neighborhood base at a point is the singleton of the point;
  hence $\mathfrak{r}_c=0$, contrary to $\mathfrak{r}_c(\omega_{2,c})=1$.
  The compactum is therefore not first countable, and the last statement is \cref{thm:observable-classes}(iii).
\end{proof}

Physically, the proof says the following.
For $c\in(0,1)$, the idealized response $\mathfrak{r}_c$ is the limit of responses of preparations whose event coordinates concentrate near $(a,b)=(c,c)$ from $a<c<b$.
The corresponding light-cone intervals shrink to $c$, so only probes whose measures have an atom at $c$ detect the limiting response.
Every actual probe has an absolutely continuous position distribution and therefore gives $\mathfrak{r}_c$ the value $0$;
the idealized probe $\omega_{2,c}$ instead gives the value $1$.
There is one such idealized response for each point $c$ in this interval, and no countable family of probes distinguishes all of them from the zero response, as the neighborhood-base argument shows.
The relation $R$ is the intersection of two comparisons, $a\le y$ and $y\le b$, one for each light ray.
A single comparison has monotone idealized responses;
the two rays together allow idealized responses supported on a single point.
Here Todorčević’s trichotomy is used in full:
it turns the failure of first countability into a dyadic tree of actual preparations of Alice’s event coordinates whose responses are equivalent to the minimal family $D_5$.
The proof suggests a tree of this kind, namely preparations localized near events whose light cones meet Bob’s slice in the intervals of a Cantor construction, each event later than its predecessor, although we do not claim that the trichotomy produces this one.

%#
%### SUBSECTION: The comparison of phases is not first countable
%#
\subsection{The comparison of phases is not NIP$^{\mathrm{fc}}$}
Orders on a line give first countable compacta of idealized responses (see \cref{thm:orders} and \cref{prop:cv-helly});
on a circle, first countability can fail (see \cref{thm:phases} below).
Identify the circle $\mathbb{T}=\RR/\ZZ$ with $[0,1)$, so that angles are measured in turns and a half-turn is $\frac12$.
Alice and Bob each hold a particle on a ring, a planar rotor, with Hilbert spaces $H_1=H_2=L^2(\mathbb{T})$, the angle acting by multiplication.
The \emph{phase comparison} is the relation observable $C=C_R$ for
\[
  R=\bigl\{(x,y)\in[0,1)^2:\ y-x\in(0,\tfrac12)\bmod1\bigr\},
\]
the event “Bob’s angle is ahead of Alice’s by less than half a turn”.
For a state $\rho$ of Alice’s rotor, the distribution $\mu_\rho$ of her angle is an absolutely continuous probability measure on $\mathbb{T}$, and every such measure occurs.
By \cref{lem:product}, the phase comparison has, in states
$\rho$ and $\sigma$ of the two rotors, the expectation value $\varphi_C(\rho,\sigma)=\int g_\rho\,d\mu_\sigma$, where $g_\rho(y)=\mu_\rho\bigl((y-\frac12,y)\bigr)$ is the probability that Alice’s angle lies in the open half-circle ending at $y$;
the function $g_\rho$ is continuous on $\mathbb{T}$, because $\mu_\rho$ has no atoms.

\begin{lemma}\label{lem:phase-bv}
  For every probability measure $\mu$ on $\mathbb{T}=[0,1)$, the function $g(y)=\mu((y-\frac12,y))$ on $[0,1)$ is $2(u-v)$ for two nondecreasing functions $u,v\colon[0,1)\to[0,1]$.
\end{lemma}

\begin{proof}
  Let $G(y)=\mu([0,y))$.
  For $y\in[\frac12,1)$ the half-circle is the interval $(y-\frac12,y)$, and $g(y)=G(y)-\mu([0,y-\frac12])$;
  for $y\in[0,\frac12)$ it is $(y+\frac12,1)\cup[0,y)$, and $g(y)=G(y)+1-\mu([0,y+\frac12])$.
  Let $V(y)=\mu([0,y+\frac12])-1$ for $0\le y<\frac12$ and $V(y)=\mu([0,y-\frac12])$ for $\frac12\le y<1$.
  Then $g=G-V$.
  The functions $G$ and $V$ are nondecreasing, with values in $[0,1]$ and $[-1,1]$, respectively;
  indeed, the function $V$ is nondecreasing on each half, and at $\frac12$ it passes from $\lim_{y\uparrow1/2}V(y)=0$ to $V(\frac12)=\mu(\{0\})\ge0$.
  Let $u=\frac12(G+1)$ and $v=\frac12(V+1)$;
  then $g=2(u-v)$.
\end{proof}

\begin{keytheorem}\label{thm:phases}
  For every $\delta>0$, no set of more than $N(\delta/4)$ responses of Alice’s states for the phase comparison is $\delta$-shattered by points of its probe space, so the phase comparison is uniformly NIP;
  it is not NIP$^{\mathrm{fc}}$.
  Hence, for every trace-norm dense sequence $(\rho_n)_{n\in\NN}$ of states of Alice’s rotor, there is an increasing injection $\iota\colon\{0,1\}^{<\NN}\to\NN$ such that the responses $\{\mathfrak{r}_{\rho_{\iota(t)}}:t\in\{0,1\}^{<\NN}\}$ are equivalent to $D_5$.
\end{keytheorem}

\begin{proof}
  By \cref{lem:phase-bv}, the function $g_\rho$ agrees on $[0,1)$ with $2(u-v)$ for two nondecreasing functions $u,v\colon[0,1)\to[0,1]$, which we extend to nondecreasing functions on $\RR$, constant on $(-\infty,0)$ and on $[1,\infty)$.
  The function $g_\rho$ vanishes off $[0,1)$, where the extended difference $2(u-v)$ need not vanish, so the hypothesis of \cref{prop:uniform-orders} holds only on $[0,1)$.
  Bob’s distributions are carried by $[0,1)$, however, and only these values of $g_\rho$ enter the proof of the proposition;
  the proof therefore applies verbatim with the constant $2$, which gives the bound $N(\delta/4)$.

  As in the proof of \cref{thm:causal}, every probe
  $\omega_2$ is represented by a probability measure $\nu$ on $\mathbb{T}$, with $\mathfrak{r}_\rho(\omega_2)=\int g_\rho\,d\nu$, and for $z\in\mathbb{T}$ the evaluation $\omega_{2,z}(M_f)=f(z)$ is an idealized probe localized at~$z$.
  An actual state $\sigma$ instead induces the probe represented by its absolutely continuous angle distribution $\mu_\sigma$.
  Let $\rho_0$ be a state in which Alice’s angle is uniformly distributed.
  Every open half-circle has measure $\frac12$, so $\mathfrak{r}_{\rho_0}(\omega_2)=\frac12$ for every probe $\omega_2$.

  For $a\in\mathbb{T}$ and large $n$, let $\rho_{a,n}$ be a state whose distribution is uniform, with weight $\frac12$ each, on the arcs $\Delta_1=(a-\frac1n,a)$ and $\Delta_2=(a+\frac12,a+\frac12+\frac1n)$:
  half of the distribution lies just behind $a$, half just ahead of the antipode.
  Write $A_y=(y-\frac12,y)$.
  If $y=a$, then the half-circle $A_y$ contains $\Delta_1$ and $\Delta_2$, so $g_{\rho_{a,n}}(a)=1$.
  If $y=a+\frac12$, then the half-circle $A_y=(a,a+\frac12)$ meets neither, and $g_{\rho_{a,n}}(a+\frac12)=0$.
  Otherwise the endpoints of $A_y$ are at some positive distance $d$ from $a$ and from $a+\frac12$;
  an open half-circle whose endpoints differ from these two antipodal points contains exactly one of them.
  Once $\frac1n<d$, the half-circle $A_y$ therefore contains the arc $\Delta_1$ just behind $a$ if it contains $a$, and the arc $\Delta_2$ just ahead of $a+\frac12$ if it contains $a+\frac12$, but not both.
  Thus $g_{\rho_{a,n}}(y)=\frac12$.
  Hence the function $g_{\rho_{a,n}}$ converges pointwise to $g_a=\frac12+\frac12\one_{\{a\}}-\frac12\one_{\{a+1/2\}}$, and, by dominated convergence,
  $\mathfrak{r}_{\rho_{a,n}}(\omega_2)\to\frac12+\frac12\nu(\{a\})-\frac12\nu(\{a+\frac12\})$ for every probe $\omega_2$ and every measure $\nu$ representing it.
  The limit $\mathfrak{r}_a$ is an idealized response, with $\mathfrak{r}_a(\omega_{2,a})=1$.

  Suppose that $\mathfrak{r}_{\rho_0}$ had a countable base of neighborhoods $(U_k)$ in the compactum of idealized responses.
  Each $U_k$ contains a set $\{\mathfrak{r}:|\mathfrak{r}(\omega_2)-\frac12|<\varepsilon_k\text{ for }\omega_2\in F_k\}$ with $F_k$ finite.
  Choose a representing measure for each probe in $F=\bigcup_kF_k$;
  together these measures have countably many atoms, so there is $a\in\mathbb{T}$ such that neither $a$ nor $a+\frac12$ is one of them.
  Then $\mathfrak{r}_a(\omega_2)=\frac12$ for every $\omega_2\in F$, so $\mathfrak{r}_a$ lies in every $U_k$, and, the space being Hausdorff, the response $\mathfrak{r}_a$ equals $\mathfrak{r}_{\rho_0}$, contrary to $\mathfrak{r}_a(\omega_{2,a})=1$.
  The compactum is therefore not first countable, and the last statement is \cref{thm:observable-classes}(iii).
\end{proof}

Here, too, the proof has a physical reading.
A state in which the angle is uniformly distributed responds $\frac12$ to every probe.
For an antipodal pair of points $a$ and $a+\frac12$, the idealized limit of preparations that place half of the angle’s distribution just behind
$a$ and half just ahead of $a+\frac12$ also responds $\frac12$, except that a probe at $a$ reads $1$ and a probe at the antipode reads $0$:
the jumps at the two ends of the half-turn point in opposite directions and cancel everywhere else.
There is one such pair for each point of the circle, and countably many probes cannot distinguish all the resulting responses from that of the uniform state.
On a line this cannot happen, because a monotone response cannot hide a jump (see the proof of \cref{thm:orders}(b));
on a circle there is no first or last point.
The same contrast appears in topological dynamics.
The enveloping semigroup of a dynamical system is the pointwise closure of its acting maps.
Glasner and Megrelishvili showed that compact metrizable linearly ordered spaces, with their order topologies and order-preserving group actions, have first countable enveloping semigroups,
while the circular analogue of the Helly space is not first countable~\cite[Theorem~8.8, Proposition~8.10 and Remark~8.11]{Glasner-Megrelishvili:2022}.
The circular space in their comparison consists of circular-order-preserving maps of the circle, whereas the responses here are real-valued;
neither result implies the other.

%#
%### SUBSECTION: Interference fringes are not NIP
%#
\subsection{Interference fringes are not NIP}
The relations studied so far in this section are all NIP.
The fixed-contrast interference patterns in \cref{prop:fringes} below show how oscillation at ever finer scales can produce the independence property.
Positions on Bob’s screen are measured in a fixed unit of length, in which the screen is the interval $[0,2\pi]$, and wavenumbers are measured in the inverse unit.
We model the experiment by a register on Alice’s side with orthonormal basis $(e_n)_{n\in\NN}$, whose label $n$ selects a fringe pattern with wavenumber $k_n=8^n$.
At position $x\in[0,2\pi]$, Bob’s detector accepts with probability $f_n(x)=\frac12(1+\cos k_nx)$, a model of the local detector response that is largest on a bright fringe.
The \emph{fringe observable} is $C=\sum_n|e_n\rangle\langle e_n|\otimes M_{f_n}$ on $\ell^2(\NN)\otimes L^2([0,2\pi],\lambda)$, where $\lambda$ is Lebesgue measure divided by $2\pi$:
the label $n$ of Alice’s state selects the test $f_n$ applied to Bob’s detector, and $C$ is thus a label-controlled observable in the sense of \cref{def:label} below.

\begin{proposition}\label{prop:fringes}
  The fringe observable is not NIP:
  Alice’s basis states $e_n$ form an independent sequence, and the closure of their responses contains a copy of $\beta\NN$.
\end{proposition}

\begin{proof}
  Let $E,F\subseteq\NN$ be finite and disjoint.
  Choose $y\in[0,1)$ whose base-$8$ digit in place $n+1$ is $0$ for $n\in F$ and $4$ for $n\in E$, and let $x=2\pi y$.
  Then $\cos k_nx=\cos(2\pi\{8^ny\})$, where $\{8^ny\}$, the fractional part of $8^ny$, lies in $[0,\frac18)$ for $n\in F$ and in $[\frac12,\frac58)$ for $n\in E$;
  hence $f_n(x)\ge\frac34$ for $n\in F$ and $f_n(x)\le\frac14$ for $n\in E$.
  By continuity of the finitely many functions $f_n$ involved, a state $\sigma$ of Bob’s detector concentrated near $x$ has expectation values at least $\frac58$ for $n\in F$ and at most $\frac38$ for $n\in E$.
  By \cref{lem:probe-space}(c), the restrictions of such states are probes;
  the sequence $(e_n)$ is therefore independent.
  As in the proof of \cref{thm:levels}(c), no subsequence of the responses converges pointwise, and \cref{fact:rosenthal} yields a subsequence whose closure is homeomorphic to $\beta\NN$.
\end{proof}

The contrast between these fringes and the orders of this section displays the finite-pattern obstruction measured by NIP.
A comparison asks on which side of a threshold a value lies;
the pattern of wavenumber $8^n$ has $8^n$ periods on Bob’s screen, each with a bright and a dark half.
The base-$8$ construction lets one prescribe bright or dark responses independently for any finite collection of labels.
This is independence of finite patterns, rather than statistical independence of the fringe intensities under a distribution of screen positions.

%#
%## SECTION: The seven minimal families in physics
%#
\section{The seven minimal families in physics}\label{sec:hepta}

\noindent The classification of a full response compactum and the classification of a selected tree of preparations concern different compact sets.
For an NIP observable, the full compactum is convex, and hereditary separability implies metrizability by \cref{cor:nip2};
the closure of a selected tree need not be convex.
Each of the seven minimal families is nevertheless realized by such a tree of actual preparations (see \cref{thm:hepta-natural} below).
Photon-number comparisons, position comparisons and spin registers provide the realizations, while the passage from functions of position to responses on probability measures preserves the indexed families and their pointwise closures.

%#
%### SUBSECTION: Label-controlled observables
%#
\subsection{Label-controlled observables}
Several of our observables are of the following kind.
Alice’s system is a register that holds a label, and the label selects the test applied to Bob’s system:
a detector that accepts Bob’s system at each point with a probability that depends on the label and varies continuously with the point.

\begin{definition}\label{def:label}
  Let $L$ be a compact metrizable space, $\lambda$ a Borel probability measure on $L$ with full support, $I$ a countable set and $(f_i)_{i\in I}$ a family of continuous functions $L\to[0,1]$.
  The \emph{label-controlled observable} is
  \[
    C=\sum_{i\in I}|e_i\rangle\langle e_i|\otimes M_{f_i}
    \qquad\text{on }\ \ell^2(I)\otimes L^2(L,\lambda),
  \]
  where $(e_i)$ is the standard basis of $\ell^2(I)$ and $M_f$ is multiplication by $f$.
\end{definition}

The sum converges strongly, that is, on each vector, because its terms act on orthogonal subspaces, and the operator $C$ lies between $0$ and the identity operator.
When $L$ contains a countable dense set $Z$ of isolated points, we take $\lambda=\sum_{x\in Z}w_x\delta_x$, a combination of the point masses $\delta_x$ with positive weights $w_x$ of sum $1$;
then the space $L^2(L,\lambda)$ is a weighted copy of $\ell^2(Z)$, and the position of Bob’s system takes its values in $Z$.
As in \cref{sec:orders}, we write $C(L)$ for the algebra of continuous complex functions on $L$, with the supremum norm, and $P(L)$ for the space of Borel probability measures on $L$, with the weak$^*$ topology.
Recall that the probe algebra $\fA_C$ of an observable $C$ is the unital C$^*$-algebra generated by its effective observables, and that the probe space $Y_C$ is the state space of $\fA_C$.

\begin{lemma}\label{lem:label}
  Let $C$ be as in \cref{def:label}, and suppose that the functions $f_i$ separate the points of $L$.
  Then the probe algebra $\fA_C$ is
  $\{M_f:f\in C(L)\}$, the probe space $Y_C$ is identified with $P(L)$, and the response of the label state $e_i$ is the function $\mathfrak{r}_{e_i}\colon P(L)\to[0,1]$ given by $\mathfrak{r}_{e_i}(\nu)=\int f_i\,d\nu$.
  The record $\mathfrak{p}_\sigma$ of an actual state $\sigma$ of Bob’s system corresponds to the measure $\mu_\sigma$ characterized by $\Tr(\sigma M_f)=\int f\,d\mu_\sigma$ for $f\in C(L)$.
  The measures obtained in this way are exactly those absolutely continuous with respect to $\lambda$.
\end{lemma}

\begin{proof}
  For a density operator $\rho$, tracing out Alice’s system yields $T_C(\rho)=\sum_i\langle e_i,\rho e_i\rangle M_{f_i}$, a norm-convergent sum, so the slice map takes values in $\{M_f:f\in C(L)\}$.
  Since $\lambda$ has full support and $f$ is continuous, $\|M_f\|=\sup|f|$;
  hence the map $f\mapsto M_f$ is an isometric $*$-isomorphism from $C(L)$ onto its image.
  The unital C$^*$-subalgebra of $C(L)$ generated by the functions $f_i$ is $C(L)$, by the Stone--Weierstrass theorem, because they separate points.
  It follows that $\fA_C=\{M_f:f\in C(L)\}$, and, by the Riesz representation theorem, its states are the measures $\nu\in P(L)$, acting by $M_f\mapsto\int f\,d\nu$.
  Finally, if $\sigma=\sum_j\lambda_j|\psi_j\rangle\langle\psi_j|$ is a density operator on $L^2(L,\lambda)$, with weights $\lambda_j\ge0$ of sum $1$ and unit vectors $\psi_j$, then $\Tr(\sigma M_f)=\int f\,\bigl(\sum_j\lambda_j|\psi_j|^2\bigr)\,d\lambda$, so the measure $\mu_\sigma$ has density $\sum_j\lambda_j|\psi_j|^2$ with respect to $\lambda$;
  conversely, a probability measure with density $p$ with respect to $\lambda$ is $\mu_\sigma$ for the pure state given by the unit vector $\sqrt{p}$.
\end{proof}

%#
%### SUBSECTION: Transfer from points to measures
%#
\subsection{Transfer from points to measures}
For a compact metrizable space $L$, a continuous function $g\colon L\to[0,1]$ determines a response $\mathfrak{r}_g\colon P(L)\to[0,1]$ by $\mathfrak{r}_g(\nu)=\int_L g\,d\nu$.
To transfer the minimal families from functions on $L$ to responses on the probe space, we must preserve their pointwise closures under this map.
\Cref{lem:transfer} below does so for countable families whose pointwise closures consist of Baire-1 functions;
the resulting sequential properties allow us to use dominated convergence and recover each limiting response from its values on the point masses.

\begin{lemma}\label{lem:transfer}
  Let $L$ be a compact metrizable space and $(f_t)_{t\in T}$ a countable family of continuous functions $L\to[0,1]$ that is relatively compact in the space $\Bone(L)$ of Baire-1 functions on $L$.
  Then the families $\{\mathfrak{r}_{f_t}:t\in T\}$, of functions on $P(L)$, and $\{f_t:t\in T\}$, of functions on $L$, are equivalent in the sense of \cref{def:equivalent}, and the family $\{\mathfrak{r}_{f_t}:t\in T\}$ is NIP.
\end{lemma}

\begin{proof}
  Write $\mathfrak{r}_t=\mathfrak{r}_{f_t}$ for $t\in T$.
  Every sequence $(f_{t_n})$ has a subsequence that converges pointwise on $L$.
  Indeed, otherwise, by \cref{fact:rosenthal}, some subsequence would have a closure homeomorphic to
  $\beta\NN$, which is not Fréchet--Urysohn (a space is Fréchet--Urysohn if every point in the closure of a subset is the limit of a sequence from that subset),
  whereas the closure of $\{f_t\}$ is a Rosenthal compactum and therefore Fréchet--Urysohn (see~\cite[Theorem~3F]{Bourgain-Fremlin-Talagrand:1978}).
  By the dominated convergence theorem, the corresponding subsequence of $(\mathfrak{r}_{t_n})$ converges pointwise on $P(L)$.

  The family $\{\mathfrak{r}_t\}$ is thus NIP.
  Indeed, if a sequence $(\mathfrak{r}_{t_n})$ and reals $a<b$ witnessed the independence property, then, since
  $P(L)$ is compact and the functions $\mathfrak{r}_t$ are continuous, for each $E\subseteq\NN$ there would be $\nu$ with $\mathfrak{r}_{t_n}(\nu)\le a$ for $n\in E$ and $\mathfrak{r}_{t_n}(\nu)\ge b$ for $n\notin E$ (see the proof of \cref{thm:dichotomy});
  taking for $E$ every other index of a pointwise convergent subsequence would yield a $\nu$ at which that subsequence does not converge, contrary to the previous paragraph.
  By \cref{thm:dichotomy}, the closure of $\{\mathfrak{r}_t\}$ is contained in $\Bone(P(L))$;
  it is a Rosenthal compactum, hence Fréchet--Urysohn.

  Let $\Phi$ send a function $G$ on $P(L)$ to the function $x\mapsto G(\delta_x)$ on $L$.
  The map $\Phi$ is continuous for the pointwise topologies, and $\Phi(\mathfrak{r}_t)=f_t$.
  The map $\Phi$ is injective on the closure of $\{\mathfrak{r}_t\}$:
  if $G$ lies in that closure, then, the closure being Fréchet--Urysohn, the function $G$ is the pointwise limit of a sequence
  $(\mathfrak{r}_{t_n})$;
  the sequence $(f_{t_n})$ then converges pointwise to $\Phi(G)$, which is therefore a bounded Borel function, and, by dominated convergence, the value $G(\nu)$ equals $\lim_n\int f_{t_n}\,d\nu=\int\Phi(G)\,d\nu$ for every $\nu$.
  Thus $G$ is determined by $\Phi(G)$.
  As a continuous injection on a compact space, the map $\Phi$ is a homeomorphism of the closure of $\{\mathfrak{r}_t\}$ onto its image;
  this image is compact, contains $\{f_t\}$ and lies in the closure of $\{f_t\}$, so it is that closure.
\end{proof}

We also need to discard points that carry no information about limits.

\begin{lemma}\label{lem:restrict}
  Let $L$ be a set, $L_0\subseteq L$, and $(g_t)_{t\in T}$ a countable pointwise bounded family of real functions on $L$ whose pointwise closure $K$ is Fréchet--Urysohn.
  Suppose that
  \begin{enumerate}[label=\textup{(\roman*)}]
  \item
    the restrictions $g_t|_{L_0}$ are pairwise distinct, and none of them lies in the closure of the others;
  \item
    if two limits of sequences $(g_{t_k})_{k\in\NN}$ with pairwise distinct indices agree on $L_0$, then they are equal.
  \end{enumerate}
  Then the families $(g_t)$ and $(g_t|_{L_0})$ are equivalent.
\end{lemma}

\begin{proof}
  Restriction $\pi\colon K\to\RR^{L_0}$ is continuous, and its image is the closure of $(g_t|_{L_0})$.
  Let $f\in K$ be different from every $g_t$.
  Since $K$ is Fréchet--Urysohn, the function $f$ is the limit of a sequence $(g_{t_k})$;
  since $K$ is Hausdorff and $f$ is not a $g_t$, no index occurs infinitely often, and, after passage to a subsequence, the indices are pairwise distinct.
  By~(i), the map $\pi$ is injective on the family;
  moreover, $\pi(g_t)\ne\pi(f)$ for every such limit $f$ and every $t$, because $\pi(f)$ lies in the closure of $\{g_s|_{L_0}:s\ne t\}$, whereas $g_t|_{L_0}$ does not.
  By~(ii), the map $\pi$ is also injective on the limits.
  Hence $\pi$ is a continuous injection of the compact space $K$, and therefore a homeomorphism onto its image.
\end{proof}

%#
%### SUBSECTION: Three settings
%#
\subsection{Three settings}
The realizations use three physical settings.

\emph{Photon numbers.}
This is the discrete comparison of \cref{sec:orders}, in which Alice’s and Bob’s values lie in $\Xi=\{x_l:l\in\NN\}\cup\{1\}$, where $x_l=1-\frac1{l+1}$.
For Alice’s number state $|l\rangle$, the state with exactly $l$ photons, the function $g_l=g_{|l\rangle}$, whose value at a point $y$ is the probability that Alice’s value lies below $y$, is the indicator of $\{x_j:j>l\}\cup\{1\}$;
the probe space is the space $P(\Xi)$ of Borel probability measures on $\Xi$, and $\mathfrak{r}_{|l\rangle}(\nu)=\int g_l\,d\nu$ (see \cref{prop:order-responses}).

\emph{Positions in a box.}
This is the continuous comparison of \cref{sec:orders}:
the probe space is $P([0,1])$, and a state $\rho$ of Alice’s particle with distribution $\mu_\rho$ has response $\nu\mapsto\int g_\rho\,d\nu$, where $g_\rho(y)=\mu_\rho([0,y))$.
We place a Cantor set in the box.
Let $\psi(\sigma)=\frac34-\sum_n\sigma(n)3^{-n-1}$ for $\sigma\in\{0,1\}^\NN$;
this is an order-reversing homeomorphism of the Cantor space $\{0,1\}^\NN$, with its lexicographic order $\prec$, onto a Cantor set $C_0\subseteq[\frac14,\frac34]$.
For a finite sequence $s\in\{0,1\}^{<\NN}$, let $I_s$ be the smallest interval containing the image $\psi(V_s)$ of the clopen set $V_s$ of infinite extensions of $s$;
it has length $\ell_s=\frac123^{-|s|}$, left endpoint $\psi(s^\frown1^\infty)$ and right endpoint $\psi(s^\frown0^\infty)$, the images of the extensions of $s$ by ones and by zeros.
If $\zeta$ is the immediate successor of $\sigma$ in $\prec$, then the interval $(\psi(\zeta),\psi(\sigma))$ is a component of $[0,1]\setminus C_0$;
the other components are $[0,\frac14)$ and $(\frac34,1]$.
The components adjacent to $I_s$ have length at least $\ell_s/2$:
each of them is either $[0,\frac14)$ or $(\frac34,1]$, of length $\frac14\ge\ell_s/2$, or the gap between $I_{u^\frown0}$ and $I_{u^\frown1}$ for some proper initial segment $u$ of $s$, of length $\ell_u/3\ge\ell_s$.

\emph{Registers.}
Alice’s system is a chain of spins of variable finite length, with Hilbert space $\ell^2(\{0,1\}^{<\NN})\cong\bigoplus_{n\in\NN}(\CC^2)^{\otimes n}$.
For $t\in\{0,1\}^{<\NN}$, the basis ket $|t\rangle=\bigotimes_{i<|t|}|t(i)\rangle$ has spin up at site $i$ exactly when $t(i)=0$, with $|0\rangle$ denoting up and $|1\rangle$ down.
Bob probes it in two ways.
In \emph{scanning}, Bob’s probe sits at a site $i\in\NN$ and asks whether Alice’s spin at that site is down, the outcome selected by $|1\rangle\langle1|$:
the test is $a_t(i)=t(i)$ for $i<|t|$ and $a_t(i)=0$ otherwise.
In \emph{template matching}, Bob holds a reference configuration $u$, finite or infinite, and asks whether Alice’s chain coincides with the beginning of his reference:
the test is $b_t(u)=1$ if $t\sqsubseteq u$ and $b_t(u)=0$ otherwise.
We extend $a_t$ to the one-point compactification
$L_a=\NN\cup\{\infty\}$ by $a_t(\infty)=0$, and we let $L_b$ be the space of finite and infinite sequences of zeros and ones, with the ultrametric $d(u,u')=2^{-|u\wedge u'|}$ for $u\ne u'$, where $u\wedge u'$ is the longest common initial segment;
it is compact, the finite sequences are isolated and dense, and the functions $b_t$ are continuous, since $\{u:t\sqsubseteq u\}$ is clopen.
The \emph{scanning}, \emph{template} and \emph{combined} observables are the label-controlled observables with labels in
$\{0,1\}^{<\NN}$, on $L_a$, $L_b$ and their disjoint sum
$L_a\oplus L_b$, with the functions $a_t$, $b_t$ and $(a_t,b_t)$, the last equal to $a_t$ on $L_a$ and to $b_t$ on $L_b$, and with measures concentrated on $\NN$, on the finite sequences, and on both, respectively.
These functions separate the points of $L_a$, of $L_b$ and of $L_a\oplus L_b$;
for instance, the functions of the empty label, $a_\emptyset=0$ and $b_\emptyset=1$, separate $L_a$ from $L_b$.
Hence \cref{lem:label} applies to all three observables.

%#
%### SUBSECTION: The realization
%#
\subsection{The realization}
Together, the three settings realize all seven minimal families.
In the next theorem, $(s_t)_{t\in\{0,1\}^{<\NN}}$ is the canonical enumeration of the regular dyadic subtree
$\Theta$ fixed in \cref{sec:classification}, and the chains $|t\rangle$, the Cantor set $C_0$, the homeomorphism $\psi$ and the intervals $I_s$, of lengths $\ell_s$, are those of the three settings.

\begin{keytheorem}\label{thm:hepta-natural}
  Each minimal family is realized by a tree of actual preparations of one of the observables of the three settings:
  for each $i=1,\dots,7$, the observable and the pure states $(\rho_t)_{t\in\{0,1\}^{<\NN}}$ of Alice’s system listed below, given by pairwise orthogonal vectors, have responses $\{\mathfrak{r}_{\rho_t}\}$ equivalent to $D_i$.
  \begin{enumerate}[label=\textup{(\arabic*)},leftmargin=*]
  \item[$D_1$:]
    the comparison of photon numbers, with Alice’s number states in any order;
  \item[$D_2$:]
    scanning, with Alice’s chains $|s_t\rangle$, whose scanning functions are $s_t^\frown0^\infty$, for $s_t\in\Theta$;
  \item[$D_3$:]
    the comparison of positions, with Alice’s particle uniformly distributed over the gap of $C_0$ immediately to the right of $\psi(s_t^\frown0^\infty)$;
  \item[$D_4$:]
    the comparison of positions, with Alice’s particle uniformly distributed over the gap of $C_0$ immediately to the left of $\psi(s_t^\frown1^\infty)$;
  \item[$D_5$:]
    template matching, with Alice’s chains $|t\rangle$;
  \item[$D_6$:]
    the combined observable, with Alice’s chains $|t\rangle$;
  \item[$D_7$:]
    the comparison of positions, with Alice’s particle distributed uniformly, half over $[\min I_{s_t}-\frac\ell2,\min I_{s_t}-\frac\ell4]$ and half over $[\max I_{s_t}+\frac\ell4,\max I_{s_t}+\frac\ell2]$, where $\ell=\ell_{s_t}$.
  \end{enumerate}
  Moreover, the causal order of spacetime realizes $D_5$ in the sense of \cref{thm:causal}.
\end{keytheorem}

\begin{proof}
  In each case the probe space is $P(L)$ for a compact metrizable space $L$, the responses are $\mathfrak{r}_{\rho_t}(\nu)=\int g_t\,d\nu$ with $g_t$ continuous, and the closure of $\{g_t\}$ consists of Baire-1 functions;
  by \cref{lem:transfer}, it suffices to show that $\{g_t\}$ is equivalent to $D_i$.
  We use the notation of \cref{sec:classification}:
  the function $v_t$ is the characteristic function of the set
  $V_t$ of infinite sequences that extend $t$, and the functions $x^+_\sigma$ and $x^-_\sigma$ are the characteristic functions of $\{\zeta\in\{0,1\}^\NN:\sigma\preceq\zeta\}$ and $\{\zeta\in\{0,1\}^\NN:\sigma\prec\zeta\}$, for the lexicographic order $\prec$.

  $D_1$.
  Index the functions $g_l$, $l\in\NN$, by $\{0,1\}^{<\NN}$ through any bijection, and write $g_t$ for the function of index $t$.
  Every sequence of pairwise distinct members converges to the indicator of $\{1\}$, which is not a member, and each member is isolated.
  Hence the closure is a countable compact space with a single non-isolated point, and so is the closure $D_1\cup\{0\}$ of $D_1$;
  the bijection $g_t\mapsto\frac1{|t|+1}v_t$, extended by sending the indicator of $\{1\}$ to $0$, is a homeomorphism.
  The space $\Xi$ is countable, so every function on it is Baire-1, every subset of a countable space being an $F_\sigma$.

  $D_2$.
  Here $g_t=a_{s_t}$ is $s_t^\frown0^\infty$ on $\NN$ and $0$ at $\infty$.
  This is the family $D_2$, extended continuously to $\NN\cup\{\infty\}$;
  every function in its closure vanishes at $\infty$, so the extension does not change the closure, and every function on the countable space $L_a$ is Baire-1.

  $D_3$ and $D_4$.
  Let $\sigma_t=s_t^\frown0^\infty$.
  If $\sigma_t\ne0^\infty$, it has an immediate predecessor $\sigma_t^-$, and $G_t=(\psi(\sigma_t),\psi(\sigma_t^-))$ is a gap;
  if $\sigma_t=0^\infty$, let $G_t=(\frac34,1]$.
  Let $\rho_t$ be uniform on $G_t$.
  Then $g_t=0$ on $[0,\psi(\sigma_t)]$, $g_t=1$ to the right of $G_t$, and $g_t$ is linear on $G_t$.
  On $C_0$, the function $g_t$ satisfies $g_t(\psi(\zeta))=0$ exactly when $\zeta\succeq\sigma_t$;
  that is, $g_t\circ\psi=1-x^+_{\sigma_t}$.
  The sequences $\sigma_t$ are pairwise distinct, and so are the gaps $G_t$.
  We apply \cref{lem:restrict} with $L_0=C_0$.
  The closure of $\{g_t\}$ consists of nondecreasing functions;
  it is therefore a Rosenthal compactum, and hence Fréchet--Urysohn.
  Condition~(i) holds because $D_3$ is discrete in its closure.
  For~(ii), let $J$ be a component of $[0,1]\setminus C_0$.
  For pairwise distinct $t_k$, eventually $G_{t_k}\ne J$, and then the function $g_{t_k}$ is constant on $\overline{J}$;
  every such limit $f$ is thus constant on $\overline{J}$, and so is determined on $J$ by its value at an endpoint of $J$ in $C_0$.
  Consequently, the family $\{g_t\}$ is equivalent to $\{g_t|_{C_0}\}$, which is carried to $D_3$ by the homeomorphism $\psi$ followed by the map $f\mapsto1-f$, a homeomorphism for the pointwise topology.
  The case of $D_4$ is symmetric:
  with $\sigma_t=s_t^\frown1^\infty$ and its immediate successor $\sigma_t^+$, the gap is $(\psi(\sigma_t^+),\psi(\sigma_t))$, or $[0,\frac14)$ if $\sigma_t=1^\infty$, and $g_t\circ\psi=1-x^-_{\sigma_t}$.

  $D_5$.
  Here $g_t=b_t$ on $L_b$, and its restriction to the infinite sequences is $v_t$.
  Let $K$ be the pointwise closure of $\{b_t\}$.
  For a finite sequence $u$, only finitely many functions $b_t$ are nonzero at $u$;
  hence any member of $K$ that is nonzero at $u$ belongs to this finite set of functions.
  Every member of $K\setminus\{b_t:t\in\{0,1\}^{<\NN}\}$ therefore vanishes on the finite sequences.
  The restriction to the closed set $\{0,1\}^\NN$ of every member of $K$ lies in $\hat{A}(2^\NN)\subseteq\Bone(\{0,1\}^\NN)$.
  The inverse image of an open set under such a member is therefore an $F_\sigma$ subset of $\{0,1\}^\NN$, together with at most countably many isolated points of $L_b$, and is thus $F_\sigma$ in $L_b$.
  By \cref{fact:baire}(2), each member of $K$ is Baire-1 on $L_b$.
  Thus $K$ is a Rosenthal compactum and is Fréchet--Urysohn.
  We may now apply \cref{lem:restrict} with $L_0=\{0,1\}^\NN$:
  condition~(i) holds because $D_5$ is discrete in its closure, and condition~(ii) holds because limits of sequences with pairwise distinct indices vanish at every finite sequence.
  The family $\{b_t\}$ is consequently equivalent to $D_5$.

  $D_6$.
  Here $g_t=(a_t,b_t)$ on $L_a\oplus L_b$.
  Every function on the countable compact metrizable space $L_a$ is Baire-1, while every member of the pointwise closure of $\{b_t\}$ is Baire-1 on $L_b$, as shown for $D_5$.
  The pointwise closure of $\{g_t\}$ therefore consists of Baire-1 functions on the disjoint sum and is a Rosenthal compactum, hence Fréchet--Urysohn.
  On $L_0=\NN\oplus\{0,1\}^\NN$, the function $g_t$ is $(t^\frown0^\infty,v_t)$, the corresponding member of $D_6$ with its two summands interchanged.
  These restrictions are discrete in their closure, and every limit of a sequence with pairwise distinct indices vanishes at $\infty$ and at each finite sequence in $L_b$.
  Thus both conditions of \cref{lem:restrict} hold, and the family $\{g_t\}$ is equivalent to $D_6$.

  $D_7$.
  Let $\ell=\ell_{s_t}$, and let $\rho_t$ be the state whose distribution is half uniform on $A_t=[\min I_{s_t}-\frac\ell2,\min I_{s_t}-\frac\ell4]$ and half uniform on $B_t=[\max I_{s_t}+\frac\ell4,\max I_{s_t}+\frac\ell2]$;
  these intervals lie in the components adjacent to $I_{s_t}$.
  On $C_0$, the function $g_t$ equals $0$ to the left of $I_{s_t}$, $\frac12$ on $I_{s_t}$ and $1$ to the right;
  that is, $g_t\circ\psi=1-x^+_{\sigma_t}+\frac12v_{s_t}$ with $\sigma_t=s_t^\frown0^\infty$, since $\psi(\sigma_t)=\max I_{s_t}$.
  Each member of the closure of $D_7$ has the form $(v,x)$, for functions $v$ and $x$ on $\{0,1\}^\NN$ with $v\le x$ that take only the values $0$ and $1$;
  hence the map $(v,x)\mapsto(1-x+\frac12v)\circ\psi^{-1}$ is a continuous injection, and it carries $D_7$ to $\{g_t|_{C_0}\}$;
  as $D_7$ is discrete in its closure, condition~(i) of \cref{lem:restrict} holds.
  For~(ii), note that $A_t$ lies in the component whose right endpoint is $\psi(s_t^\frown1^\infty)$, and $B_t$ in the component whose left endpoint is $\psi(s_t^\frown0^\infty)$.
  These sequences are pairwise distinct for $s_t\in\Theta$, so each component of $[0,1]\setminus C_0$ contains at most one interval $A_t$ and at most one interval $B_t$.
  For pairwise distinct $t_k$, eventually neither $A_{t_k}$ nor $B_{t_k}$ lies in a given component $J$, and then the function $g_{t_k}$ is constant on $\overline{J}$.
  Every such limit is therefore constant on $\overline{J}$ and is determined by its values on $C_0$.

  Finally, the photon and register preparations are given by distinct basis vectors.
  In the box, take the unit vectors $\sqrt{d\mu_{\rho_t}/dy}$.
  Within each of the $D_3$ and $D_4$ realizations, the supports lie in distinct Cantor gaps.
  For $D_7$, each gap contains at most one interval $A_t$ and at most one interval $B_u$, by the distinctness of the endpoint sequences fixed in \cref{sec:classification}.
  If an internal gap contains both, its length is at least each of $\ell_{s_t}$ and $\ell_{s_u}$;
  the interval $A_t$ lies in its right half and $B_u$ in its left half, so they meet in at most the midpoint.
  Each exterior gap contains at most one support interval.
  Within each realization, the supports of distinct preparations consequently meet in at most finitely many points, whence the vectors are pairwise orthogonal in $L^2[0,1]$.
\end{proof}

The realizations display what distinguishes the families.
\begin{itemize}[leftmargin=*]
\item
  $D_1$, \emph{escape to infinity}:
  an oscillator excited ever higher has a single idealized limit, and each actual preparation is isolated from it.
\item
  $D_2$, \emph{infinite configurations}:
  the idealized limits of the chains are the infinite spin configurations, a Cantor set of them, and countably many site probes tell them all apart, so the closure is metrizable;
  but any two configurations differ by $1$ at some site, so the closure is not separable in the supremum norm.
\item
  $D_3$ and $D_4$, \emph{one-sided localization}:
  a particle that approaches a point from only one side produces one half of a split pair, and the split Cantor set appears;
  the closure is hereditarily separable but not metrizable.
\item
  $D_5$, \emph{point-supported limits}:
  light cones that shrink to a point on Bob’s slice, or templates following longer and longer initial segments of one infinite configuration, yield a point-supported limiting profile;
  the corresponding idealized response evaluates the mass that a probe assigns to the point, and the closure of the selected responses is not first countable.
\item
  $D_6$, \emph{twins over a metrizable base}:
  each infinite configuration has two idealized versions, one that every template misses and one that the template of the configuration itself detects;
  the scanning probes alone cannot tell the twins apart, and, as seen by these probes, the limits form a Cantor set.
\item
  $D_7$, \emph{twins over a split base}:
  a particle that approaches a point from both sides at once produces the midpoint of a split pair, as the symmetric state does in the Ising model (see \cref{sec:ssb} below);
  the twins of the closure sit over the split Cantor set.
\end{itemize}

The scanning and template observables are NIP (see \cref{prop:label-nip} and \cref{cor:registers} below), which places them in the classification of \cref{thm:observable-classes};
the combined observable is NIP by \cref{prop:label-nip} as well, since its functions separate points and every sequence of them has a subsequence that converges pointwise on each of the two summands (see the proof of \cref{cor:registers}).
The proof of \cref{prop:label-nip} uses a theorem of Stegall.
Recall that a subset of a Banach space is \emph{weakly precompact} if it is bounded and every sequence in it has a weakly Cauchy subsequence,
one whose values under each continuous linear functional form a Cauchy sequence, and that the \emph{absolutely convex hull} of a set consists of the finite combinations
$\sum_k a_kx_k$ of its elements with $\sum_k|a_k|\le1$.
Stegall’s theorem states that the absolutely convex hull of a weakly precompact set is weakly precompact (see Rodríguez~\cite{Rodriguez:2021} for a quantitative version).
We say that an observable is \emph{NIP for Bob’s responses} if the observable obtained from it by exchanging the two tensor factors is NIP for Alice’s responses.

\begin{proposition}\label{prop:label-nip}
  Let $C$ be a label-controlled observable whose functions $f_i$ separate the points of $L$ and have the property that every sequence of them has a subsequence converging pointwise on $L$.
  Then the observable $C$ is NIP for Alice’s responses.
\end{proposition}

\begin{proof}
  For a density operator $\rho$, let $p_i=\langle e_i,\rho e_i\rangle$, a probability distribution on $I$, and $g_\rho=\sum_ip_if_i$, a norm-convergent sum in $C(L)$;
  then $\mathfrak{r}_\rho(\nu)=\int g_\rho\,d\nu$.
  By dominated convergence, a sequence of functions $f_i$ that converges pointwise on $L$ is weakly Cauchy in $C(L)$;
  hence the set $\{f_i\}$ is weakly precompact.
  By Stegall’s theorem, so is its absolutely convex hull, and so is the norm closure of that hull:
  given $y_n$ in the closure, choose $x_n$ in the hull with $\|x_n-y_n\|<1/n$ and a weakly Cauchy subsequence of $(x_n)$;
  the corresponding subsequence of $(y_n)$ is weakly Cauchy.
  The partial sums of $g_\rho$ lie in the hull, and every function $g_\rho$ therefore lies in this closure.
  Consequently, every sequence $(g_{\rho_n})$ has a weakly Cauchy subsequence, and, since every $\nu\in P(L)$ is a continuous linear functional on $C(L)$, the corresponding subsequence of $(\mathfrak{r}_{\rho_n})$ converges pointwise on $P(L)=Y_C$.
  As in the proof of \cref{lem:transfer}, an independent sequence has no such subsequence.
\end{proof}

\begin{corollary}\label{cor:registers}
  \leavevmode
  \begin{enumerate}[label=\textup{(\alph*)}]
  \item
    The scanning observable is unstable and NIP$^{\mathrm{met}}$, but not uniformly NIP.
  \item
    The template observable is NIP but not NIP$^{\mathrm{fc}}$.
  \item
    The scanning observable is not NIP for Bob’s responses.
    Hence NIP for Alice’s responses does not imply NIP for Bob’s responses.
  \end{enumerate}
\end{corollary}

\begin{proof}
  (a) Every sequence of functions $a_t$ on the countable set $L_a$ has a pointwise convergent subsequence, by a diagonal argument, so the observable is NIP by \cref{prop:label-nip}.
  For a trace-norm dense sequence $(\rho_n)_{n\in\NN}$, the closure of the responses is that of $\{\mathfrak{r}_{\rho_n}\}$ (see \cref{lem:separable}).
  Every function on the countable compact metrizable space $L_a$ is Baire-1, so \cref{lem:transfer} applies and identifies this closure with the closure of $\{g_{\rho_n}\}$ in $\RR^{L_a}$;
  since $L_a$ is countable, this space is metrizable.
  For each positive integer $m$, the chains $t_1,\dots,t_m$ of length $2^m$ with $t_j(i)$ the $j$-th binary digit of $i$ are shattered, with values $0$ and $1$, by the site probes $i<2^m$;
  hence the observable is not uniformly NIP.
  Dually, on the first $m$ site probes, the chains of length $m$ realize every binary pattern;
  the fat-shattering dimension of the response family is therefore infinite at every positive gap at most $1$, with the gap convention of \cref{rem:uniform}.
  Finally, the observable is unstable:
  the functions $a_{1^k}$ of the chains $1^k$ of $k$
  down spins, for $k\in\NN$, converge pointwise to the indicator of $\NN$, so the idealized response $\nu\mapsto\nu(\NN)$, which is not continuous on $P(L_a)$, is a limit of responses (see \cref{thm:levels}(a)).

  (b) Let $(t_k)_{k\in\NN}$ be a sequence of labels in $\{0,1\}^{<\NN}$.
  By compactness of $L_b$, there is a strictly increasing map $j\colon\NN\to\NN$ such that $(t_{j(k)})_{k\in\NN}$ converges in $L_b$.
  If its limit is finite, this subsequence is eventually constant, so the corresponding functions converge pointwise.
  If its limit is an infinite sequence $\sigma\in\{0,1\}^\NN$, then the values satisfy $b_{t_{j(k)}}(u)\to0$ for each finite sequence $u$, and $b_{t_{j(k)}}(\zeta)\to0$ for each infinite sequence $\zeta\ne\sigma$:
  eventually $t_{j(k)}$ agrees with $\sigma$ beyond the first place where $\sigma$ and $\zeta$ differ.
  Passing to a further subsequence on which $b_{t_{j(k)}}(\sigma)$ is constant gives pointwise convergence on all of $L_b$.
  Hence the observable is NIP by \cref{prop:label-nip}.
  By \cref{thm:hepta-natural}, the closure of its responses contains a closed copy of $\hat{A}(2^\NN)$, which is not first countable.

  (c) For $i\in\NN$, consider the state of Bob’s system in which his probe sits at the site $i$;
  these states form an independent sequence:
  for disjoint finite sets $E,F\subseteq\NN$, choose a positive integer $q$ larger than every member of $E\cup F$ and let Alice’s chain $t$ of length $q$ have $t(i)=1$ exactly for $i\in F$.
  Then $a_t(i)=1$ for $i\in F$ and $a_t(i)=0$ for $i\in E$.
  Alice’s actual states restrict to points of the probe space of the exchanged observable (see \cref{lem:probe-space}(c)), so the family of Bob’s responses has the independence property.
  This is the relation $i\in u$ of \cref{rem:symmetry}, realized by a physical observable.
\end{proof}

The scanning observable shows the difference between NIP and uniform NIP (see \cref{rem:uniform}):
its idealized responses form a metrizable compactum, yet the response family has infinite fat-shattering dimension at a fixed positive gap.
In distribution-free learning from independent samples, the aim is to predict an unknown response with small mean absolute error under the probe distribution, with a prescribed probability of success.
No sample bound depending only on the required accuracy and probability of success can achieve this for the scanning family uniformly over target responses and probe distributions~\cite{Alon-BenDavid-CesaBianchi-Haussler:1997}.
The obstruction is already visible on the first $m$ sites, where chains of length $m$ realize every binary pattern.

%#
%## SECTION: Spontaneous symmetry breaking is a split pair
%#
\section{Spontaneous symmetry breaking is a split pair}\label{sec:ssb}

\noindent In this section and the next we show that the levels of \cref{sec:classification} have a physical meaning beyond bipartite observables.
Several notions of physics rest on an order of limits of a family of responses:
a preparation, such as a finite sample, has a response, such as its magnetization as a function of an external field, and the notion depends on whether the sample becomes infinite before or after the field is switched off.

The common setting is a pointwise bounded family $\cR$ of continuous functions on a compact metrizable probe space $Y$;
its idealized responses are the members of the pointwise closure $\overline\cR$.
We call the family \emph{stable} when, for every sequence $(r_k)_{k\in\NN}$ in $\cR$ and every sequence $(y_l)_{l\in\NN}$ in $Y$, the iterated limits $\lim_l\lim_kr_k(y_l)$ and $\lim_k\lim_lr_k(y_l)$ agree whenever both exist.
For a uniformly bounded family, as for the magnetizations below, Grothendieck’s double limit exchange condition~\cite{Grothendieck:1952} identifies stability with continuity of every idealized response;
countably many probes then separate the closure, which is metrizable by \cref{prop:countable-probes}.
The NIP characterization of \cref{thm:dichotomy} applies to pointwise bounded families, and \cref{def:niplevels} assigns their finer levels when the closure is separable.
For a fixed observable, convexity follows from allowing all mixed preparations;
for the parameterized families considered here, we compute the closures directly from the stated properties of the physical models.

The standard example of an order of limits is spontaneous symmetry breaking:
for a ferromagnet below its critical temperature, letting the external field vanish after the volume becomes infinite yields a nonzero magnetization, whereas the other order yields zero~\cite{Friedli-Velenik:2017}.
This is an instability in our sense, absent above the critical temperature (see \cref{thm:ssb} below).
With an offset field that moves the transition point, the field at which the magnetization changes sign, the family of responses below the critical temperature is moreover NIP$^{\mathrm{fc}}$
but not NIP$^{\mathrm{hs}}$:
the idealized responses of the two pure phases, the equilibrium states magnetized in opposite directions,
form a split pair, two functions that differ only at the transition point, and the idealized response of the symmetric state, the equal mixture of the two phases, is their midpoint, as in the proof of \cref{thm:convex}.

%#
%### SUBSECTION: The Ising model
%#
\subsection{The Ising model}
We recall the facts we need;
all of them can be found in the book of Friedli and Velenik~\cite{Friedli-Velenik:2017}.
Fix an integer dimension $d\ge2$, a coupling energy $J>0$ and an inverse temperature $\beta=1/(k_BT)>0$, where $k_B$ is the Boltzmann constant and $T$ the temperature.
The external field $h$ is a real number in units of energy, so $\beta J$ and $\beta h$ are dimensionless.
For each positive integer half-width $N$, let $\Lambda_N=\{-N,\dots,N\}^d$ with periodic boundary conditions, and define $Z_N\colon\RR\to(0,\infty)$ and $p_N\colon\RR\to\RR$ by
\[
  Z_N(h)=\sum_{s\in\{-1,1\}^{\Lambda_N}}\exp\Bigl(\beta J\sum_{i\sim j}s_is_j+\beta h\sum_is_i\Bigr),
  \qquad p_N(h)=\frac1{|\Lambda_N|}\log Z_N(h),
\]
where $i\sim j$ ranges over the unordered nearest-neighbor pairs of the torus, counted once, and the sum over $i$ ranges over $\Lambda_N$;
the corresponding energy of a configuration $s$ is $E_{N,h}(s)=-J\sum_{i\sim j}s_is_j-h\sum_i s_i$.
The Gibbs state is the probability measure on configurations that gives $s$ a weight proportional to $e^{-\beta E_{N,h}(s)}$;
its normalizing constant $Z_N(h)$ is the \emph{partition function}, and the function $p_N$ is the \emph{pressure} of the finite system, that is, $-\beta$ times the free energy per site.
The magnetization per site in the Gibbs state is the function $m_N\colon\RR\to[-1,1]$ given by $m_N(h)=p_N'(h)/\beta$.
\begin{enumerate}[label=\textup{(F\arabic*)},leftmargin=*]
\item
  The function $p_N$ is smooth, convex and even.
  Hence the function $m_N$ is continuous, nondecreasing and odd, with $m_N(0)=0$.
  The bound $|m_N|\le1$ holds because $m_N(h)$ is the mean, in the Gibbs state, of the average spin $|\Lambda_N|^{-1}\sum_is_i$, which lies in $[-1,1]$.
\item
  The limit $p(h)=\lim_Np_N(h)$, the pressure, exists for every $h$;
  it is convex and even, and it is differentiable at every $h\ne0$, by the Lee--Yang theorem~\cite{Lee-Yang:1952b}.
  Since $p_N$ and $p$ are convex, $m_N(h)\to p'(h)/\beta$ at every $h\ne0$, by Griffiths’ lemma on derivatives of convex functions.
  Indeed, convexity bounds $p_N'(h)$ between the secant slopes over $[h-\eta,h]$ and $[h,h+\eta]$ for $\eta>0$;
  first let $N\to\infty$, then $\eta\to0$ at a differentiability point of $p$.
\item
  The spontaneous magnetization $m^*(\beta)=p'(0^+)/\beta$ is the right derivative of $p$ at $0$ divided by $\beta$, and $p'(0^-)/\beta=-m^*(\beta)$.
  With $J$ fixed, the critical inverse temperature $\beta_c=\inf\{\beta:m^*(\beta)>0\}$ is positive, and it is finite by Peierls’s argument~\cite{Peierls:1936,Griffiths:1964};
  since $m^*(\beta)$ is the magnetization of the plus state
  (the infinite-volume limit of the Gibbs states with all boundary spins fixed up), it is nondecreasing in $\beta$, by the Griffiths--Kelly--Sherman inequalities, and $m^*(\beta)>0$ for $\beta>\beta_c$ (see~\cite[Chapter~3]{Friedli-Velenik:2017}).
\end{enumerate}
As $p'/\beta$ is nondecreasing and odd, (F3) implies that $p'(t)/\beta\ge m^*$ for $t>0$ and $p'(t)/\beta\le-m^*$ for $t<0$.

%#
%### SUBSECTION: Preparations and probes
%#
\subsection{Preparations and probes}
A preparation consists of a finite system $\Lambda_N$ and an offset field $a\in[-1/\beta,1/\beta]$, that is, a uniform field built into the sample;
a probe is an applied field $h\in Y=[-2/\beta,2/\beta]$, a value of an external parameter rather than a preparation of a second system.
The system then feels the field $h-a$, and the response of the preparation $(N,a)$ is the function $r_{N,a}\colon Y\to[-1,1]$ given by
\[
  r_{N,a}(h)=m_N(h-a)\qquad(h\in Y).
\]
The offset is a device that moves the transition point;
without it, the family $\{m_N|_Y:N\ge1\}$ has only one possible transition point, at $h=0$.
Let $\cR=\{r_{N,a}:N\ge1,\ a\in[-1/\beta,1/\beta]\}$.
Each $r_{N,a}$ is continuous and nondecreasing, and $r_{N,a'}\to r_{N,a}$ uniformly as $a'\to a$, because $m_N$ is uniformly continuous on compact sets.
The responses with $\beta a$ rational are therefore dense in $\cR$, and the closure $\overline\cR$ is separable.
Every function in $\overline\cR$ is nondecreasing, and so of the first Baire class;
thus $\overline\cR$ is a Rosenthal compactum, and, by \cref{thm:dichotomy}, the family $\cR$ is NIP.

We write $m^*=m^*(\beta)$, and for $a\in[-1/\beta,1/\beta]$ and $c\in[-m^*,m^*]$ we let $g_{a,c}\colon Y\to\RR$ be given by $g_{a,c}(h)=p'(h-a)/\beta$ for $h\ne a$ and $g_{a,c}(a)=c$.
These functions describe the idealized responses.

\begin{lemma}\label{lem:ssb-closure}
  The following hold:
  \begin{enumerate}[label=\textup{(\roman*)}]
  \item
    every $g_{a,c}$ with $-1/\beta<a<1/\beta$ belongs to $\overline\cR$;
  \item
    every function in $\overline\cR\setminus\cR$ equals $g_{a,c}$ for some $a\in[-1/\beta,1/\beta]$ and $c\in[-m^*,m^*]$.
  \end{enumerate}
\end{lemma}

\begin{proof}
  (i) Let $-1/\beta<a<1/\beta$.
  If $m^*>0$, then, since $\overline\cR$ is closed and $g_{a,c}\to g_{a,\pm m^*}$ pointwise as $c\to\pm m^*$, it suffices to take $|c|<m^*$, and by symmetry $0\le c<m^*$;
  if $m^*=0$, then $c=0$.
  If $c=0$, then $r_{N,a}\to g_{a,0}$, by (F1) at $h=a$ and by (F2) elsewhere.
  If $0<c<m^*$, then, for every $\eta>0$, we have $m_N(\eta)\to p'(\eta)/\beta\ge m^*>c$, by (F2) and (F3).
  Since $m_N(0)=0$, for all large $N$ the intermediate value theorem yields $\varepsilon_N\in(0,1/\beta)$ with $m_N(\varepsilon_N)=c$;
  for each $\eta\in(0,1/\beta)$, eventually $m_N(\eta)>c$, and then $\varepsilon_N<\eta$, because $m_N$ is nondecreasing.
  Hence $\varepsilon_N\to0$.
  Let $a_N=a-\varepsilon_N$.
  Then $r_{N,a_N}(a)=m_N(\varepsilon_N)=c$.
  For $h\ne a$, let $t=h-a$ and fix $\eta'$ with $0<\eta'<|t|$;
  for large $N$, we have $t-\eta'\le h-a_N\le t+\eta'$, and hence, by monotonicity,
  $m_N(t-\eta')\le r_{N,a_N}(h)\le m_N(t+\eta')$.
  By (F2), the outer terms converge to $p'(t-\eta')/\beta$ and $p'(t+\eta')/\beta$, which tend to $p'(t)/\beta$ as $\eta'\to0$, because $p'$ is continuous at $t\ne0$.
  Finally, as $a>-1/\beta$, the offsets $a_N$ lie in $[-1/\beta,1/\beta]$ for large $N$, and $r_{N,a_N}\to g_{a,c}$.

  (ii) Let $g\in\overline\cR\setminus\cR$.
  Since $\overline\cR$ is a Rosenthal compactum, it is Fréchet--Urysohn~\cite[Theorem~3F]{Bourgain-Fremlin-Talagrand:1978}, so the function $g$ is the limit of a sequence $r_{N_k,a_k}$.
  Passing to a subsequence, we may assume that $a_k\to a$ and that either $N_k$ is constant or $N_k\to\infty$.
  In the first case, uniform continuity forces $g=r_{N,a}\in\cR$, which is excluded.
  In the second case, the argument of the previous paragraph shows that $g(h)=p'(h-a)/\beta$ for $h\ne a$.
  Finally, given $\eta>0$, we have $|a-a_k|<\eta$ for all large $k$, and then $|r_{N_k,a_k}(a)|=|m_{N_k}(a-a_k)|\le m_{N_k}(\eta)$, because $m_{N_k}$ is odd and nondecreasing;
  the right side tends to $p'(\eta)/\beta$, and $p'(\eta)/\beta\to m^*$ as $\eta\to0$.
  Hence $|g(a)|\le m^*$, and $g=g_{a,g(a)}$.
\end{proof}

%#
%### SUBSECTION: The classification detects the phase transition
%#
\subsection{The classification detects the phase transition}
By \cref{lem:ssb-closure}, the classification of the family $\cR$ of responses with offsets reduces to the behavior of $p'$ at $0$.

\begin{keytheorem}\label{thm:ssb}
  Let $d\ge2$ and $\beta>0$ be given, let $m^*(\beta)$ be the spontaneous magnetization and $\beta_c$ the critical inverse temperature of~(F3), and let $\cR$ be the family of responses $r_{N,a}$ with offsets.
  \begin{enumerate}[label=\textup{(\alph*)}]
  \item
    If $m^*(\beta)=0$, in particular if $\beta<\beta_c$, then every idealized response is continuous:
    the family $\cR$ is stable, and $\overline\cR$ is metrizable.
  \item
    If $m^*(\beta)>0$, in particular if $\beta>\beta_c$, then the family $\cR$ is unstable, and $\overline\cR$ is first countable but not hereditarily separable:
    the family is NIP$^{\mathrm{fc}}$ but not NIP$^{\mathrm{hs}}$.
  \end{enumerate}
  The family $\{m_N|_Y:N\ge1\}$ without offsets is stable exactly when $m^*(\beta)=0$, and it is NIP$^{\mathrm{met}}$ for every $\beta$.
\end{keytheorem}

\begin{proof}
  (a) If $m^*=0$, then the pressure $p$ is differentiable everywhere, and its derivative $p'$, being the derivative of a convex function, is continuous.
  By \cref{lem:ssb-closure}(ii), every idealized response outside $\cR$ has the form $g_{a,0}(h)=p'(h-a)/\beta$ and is therefore continuous;
  so are the responses in $\cR$.
  By Grothendieck’s theorem~\cite{Grothendieck:1952}, a uniformly bounded family of continuous functions on a compact space whose pointwise closure consists of continuous functions satisfies the double limit condition, so the family $\cR$ is stable.
  The closure is metrizable by the argument of \cref{prop:stable-metrizable}:
  evaluation at the rational points of $Y$ is a continuous injection on continuous functions.

  (b) Let $m^*>0$.
  {\emergencystretch=1em
  For the family without offsets, take $h_l=1/(\beta l)$ for $l\ge1$ and a sequence $(N_k)_{k\in\NN}$ of positive integers with $N_k\to\infty$.
  By (F2) and (F3), $\lim_l\lim_km_{N_k}(h_l)=\lim_lp'(1/(\beta l))/\beta=m^*$, whereas, by continuity of each $m_N$ and (F1), $\lim_k\lim_lm_{N_k}(h_l)=\lim_km_{N_k}(0)=0$.
  The family $\{m_N\}$, and with it the family $\cR$, is therefore unstable.\par}

  Every function in $\overline\cR$ is nondecreasing, and the argument in the proof of \cref{thm:orders}(b) shows that the space of nondecreasing functions $Y\to[-1,1]$ is first countable;
  hence so is $\overline\cR$.
  For $a\in(-1/\beta,1/\beta)$, the neighborhood $\{g:|g(a)|<m^*/2\}$ of $g_{a,0}$ contains no $g_{b,0}$ with $b\ne a$, because $|g_{b,0}(a)|=|p'(a-b)/\beta|\ge m^*$.
  The set $\{g_{a,0}:-1/\beta<a<1/\beta\}$ is thus an uncountable discrete subspace, and $\overline\cR$ is not hereditarily separable.

  The first part of the last statement follows from part (a), since $\{m_N\}=\{r_{N,0}\}$ is a subfamily of $\cR$, and from the computation at the beginning of~(b).
  For the second part, the sequence $(m_N)$ converges pointwise, by (F1) and (F2), so its closure is countable, and a countable compact Hausdorff space is metrizable.
\end{proof}

For an interior offset $a\in(-1/\beta,1/\beta)$, the proof describes the limiting responses at the transition point $h=a$.
At $h=a$, the idealized responses take every value $c$ between $-m^*$ and $m^*$.
Suppose first that $m^*>0$.
The extreme values $\pm m^*$ are the magnetizations of the two pure phases, obtained by letting the field reach the transition from one side after the volume has become infinite;
they are the values at $a$ of the idealized responses $g_{a,m^*}$ and $g_{a,-m^*}$, which form a split pair.
The value $c=(2\lambda-1)m^*$ is the magnetization of the mixture of the two phases with weights
$\lambda$ and $1-\lambda$, for $0\le\lambda\le1$, and it is reached through finite-size rounding of the transition, the smoothing of the jump of the magnetization in a finite sample~\cite{Borgs-Kotecky:1990}, by tuning the offset as the volume grows.
As $a$ varies over $(-1/\beta,1/\beta)$, the symmetric responses $g_{a,0}$ form the uncountable discrete set that destroys hereditary separability.
They are the midpoints of the split pairs, as in the proof of \cref{thm:convex}.
If $m^*=0$, in particular above the critical temperature, the two extreme values coincide, and all idealized responses are continuous.
In this sense the classification detects the phase transition:
passing from $\beta<\beta_c$ to $\beta>\beta_c$ moves the family with offsets from the stable, metrizable class to the class of NIP$^{\mathrm{fc}}$
families that are not NIP$^{\mathrm{hs}}$, the class of the compactum of idealized responses of the comparison of positions in a box, a closed subspace of the Helly space.
Without offsets, the family stays NIP$^{\mathrm{met}}$, and only stability registers the transition.

%#
%### SUBSECTION: The quantum Ising model
%#
\subsection{The quantum Ising model}\label{sec:qising}
The analogue of \cref{thm:ssb} holds for a Hamiltonian whose terms do not commute, the Ising model in a transverse field (see \cref{thm:qssb} below).
The only ingredient of the classical proof that requires an argument for this model is the differentiability of the pressure away from zero field.
It can be derived from the Griffiths--Hurst--Sherman inequality,
which holds for this model through its space-time representation, a description of the quantum model by random configurations in one more dimension, that of imaginary time~\cite[Lemma~4.6]{Bjornberg-Grimmett:2009};
we derive it instead from a Lee--Yang theorem, which yields more (see \cref{prop:q-ly,prop:qanalytic} below):
the partition function has no zeros off the imaginary axis of the complex longitudinal field (the field along the axis of the couplings), and the pressure is real analytic away from zero field.

Fix a positive integer $d$, a coupling $J>0$, a transverse field $\Gamma\ge0$ and an inverse temperature $\beta>0$;
the coupling and the transverse field are energies, and the inverse temperature is the reciprocal of the thermal energy $k_BT$, so that $\beta J$ and $\beta\Gamma$ are pure numbers.
For each positive integer half-width $N$, take the torus $\Lambda_N=\{-N,\dots,N\}^d$ with periodic boundary conditions and Hilbert space $(\CC^2)^{\otimes\Lambda_N}$.
The Pauli matrices
$\sigma^x=\left(\begin{smallmatrix}0&1\\1&0\end{smallmatrix}\right)$
and $\sigma^z=\left(\begin{smallmatrix}1&0\\0&-1\end{smallmatrix}\right)$ act at each site $i\in\Lambda_N$ as $\sigma_i^x$ and $\sigma_i^z$, respectively, with the identity on the other tensor factors.
Let $S_N=\sum_{i\in\Lambda_N}\sigma_i^z$.
Define the Hamiltonian and the entire function $Z_N\colon\CC\to\CC$ by
\[
  H_N=-J\sum_{i\sim j}\sigma^z_i\sigma^z_j-\Gamma\sum_i\sigma^x_i,\qquad
  Z_N(h)=\Tr\exp\bigl(-\beta(H_N-hS_N)\bigr),
\]
where $h\in\CC$, the bond sum runs over unordered nearest-neighbor pairs, counted once, and the site sum runs over $\Lambda_N$;
the longitudinal field $h$ is an energy, as in the classical model.
For real $h$, the partition function is positive.
Define the pressure, $-\beta$ times the free energy per site, by $p_N\colon\RR\to\RR$, $p_N(h)=|\Lambda_N|^{-1}\log Z_N(h)$.
Define the magnetization $m_N\colon\RR\to[-1,1]$ by $m_N=p_N'/\beta$.
For $\Gamma=0$ this is the classical model with the same $J$, $\beta$ and $h$.
\begin{enumerate}[label=\textup{(Q\arabic*)},leftmargin=*]
\item
  The function $p_N$ is smooth, convex and even, and $m_N(h)$ is the expectation of $S_N/|\Lambda_N|$ in the Gibbs state, the density operator $e^{-\beta(H_N-hS_N)}/Z_N(h)$;
  hence the function $m_N$ is continuous, nondecreasing and odd, with $|m_N|\le1$, because $\|S_N\|=|\Lambda_N|$, and $m_N(0)=0$.
  For self-adjoint matrices $X$ and $S$ and real $u$, write $A=X+uS$.
  The matrix derivative formula
  \[
    \frac d{du}e^{X+uS}=\int_0^1e^{(1-t)A}Se^{tA}\,dt
  \]
  and cyclicity of the trace give $\frac d{du}\Tr e^{X+uS}=\Tr(Se^{X+uS})$.
  By the Gibbs variational formula, the function $u\mapsto\log\Tr e^{X+uS}=\sup_\rho\{\Tr\rho(X+uS)-\Tr\rho\log\rho\}$, where the supremum runs over density operators $\rho$ and $0\log0=0$, is convex as a supremum of affine functions.
  Taking $X=-\beta H_N$, $S=S_N$ and $u=\beta h$ gives $p_N'(h)=\beta m_N(h)$.
  The unitary $U=\prod_i\sigma^x_i$ commutes with $H_N$ and satisfies $US_NU^*=-S_N$, so $Z_N(-h)=Z_N(h)$ for all complex $h$.
\item
  The limit $p(h)=\lim_Np_N(h)$ exists for every real $h$;
  it is convex and even.
  For free boundary conditions, in which the bonds that wrap around the torus are omitted, the existence is standard~\cite{Simon:1993}, and the periodic Hamiltonian differs in operator norm from the free one by at most
  $J$ times the number of these bonds, which is at most $d(2N+1)^{d-1}$, so the two pressures differ by at most $\beta Jd/(2N+1)$, because $|\log\Tr e^A-\log\Tr e^B|\le\|A-B\|$ for self-adjoint matrices $A$ and $B$.
\end{enumerate}

The proof of \cref{prop:q-ly} below uses the Suzuki--Trotter mapping~\cite{Suzuki:1976b} between quantum spin systems and classical Ising systems in one more dimension to show
that the partition function of the quantum model is a limit of positive multiples of partition functions of classical Ising models on a larger graph, in which each site of the torus is repeated
$M$ times, for a positive integer $M$ that tends to infinity, one copy for each slice of “imaginary time”.
All couplings, between neighboring sites of the torus and between consecutive copies of a site, are ferromagnetic, so the classical Lee--Yang theorem applies to each of these classical models:
none of their partition functions vanishes off the imaginary axis.
A limit of holomorphic functions without zeros has no zeros either, unless it vanishes identically (see \cref{fact:toolbox}(b)), and so the quantum partition function inherits the zero-free half-planes.
The analyticity of the pressure then follows, in \cref{prop:qanalytic}, from a uniform bound and two classical theorems of complex analysis.
The two proofs use only the following standard facts.

\begin{fact}\label{fact:toolbox}
  \leavevmode
  \begin{enumerate}[label=\textup{(\alph*)}]
  \item
    \emph{Lee--Yang theorem}~\cite{Lee-Yang:1952b} (see~\cite[Chapter~3]{Friedli-Velenik:2017}).
    Let $G$ be a finite graph with vertex set $V$ and edge set $E$, let $J_e\ge0$ for $e\in E$, and let $h_v$ be complex numbers with $\operatorname{Re}h_v>0$ for $v\in V$.
    Then $\sum_{s\in\{-1,1\}^V}\exp\bigl(\sum_{e=\{u,w\}\in E}J_es_us_w+\sum_{v\in V}h_vs_v\bigr)\ne0$.
  \item
    \emph{Hurwitz’s theorem}~\cite{Titchmarsh:1939}.
    If holomorphic functions without zeros on a connected open set $U\subseteq\CC$ converge uniformly on compact subsets of $U$ to $f$, then either $f$ has no zeros in $U$ or $f$ vanishes identically.
  \item
    \emph{Borel--Carathéodory inequality}~\cite{Titchmarsh:1939}.
    If $f$ is holomorphic on a neighborhood of the closed disc $|z-z_0|\le R$ and $\operatorname{Re}f\le A$ there, then $|f(z)-f(z_0)|\le\frac{2r}{R-r}\bigl(A-\operatorname{Re}f(z_0)\bigr)$ for $|z-z_0|\le r<R$.
  \item
    \emph{Vitali’s theorem}~\cite{Titchmarsh:1939}.
    A sequence of holomorphic functions on a connected open set $U$ that is uniformly bounded on each compact subset of $U$ and converges at every point of a set with an accumulation point in $U$ converges uniformly on compact subsets of $U$;
    by Weierstrass’s theorem, so do the derivatives.
  \end{enumerate}
\end{fact}

\begin{proposition}[Lee--Yang theorem for the quantum Ising model]\label{prop:q-ly}
  For every $N$, the function $Z_N$ is entire, $Z_N(h)\ne0$ whenever $\operatorname{Re}h\ne0$, and $|Z_N(h)|\le Z_N(\operatorname{Re}h)$.
\end{proposition}

\begin{proof}
  Write $X=\beta J\sum_{i\sim j}\sigma^z_i\sigma^z_j+\beta hS_N$, which is diagonal in the basis $\{|s\rangle:s\in\{-1,1\}^{\Lambda_N}\}$ of eigenvectors of the operators $\sigma^z_i$, and $W=\beta\Gamma\sum_i\sigma^x_i$.
  The function $Z_N(h)=\Tr e^{X+W}$ is entire, because the exponential series converges locally uniformly in $h$.

  \emph{Step 1: the Lie--Trotter formula.}
  For matrices $A$ and $B$, we have $e^{A/M}e^{B/M}=I+(A+B)/M+O(M^{-2})$, uniformly for $A$ and $B$ in a bounded set;
  by taking the principal logarithm, we obtain $e^{A/M}e^{B/M}=\exp\bigl((A+B)/M+R_M\bigr)$ with $\|R_M\|=O(M^{-2})$, and hence $(e^{A/M}e^{B/M})^M=\exp(A+B+MR_M)\to e^{A+B}$, uniformly for $A$ and $B$ in a bounded set.
  Applied to $X$ and $W$, this yields $Z_{N,M}(h)=\Tr(e^{X/M}e^{W/M})^M\to Z_N(h)$ as $M\to\infty$, uniformly for $h$ in compact subsets of $\CC$.

  \emph{Step 2: a classical ferromagnet.}
  Suppose first that $\Gamma>0$, and let $\varepsilon=\beta\Gamma/M$.
  Since $e^{\varepsilon\sigma^x}=\cosh\varepsilon\,I+\sinh\varepsilon\,\sigma^x$, the matrix element
  $\langle s|e^{\varepsilon\sigma^x}|s'\rangle$ of a single spin is
  $\cosh\varepsilon$ if $s=s'$ and $\sinh\varepsilon$
  if $s\ne s'$, that is, $\langle s|e^{\varepsilon\sigma^x}|s'\rangle=c_\varepsilon e^{K_\varepsilon ss'}$ with $c_\varepsilon=(\sinh\varepsilon\cosh\varepsilon)^{1/2}>0$ and $K_\varepsilon=\frac12\log\coth\varepsilon>0$.
  Insertion of the basis between the factors of $(e^{X/M}e^{W/M})^M$ shows that, for $M\ge3$,
  \[
    Z_{N,M}(h)=c_\varepsilon^{M|\Lambda_N|}\sum_{s^1,\dots,s^M}\exp\sum_{k=1}^M\Bigl(\frac{\beta J}M\sum_{i\sim j}s^k_is^k_j
    +K_\varepsilon\sum_is^k_is^{k+1}_i+\frac{\beta h}M\sum_is^k_i\Bigr),
  \]
  where $s^1,\dots,s^M$ range over $\{-1,1\}^{\Lambda_N}$ and $s^{M+1}=s^1$.
  This is $c_\varepsilon^{M|\Lambda_N|}$ times the partition function of a classical Ising model on the finite graph $\Lambda_N\times\{1,\dots,M\}$, with nonnegative couplings and the same dimensionless field $\beta h/M$ at every site.
  By \cref{fact:toolbox}(a), it does not vanish when $\operatorname{Re}h>0$.
  Moreover, by the triangle inequality applied to the sum, $|Z_{N,M}(h)|\le Z_{N,M}(\operatorname{Re}h)$.
  If $\Gamma=0$, then the function $Z_N$ itself is a classical Ising partition function, and the same two statements hold directly.

  \emph{Step 3: Hurwitz’s theorem.}
  The functions $Z_{N,M}$ have no zeros in the connected open set $\{\operatorname{Re}h>0\}$ and converge to $Z_N$ uniformly on its compact subsets, and $Z_N$ is not identically zero there, since it is positive on the positive real axis.
  By \cref{fact:toolbox}(b), the function $Z_N$ has no zeros in $\{\operatorname{Re}h>0\}$, and by the evenness in~(Q1) none in $\{\operatorname{Re}h<0\}$.
  The inequality $|Z_N(h)|\le Z_N(\operatorname{Re}h)$ follows from Step~2 in the limit.
\end{proof}

\begin{proposition}\label{prop:qanalytic}
  The pressure $p$ is real analytic on $\RR\setminus\{0\}$, and $m_N(h)\to p'(h)/\beta$ for every $h\ne0$.
\end{proposition}

\begin{proof}
  By evenness it suffices to treat $h>0$.
  The half-plane $\{\operatorname{Re}h>0\}$ is simply connected and $Z_N$ does not vanish there, so the function $Z_N$ has a holomorphic logarithm on it, and we let $f_N$ be $|\Lambda_N|^{-1}$ times the logarithm that is real on $(0,\infty)$;
  thus $f_N=p_N$ on $(0,\infty)$ and $\operatorname{Re}f_N(h)=|\Lambda_N|^{-1}\log|Z_N(h)|\le p_N(\operatorname{Re}h)$, by \cref{prop:q-ly}.

  We bound $p_N$ on the real line uniformly in $N$.
  The torus has $d|\Lambda_N|$ bonds, whence $\|\beta H_N\|\le\beta(dJ+\Gamma)|\Lambda_N|$.
  Since the trace of the exponential of a self-adjoint matrix of size
  $2^{|\Lambda_N|}$ and norm at most $c$ lies between $2^{|\Lambda_N|}e^{-c}$ and $2^{|\Lambda_N|}e^c$, it follows that $|p_N(0)-\log2|\le\beta(dJ+\Gamma)=:b$;
  as $|p_N'|\le\beta$, also $|p_N(x)-\log2|\le b+\beta|x|$ for real $x$.

  Fix $x_0>0$, and let $R=x_0/2$ and $r=x_0/4$.
  On the disc $|h-x_0|\le R$ we have $\operatorname{Re}f_N\le\log2+b+\frac32\beta x_0$, and $f_N(x_0)=p_N(x_0)\ge\log2-b-\beta x_0$.
  By the Borel--Carathéodory inequality (see \cref{fact:toolbox}(c)),
  \[
    |f_N(h)-f_N(x_0)|\le\frac{2r}{R-r}\bigl(\sup_{|z-x_0|\le R}\operatorname{Re}f_N(z)-\operatorname{Re}f_N(x_0)\bigr)
    \le2\bigl(2b+\tfrac52\beta x_0\bigr)
  \]
  for $|h-x_0|\le r$.
  Since also $|f_N(x_0)|\le\log2+b+\beta x_0$, the functions $f_N$ are uniformly bounded on each disc $\{|h-x_0|<x_0/4\}$, and therefore locally uniformly bounded on the connected open set $U$ that is the union of these discs, which contains $(0,\infty)$.
  They converge on $(0,\infty)$, by~(Q2).
  By Vitali’s theorem (see \cref{fact:toolbox}(d)), they converge uniformly on compact subsets of $U$ to a holomorphic function $F$, which equals $p$ on $(0,\infty)$;
  hence the pressure $p$ is real analytic there, and $m_N=f_N'/\beta\to F'/\beta=p'/\beta$ on $(0,\infty)$.
\end{proof}

We define the spontaneous magnetization of the quantum model by
\[
  m^*(\beta,\Gamma)=\frac{p'(0^+)}{\beta};
\]
by evenness, $p'(0^-)/\beta=-m^*(\beta,\Gamma)$.
With offsets as in the classical model, let $r_{N,a}(h)=m_N(h-a)$ for $N\ge1$, $a\in[-1/\beta,1/\beta]$ and $h\in Y=[-2/\beta,2/\beta]$, and let $\cR_\Gamma$ be the family of these functions.
For $a\in[-1/\beta,1/\beta]$ and $c\in[-m^*,m^*]$, where $m^*=m^*(\beta,\Gamma)$, let $g_{a,c}$ be defined as before \cref{lem:ssb-closure}, with the pressure $p$ of the quantum model.

\begin{keytheorem}\label{thm:qssb}
  Let $d\ge1$, $J>0$, $\Gamma\ge0$ and $\beta>0$ be given, let $m^*(\beta,\Gamma)$ be the spontaneous magnetization of the quantum model, and let $\cR_\Gamma$ be its family of responses $r_{N,a}$ with offsets.
  \begin{enumerate}[label=\textup{(\alph*)}]
  \item
    If $m^*(\beta,\Gamma)=0$, then every idealized response is continuous:
    the family $\cR_\Gamma$ is stable, and its closure is metrizable.
  \item
    If $m^*(\beta,\Gamma)>0$, then the family $\cR_\Gamma$ is unstable, and its closure is first countable but not hereditarily separable:
    the family is NIP$^{\mathrm{fc}}$ but not NIP$^{\mathrm{hs}}$.
    For each $a\in(-1/\beta,1/\beta)$, the functions $g_{a,m^*}$ and
    $g_{a,-m^*}$ lie in the closure and differ only at $a$, and their midpoints $g_{a,0}$ form an uncountable discrete set.
  \end{enumerate}
  The family $\{m_N\}$ without offsets is stable exactly when $m^*(\beta,\Gamma)=0$, and it is NIP$^{\mathrm{met}}$ in all cases.
\end{keytheorem}

\begin{proof}
  The discussion preceding \cref{lem:ssb-closure}, the lemma and the proof of \cref{thm:ssb} use only the following properties:
  the function $m_N$ is continuous, nondecreasing and odd, with $|m_N|\le1$ and $m_N(0)=0$;
  the limit $p$ exists, is convex and even, and is differentiable at every $h\ne0$, with $m_N(h)\to p'(h)/\beta$;
  and $p'(0^\pm)/\beta=\pm m^*$.
  These hold by~(Q1), (Q2) and \cref{prop:qanalytic}, and the proofs apply verbatim with $m^*(\beta,\Gamma)$ in place of $m^*(\beta)$.
\end{proof}

When $m^*>0$, for each interior offset $a\in(-1/\beta,1/\beta)$ the values $g_{a,\pm m^*}(a)=\pm m^*$ are the magnetizations of the two pure phases at the transition point $h=a$, and the value $g_{a,0}(a)=0$ is that of the symmetric state.
The classification distinguishes vanishing from nonvanishing spontaneous magnetization at positive temperature.
It does not determine the critical parameters; those enter through the physical results recalled below.
Both cases of the theorem occur, by results of Björnberg and Grimmett~\cite{Bjornberg-Grimmett:2009}:
for $d=1$ and $\beta<\infty$ the spontaneous magnetization vanishes, and for
$d\ge2$ and $\Gamma>0$ there is a critical value of
$2J/\Gamma$, depending on $d$ and $\beta\Gamma$, below which $m^*=0$ and above which $m^*>0$ (see~\cite[Section~1 and Theorems~1.1--1.2]{Bjornberg-Grimmett:2009}, where the existence of the critical value is attributed to standard methods).
Their conventions match ours as follows.
They use the zero-field Hamiltonian
$-\frac12\lambda\sum\sigma^z\sigma^z-\delta\sum\sigma^x$, and, with a longitudinal term $-\gamma\sum\sigma^z$, the parameters correspond to
$\lambda=2J$, $\delta=\Gamma$ and $h=\gamma$.
They define the spontaneous magnetization as the limit, as
$\gamma\downarrow0$, of the infinite-volume magnetization at field $\gamma$ on the same periodic boxes $\Lambda_N$, through the space-time representation.
By translation invariance and \cref{prop:qanalytic}, the infinite-volume magnetization is
$p'(\gamma)/\beta$, and since the right derivative of a convex function is the right limit of its derivative, their spontaneous magnetization is $p'(0^+)/\beta=m^*(\beta,\Gamma)$.
Order at low temperature and small $\Gamma/J$ goes back to Ginibre~\cite{Ginibre:1969}.

The Lee--Yang theorem~\cite{Lee-Yang:1952b} was extended to general classical ferromagnets by Newman~\cite{Newman:1974} and by Lieb and Sokal~\cite{Lieb-Sokal:1981},
and Lee--Yang theorems for quantum spin systems go back to Asano~\cite{Asano:1970b} and to Suzuki and Fisher~\cite{Suzuki-Fisher:1971b}.
The formulation of Fröhlich and Rodriguez~\cite[Theorem~5]{Frohlich-Rodriguez:2012b}, applied to $H_N$ after relabeling the axes, with $h_1=h$, $h_2=\Gamma$ and $h_3=0$ in their condition $\operatorname{Re}h_1>|h_2|+|h_3|$, gives the zero-free region $\operatorname{Re}h>\Gamma$.
\Cref{prop:q-ly} gives the whole half-plane $\operatorname{Re}h>0$ for the transverse-field Ising model, by a Suzuki--Trotter approximation by classical ferromagnets followed by Hurwitz’s theorem, and \cref{prop:qanalytic} derives from it the analytic input needed for the response classification;
the argument uses standard quantum Lee--Yang methods, in the spirit of Simon and Griffiths~\cite{Simon-Griffiths:1973} and of Suzuki~\cite{Suzuki:1976b}.

%#
%## SECTION: Chiral symmetry breaking and metallic conduction are instabilities
%#
\section{Chiral symmetry breaking and metallic conduction are instabilities}\label{sec:spectral}

\noindent The chiral-condensate and current-response applications distinguish failure of stability from failure of metrizability.
For the response families considered below, the limiting response can be discontinuous at one probe value while the full pointwise closure remains metrizable.
Under the spectral hypotheses stated below, we compute
the discrepancy between the iterated limits from the limiting density of Dirac eigenvalues at zero for the condensate and from the mass at zero of a limiting weighted excitation measure for the current response (see \cref{thm:chiral,thm:sc}).
The density describes the near-zero spectrum of the Dirac operator governing quarks in a gluon field;
the atom describes the current spectral weight accumulating at zero excitation energy.
A common lemma places both response families in NIP$^{\mathrm{met}}$ and identifies stability with continuity at the exceptional probe value.

\begin{lemma}\label{lem:single}
  Let $Q$ be a compact metrizable space, $q_0$ a point of $Q$ that is not isolated, and $q\mapsto g_q$ a map from $Q$ into $\RR^{[0,1]}$ that is continuous for the topology of pointwise convergence.
  Suppose that $g_q$ is continuous on $[0,1]$ for $q\ne q_0$, that $G=g_{q_0}$ is continuous on $(0,1]$, and that the one-sided limit
  $G(0^+)=\lim_{s\to0^+}G(s)$ exists as a finite real number.
  Then the family $\cR=\{g_q:q\ne q_0\}$ is NIP$^{\mathrm{met}}$, and it is stable if and only if $G$ is continuous at $0$.
  For every sequence $(q_k)_{k\in\NN}$ in $Q\setminus\{q_0\}$ converging to $q_0$ and every sequence $(s_l)_{l\in\NN}$ in $(0,1]$ converging to $0$, we have
  \[
    \lim_l\lim_kg_{q_k}(s_l)=G(0^+),\qquad\lim_k\lim_lg_{q_k}(s_l)=G(0).
  \]
\end{lemma}

\begin{proof}
  The image of $Q$ is compact, contains $\cR$ and is contained in $\overline\cR$, because $q_0$ is not isolated;
  it is thus $\overline\cR$.
  Its members are continuous on $(0,1]$, so two distinct members differ at a rational point of $[0,1]$, and $\overline\cR$ is metrizable by \cref{prop:countable-probes}.
  The function $G$ is of the first Baire class, as the pointwise limit of the continuous functions that equal $G$ on $[1/n,1]$ and are linear on $[0,1/n]$ with value $G(0)$ at $0$.
  For each $s\in[0,1]$, the function $q\mapsto g_q(s)$ is continuous on the compact space $Q$, so the family $\cR$ is pointwise bounded.
  Every member of its closure is of the first Baire class, and \cref{thm:dichotomy} therefore shows that $\cR$ is NIP;
  its closure is compact metrizable and hence separable, so it is NIP$^{\mathrm{met}}$.

  Suppose first that $G$ is continuous at $0$, so that every $g_q$, including $g_{q_0}$, is continuous on $[0,1]$.
  To verify stability, let $(p_k)_{k\in\NN}$ be a sequence in $Q\setminus\{q_0\}$ and $(t_l)_{l\in\NN}$ a sequence in $[0,1]$ for which both iterated limits of $g_{p_k}(t_l)$ exist as real numbers.
  By compact metrizability, we may pass to subsequences for which $p_k\to p_*$ in $Q$ and $t_l\to t_*$ in $[0,1]$;
  these passages do not change either existing iterated limit.
  Pointwise continuity in the parameter gives $\lim_k g_{p_k}(t_l)=g_{p_*}(t_l)$ for each $l$, while continuity in the probe gives $\lim_l g_{p_k}(t_l)=g_{p_k}(t_*)$ for each $k$.
  Consequently,
  \[
    \lim_l\lim_k g_{p_k}(t_l)
      =\lim_l g_{p_*}(t_l)=g_{p_*}(t_*),\qquad
    \lim_k\lim_l g_{p_k}(t_l)
      =\lim_k g_{p_k}(t_*)=g_{p_*}(t_*).
  \]
  Thus the two iterated limits agree, as required.

  Now take sequences $(q_k)$ and $(s_l)$ as in the statement.
  For each $l$, pointwise continuity in the parameter gives $\lim_k g_{q_k}(s_l)=G(s_l)$;
  for each $k$, continuity of $g_{q_k}$ at $0$ gives $\lim_l g_{q_k}(s_l)=g_{q_k}(0)$.
  It follows that
  \[
    \lim_l\lim_k g_{q_k}(s_l)=\lim_l G(s_l)=G(0^+),\qquad
    \lim_k\lim_l g_{q_k}(s_l)=\lim_k g_{q_k}(0)=G(0).
  \]
  Such a sequence $(q_k)$ exists because $q_0$ is not isolated and $Q$ is metrizable, and we may take $s_l=1/(l+1)$.
  If $G$ is discontinuous at $0$, then $G(0^+)\ne G(0)$;
  the two displayed limits are finite and unequal, so $\cR$ is not stable.
\end{proof}

%#
%### SUBSECTION: The chiral condensate
%#
\subsection{The chiral condensate}
We work with a lattice regularization, in which Euclidean spacetime (spacetime with imaginary time) is replaced by a lattice of points in a box of volume $V$;
we use lattice units, with $\hbar=c=1$ and lattice spacing $1$, in which masses and volumes are pure numbers.
The quarks are described by an anti-self-adjoint Dirac operator $D_V$, as for staggered fermions, a standard discretization of quarks on the lattice;
the Dirac operator governs the motion of the quarks in a given gauge field, the gluon field with which they interact.
For each gauge field in the chosen ensemble, we assume that the spectrum consists of nonzero pairs $\pm i\lambda_n$, indexed by $n$ in a finite set, with $\lambda_n>0$.
This assumption excludes exact zero modes, eigenvectors with eigenvalue $0$;
if a finite-volume operator has $n_0$ such modes, they contribute $n_0/(mV)$ to its condensate at mass $m>0$ before averaging over gauge fields.
The spectral density is averaged over gauge fields with a measure that does not depend on the mass $m$.
Two standard settings have this property:
the quenched theory, in which the quarks do not act back on the gauge field, and the valence condensate at a fixed sea quark mass, where the valence quarks are those of the condensate, whose mass $m$ varies, and the sea quarks are those that act back on the gauge field.
The second is a standard setting for the Banks--Casher relation, whose original form uses a measure that depends on $m$.
For each finite volume $V$, define the condensate $\Sigma_V\colon[0,\infty)\to[0,\infty)$ by the following formula for $m>0$, together with $\Sigma_V(0)=0$:
\[
  \Sigma_V(m)=\Bigl\langle\frac1V\operatorname{Tr}(D_V+m)^{-1}\Bigr\rangle
  =\int\frac{2m}{m^2+\lambda^2}\,d\kappa_V(\lambda),\qquad
  \kappa_V=\Bigl\langle\frac1V\sum_n\delta_{\lambda_n}\Bigr\rangle,
\]
where the angle brackets denote the average over gauge fields and $\delta_{\lambda_n}$ is the point mass at $\lambda_n$.
We assume:
\begin{enumerate}[label=\textup{(C\arabic*)},leftmargin=*]
\item
  each $\kappa_V$ is a positive measure on $(0,\Lambda]$, of total mass at most $M_0$, against which the function $\lambda\mapsto1/\lambda$ is integrable, where the cutoff $\Lambda$ and the bound $M_0$ do not depend on $V$;
\item
  as $V\to\infty$ along a given sequence of volumes, the measure $\kappa_V$ converges weakly on $[0,\Lambda]$ to a measure $\kappa$ that, on some interval $[0,\varepsilon_0]$, has a density $\rho$ continuous at $0$.
\end{enumerate}
The cutoff and the bound on the total mass in (C1) hold on the lattice;
the integrability of $1/\lambda$ holds for each single gauge field, whose spectrum is finite and avoids $0$, and for the average over gauge fields it is an assumption on the finite-volume spectrum near $0$.
Hypothesis (C2) is the physical assumption at issue, the existence of a limiting density $\rho$ of near-zero eigenvalues per unit volume, continuous at $0$.
A preparation is a volume $V$, a probe is a mass $m\in[0,1]$, and the response is the restriction $\Sigma_V|_{[0,1]}$.
We retain the notation $\Sigma_V$ for this restriction when discussing the response family.
Each $\Sigma_V$ is continuous on $[0,1]$, with value $0$ at $0$, by dominated convergence, because $2m/(m^2+\lambda^2)\le1/\lambda$ for $\lambda>0$.

\begin{lemma}\label{lem:poisson}
  Assume \textup{(C1)} and \textup{(C2)}.
  Then $\Sigma_V(m)\to\Sigma(m)=\int\frac{2m}{m^2+\lambda^2}\,d\kappa(\lambda)$ for every $m>0$, the function $\Sigma$ is continuous on $(0,1]$, and $\Sigma(m)\to\pi \rho(0)$ as $m\to0^+$.
\end{lemma}

\begin{proof}
  For fixed $m>0$, the function $\lambda\mapsto2m/(m^2+\lambda^2)$ is continuous and bounded on $[0,\Lambda]$, so the first statement is weak convergence, and the second follows from dominated convergence.
  For the third, choose $\varepsilon\in(0,\varepsilon_0]$ so small that $\rho$ is bounded on $[0,\varepsilon]$;
  since $\kappa$ has the density $\rho$ there, it has no atom at $0$.
  On $(\varepsilon,\Lambda]$ the integrand is at most $2m/\varepsilon^2$, and the contribution of this interval therefore tends to $0$.
  On $[0,\varepsilon]$, the substitution $\lambda=mu$ shows that
  \[
    \int_0^\varepsilon\frac{2m}{m^2+\lambda^2}\rho(\lambda)\,d\lambda
    =\int_0^{\varepsilon/m}\frac{2}{1+u^2}\rho(mu)\,du,
  \]
  and the right side tends to $\rho(0)\int_0^\infty\frac{2\,du}{1+u^2}=\pi \rho(0)$, by dominated convergence.
\end{proof}

\begin{keytheorem}\label{thm:chiral}
  Assume \textup{(C1)} and \textup{(C2)}.
  \begin{enumerate}[label=\textup{(\alph*)}]
  \item
    The iterated limits of the condensate are
    \[
      \lim_{m\to0}\lim_{V\to\infty}\Sigma_V(m)=\pi \rho(0),\qquad
      \lim_{V\to\infty}\lim_{m\to0}\Sigma_V(m)=0.
    \]
  \item
    The family $\{\Sigma_V\}$ is stable if and only if $\rho(0)=0$.
  \item
    In all cases the family is NIP$^{\mathrm{met}}$.
  \end{enumerate}
\end{keytheorem}

\begin{proof}
  Let $Q$ consist of the numbers $1/V$ along the sequence of volumes, together with $q_0=0$, let $g_{1/V}=\Sigma_V$, and let $G(m)=\Sigma(m)$ for $m>0$ and $G(0)=0$.
  By \cref{lem:poisson} and $\Sigma_V(0)=0$, the map $q\mapsto g_q$ is continuous, the function $G$ is continuous on $(0,1]$, and $G(0^+)=\pi \rho(0)$.
  The theorem now follows from \cref{lem:single}.
\end{proof}

Letting the volume grow before the mass vanishes keeps a density of near-zero modes, and the condensate is $\pi \rho(0)$;
the other order sends the mass to $0$ in each finite volume, and the integrability in \textup{(C1)} makes its condensate tend to $0$.
Thus the family is stable exactly when the limiting valence condensate $\pi\rho(0)$ of the chosen ensemble vanishes.
When $\Sigma_0=\pi \rho(0)>0$, the microscopic regime with $\mu=mV\Sigma_0$ fixed concerns spectral data at the scale $1/V$.
Chiral random-matrix ensembles provide explicit condensate profiles in this regime, with their dependence on the microscopic spectrum, flavor content and topology~\cite{Leutwyler-Smilga:1992,Verbaarschot-Wettig:2000}.
The condensate family has a single exceptional probe value, $m=0$.
In the Ising model, by contrast, the offsets produce a continuum of transition points, and the intermediate values at each of them are values of idealized responses, obtained by tuning the offset as the volume grows;
a condensate family with a continuum of transition points is sought in \cref{sec:questions}.

%#
%### SUBSECTION: The current response
%#
\subsection{The current response}
The Scalapino--White--Zhang criterion concerns the current response of a clean many-body system, one free of disorder, at zero temperature.
Let $H$ be its Hamiltonian, with ground state $\Omega$ of energy $E_0$;
fix two perpendicular directions in space, and, for a real number
$q>0$, let $J_q$ be the transverse current operator at the wave vector of length $q$ along the first direction, the component along the second direction of the Fourier transform of the current density.
For a finite system of volume $V$, the operator $J_q$ is extensive.
We let $j_q=J_q/\sqrt V$, so that its spectral measure has weights normalized per unit volume.
For an infinite system, the objects below are the infinite-volume limits of the corresponding finite-volume spectral measures, and \cref{thm:sc} below is a statement about these measures.
We use units in which $\hbar=1$, so that frequencies are energies, and fixed units of length and energy, such as the lattice spacing and the hopping amplitude of a lattice model.
We consider the case in which the spectral measures of the opposite Fourier currents $J_q$ and $J_{-q}=J_q^*$ agree.
We restrict the parameters to $q\in(0,1]$ and $\omega\in[0,1]$.
Under the integrability assumption (S1) below, the two Lehmann resolvents combine to give the paramagnetic current response at imaginary frequency $i\omega$ in the form
\[
  \Pi(q,i\omega)=\int\frac{2\nu}{\nu^2+\omega^2}\,d\mu_q(\nu),
\]
where $\mu_q$ is the spectral measure of $H-E_0$ in the vector $j_q\Omega$, restricted to $(0,\infty)$;
for a finite system, $\mu_q=\sum_k|\langle\psi_k,j_q\Omega\rangle|^2\delta_{E_k-E_0}$, the sum running over the
$k$ with $E_k>E_0$, where the vectors $\psi_k$, for $k$ in a countable index set, form an orthonormal basis of eigenvectors of $H$, with $H\psi_k=E_k\psi_k$.
The total response on the same parameter domain is $\mathcal{K}(q,i\omega)=D_0-\Pi(q,i\omega)$, where
$D_0$ is the diamagnetic term, the expectation, per unit volume, of minus the kinetic energy of the bonds parallel to the current on a lattice, and, in the continuum, the ratio of the density of the charge carriers to their mass.
Following Scalapino, White and Zhang~\cite{Scalapino-White-Zhang:1993}, we define the \emph{superfluid weight} and the \emph{Drude weight}, up to their factor~$\pi e^2$, where $e$ is the charge of the carriers, by
\[
  D_s=\lim_{q\to0}\mathcal{K}(q,0),\qquad D=\lim_{\omega\to0}\lim_{q\to0}\mathcal{K}(q,i\omega).
\]
The superfluid weight is the transverse static limit of Scalapino, White and Zhang.
Their Drude weight, and that of Kohn~\cite{Kohn:1964}, is taken at $q=0$ and continued from imaginary to real frequency;
it agrees with $D$ when the spectral measures $\mu_q$ depend continuously on $q$ at $q=0$ away from zero energy and when, in addition, the imaginary-frequency limit agrees with the real-frequency one.
In an insulator $D_s=D=0$;
in a clean normal metal $D>0=D_s$;
positive $D_s$ describes a superconducting response.
A gap in the transverse current spectral weight that persists as $q\to0$ is a sufficient condition for $D=D_s$~\cite{Scalapino-White-Zhang:1993}.
In a disordered metal the Drude weight vanishes, and at positive temperature the normal fluid, the thermally excited carriers that do not condense, yields $D>D_s$ even in a superconductor;
both cases lie outside the setting of this subsection, a clean system at zero temperature.

It is convenient to absorb a factor into the measure:
let $d\zeta_q(\nu)=\frac2\nu\,d\mu_q(\nu)$, so that
\[
  \Pi(q,i\omega)=\int\frac{\nu^2}{\nu^2+\omega^2}\,d\zeta_q(\nu),\qquad
  \Pi(q,0)=\zeta_q\bigl((0,\infty)\bigr).
\]
We assume:
\begin{enumerate}[label=\textup{(S\arabic*)},leftmargin=*]
\item
  each $\zeta_q$, for $0<q\le1$, is a positive measure on $(0,\Lambda]$ of total mass at most $D_0$, where the bandwidth $\Lambda$ does not depend on $q$;
\item
  there is a measure $\zeta_0$ on $[0,\Lambda]$ such that, for every continuous function
  $h$ on $[0,\Lambda]$, the map $q\mapsto\int h\,d\zeta_q$, with each $\zeta_q$ regarded as a measure on $[0,\Lambda]$, is continuous on $(0,1]$ and converges to $\int h\,d\zeta_0$ as $q\to0$.
\end{enumerate}
The bound in (S1) is the statement that $\mathcal{K}(q,0)\ge0$, that is, that the ground state is a local minimum of the energy as a function of a static vector potential.
This is an assumption:
for lattice fermions it can fail, as Lieb’s flux-phase theorem shows~\cite{Lieb:1994}.
The compact support in (S1) is a convenience:
since $\nu^2/(\nu^2+\omega^2)$ extends continuously to $[0,\infty]$, the proof below works for tight families of measures on $(0,\infty)$.
The limit $\zeta_0$ may have an atom at $0$, which carries the weight of current-carrying excitations whose energy tends to $0$ with $q$.
A preparation is the length $q\in(0,1]$ of the wave vector, a probe is a frequency $\omega\in[0,1]$, and the response is $g_q(\omega)=\Pi(q,i\omega)$.

\begin{keytheorem}\label{thm:sc}
  Assume \textup{(S1)} and \textup{(S2)}.
  \begin{enumerate}[label=\textup{(\alph*)}]
  \item
    We have $D_s=D_0-\zeta_0\bigl([0,\Lambda]\bigr)$ and $D=D_0-\zeta_0\bigl((0,\Lambda]\bigr)$, whence $D-D_s=\zeta_0(\{0\})$.
  \item
    The family of responses $\{g_q:0<q\le1\}$ is stable if and only if $D=D_s$.
    In particular, insulating responses with $D=D_s=0$ and superconducting responses satisfying $D=D_s>0$ are stable, whereas clean normal-metal responses with $D>0=D_s$ are unstable.
  \item
    In all cases the family is NIP$^{\mathrm{met}}$.
  \end{enumerate}
\end{keytheorem}

\begin{proof}
  For each $\omega\in(0,1]$, define $\chi_\omega\colon[0,\Lambda]\to[0,1]$ by $\chi_\omega(\nu)=\nu^2/(\nu^2+\omega^2)$.
  This function is continuous and satisfies $\chi_\omega(0)=0$.
  Each $g_q$ is continuous on $[0,1]$, by dominated convergence.
  By (S2), applied to the continuous function $1$, we have $g_q(0)=\zeta_q((0,\Lambda])\to\zeta_0([0,\Lambda])$, and the formula for $D_s$ follows;
  applied to $\chi_\omega$ with $\omega>0$, (S2) shows that $g_q(\omega)\to G(\omega)=\int_{(0,\Lambda]}\chi_\omega\,d\zeta_0$, and, by dominated convergence, $G(\omega)\to\zeta_0((0,\Lambda])$ as $\omega\to0$, which yields the formula for $D$.
  This proves~(a).
  Let $g_0=G$, with $G$ extended to $0$ by $G(0)=\zeta_0([0,\Lambda])$.
  The map $[0,1]\to[0,D_0]^{[0,1]}$, given by $q\mapsto g_q$, is continuous for the pointwise topology, by (S2) and the computation just made, and $G$ is continuous on $(0,1]$ with $G(0)-G(0^+)=\zeta_0(\{0\})=D-D_s$.
  Parts (b) and~(c) now follow from \cref{lem:single}.
\end{proof}

Under \textup{(S1)} and \textup{(S2)}, the discrepancy $D-D_s$ is exactly the mass at zero of the limiting weighted current spectral measure $\zeta_0$.
A positive mass at zero therefore has two equivalent descriptions:
the two orders of limits of the current response disagree, and the family of responses is unstable.
When this mass vanishes, the limiting response is continuous and the family is stable.
In both cases the closure remains metrizable, so the distinction between stable and unstable current responses lies within NIP$^{\mathrm{met}}$.

%#
%## SECTION: Questions
%#
\section{Questions}\label{sec:questions}

We close with problems concerning the possible response compacta, the role of noncommutative probe algebras, and extensions of the classification to further physical response families.
We also ask for finite combinatorial tests of the finer topological levels and for a proof without forcing of Todorčević’s step for points without a countable neighborhood base.

\begin{enumerate}[label=\textup{(\arabic*)},leftmargin=*]
\item
  \emph{Light cones.}
  In \cref{thm:causal} Bob detects at a fixed time.
  Locate the causal order in the classification when Bob’s detection time is also uncertain, and in spacetimes of dimension $1+d$, with $d\ge2$ dimensions of space.
\item
  \emph{The minimal families of the box.}
  By \cref{thm:hepta-natural}, the comparison of positions in a box realizes the minimal families $D_3$, $D_4$ and $D_7$ of the heptachotomy.
  Does it realize $D_6$?
  Which of $D_6$ and $D_7$ does \cref{thm:orders}(b) produce from a given dense sequence?
\item
\label{ques:noncommutative} \emph{Noncommutative probes.}
  Find a natural observable that is NIP but unstable and whose probe algebra is not commutative, and locate it in the classification.
  Can noncommutativity produce a phenomenon in the classification that observables with commutative probe algebras cannot produce, for instance in the minimal families of the heptachotomy that arise from a given dense sequence of preparations?
\item
  \emph{Orientations.}
  By \cref{thm:phases}, the comparison of phases on a circle is NIP but not NIP$^{\mathrm{fc}}$.
  Locate in the classification the comparison of directions on a sphere, “Bob’s direction lies in the open hemisphere centered at Alice’s”, and relative orientations in higher dimensions.
\item
\label{ques:finitary-levels} \emph{Finitary levels.}
  Is there a finitary, combinatorial characterization of NIP$^{\mathrm{fc}}$, NIP$^{\mathrm{hs}}$ and NIP$^{\mathrm{met}}$, in the spirit of the fat-shattering dimension for uniform NIP, sharpening \cref{prop:levels-separation}?
\item
  \emph{Quantum phase transitions.}
  Extend \cref{thm:qssb} to the ground state of the Ising model in a transverse field, where the transition is driven by the transverse field $\Gamma$ at zero temperature.
\item
  \emph{A continuum of transitions.}
  Is there a natural family of chiral condensates with a continuum of transition points, obtained for instance from twisted boundary conditions or a chemical potential, that is NIP$^{\mathrm{fc}}$ but not NIP$^{\mathrm{hs}}$?
\item
  \emph{Positive temperature.}
  Extend \cref{thm:sc} to positive temperature, where the Matsubara frequencies, the imaginary frequencies at which thermal response functions are evaluated, form a discrete set, and the Drude weight is obtained by extrapolation.
\item
  \emph{Unbounded formulas.}
  In Keisler’s setting, classify the energy functional
  $\rho\mapsto\Tr(\rho H)$ of a Hamiltonian $H$, taking as shards (the parts of the state space on which every predicate stays within prescribed bounds) the sets of states that satisfy the energy constraints $\Tr(\rho H)\le E$, for real numbers $E$.
\item
  \emph{Forcing.}
  Can Todorčević’s step that produces a copy of $A(2^\NN)$ from a point without a countable neighborhood base be proved without forcing?
\end{enumerate}

%#
%## SECTION: Acknowledgments
%#
\section*{Acknowledgments}
We thank Stevo Todorčević for pointing out to us that \cref{thm:convex} is known, in the larger class of locally connected Rosenthal compacta, and for the references that establish it.

%#
%## SECTION: Declarations
%#
\section*{Declarations}
\noindent\emph{Competing interests.}
The authors have no competing interests to
declare.

\noindent\emph{Data availability.}
No data were generated or analyzed in this
work.

\bibliographystyle{amsalpha}
\bibliography{references}

\end{document}

%% file: venn_body.tex
\begin{tikzpicture}[font=\scriptsize, x=1cm, y=1cm]
  \colorlet{vtext}{deepblue!85!black}
  % fills, outermost first
  \fill[urlpurple!8]  \VennRing{-7.8}{7.8}{4.6};
  \fill[urlpurple!17] \VennRing{-7.6}{5.4}{4.25};
  \fill[midblue!13]   \VennRing{-7.4}{3.0}{3.9};
  \fill[midblue!19]   \VennRing{-7.2}{0.5}{3.6};
  \fill[midblue!25]   \VennRing{-7.05}{-0.1}{3.35};
  \fill[deepblue!62]  \VennRing{-6.75}{-3.05}{2.75};
  \fill[deepblue!85]  (-4.9,-1.1) ellipse (1.65 and 1.4);
  % boundaries
  \draw[urlpurple!75, line width=0.7pt] \VennRing{-7.8}{7.8}{4.6};
  \draw[urlpurple, line width=0.7pt]    \VennRing{-7.6}{5.4}{4.25};
  \draw[midblue, line width=0.7pt]      \VennRing{-7.4}{3.0}{3.9};
  \draw[midblue, line width=0.8pt, dash pattern=on 2.5pt off 1.5pt] \VennRing{-7.2}{0.5}{3.6};
  \draw[deepblue, line width=0.7pt]     \VennRing{-7.05}{-0.1}{3.35};
  \draw[deepblue, line width=0.8pt]     \VennRing{-6.75}{-3.05}{2.75};
  \draw[white, line width=0.6pt]        (-4.9,-1.1) ellipse (1.65 and 1.4);
  % class names
  \node[urlpurple!90!black, font=\footnotesize\bfseries] at (6.0,2.0) {not NIP};
  \node[urlpurple!90!black, font=\footnotesize\bfseries] at (3.5,2.05) {NIP};
  \node[vtext, font=\footnotesize\bfseries] at (0.95,2.35) {NIP$^{\mathrm{fc}}$};
  \node[vtext, font=\footnotesize\bfseries] at (-2.3,2.25) {NIP$^{\mathrm{met}}$};
  \node[white, font=\footnotesize\bfseries] at (-4.9,2.6) {stable};
  % contents
  \node[text width=1.9cm, align=center, vtext] at (6.25,-0.15)
  {\textit{$\beta\NN$ inside}\\[3pt] interference fringes of ever finer spacing\\[3pt] a scanned spin chain, roles exchanged};
  \node[text width=1.9cm, align=center, vtext] at (4.05,-0.15)
  {\textit{$\hat{A}(2^\NN)$ inside}\\[3pt] light cones shrinking to points\\[3pt] phases on a circle\\[3pt] templates growing without end\\[4pt] \textbf{$D_5$}};
  \node[text width=1.95cm, align=center, vtext] at (1.5,-0.2)
  {\textit{duplicate inside}\\[3pt] positions in a box\\[1pt] quadratures of light\\[3pt] Ising with offsets, classical and quantum\\[3pt] \textbf{$D_6$, $D_7$}};
  \node[text width=2.55cm, align=center, vtext] at (-1.55,-0.2)
  {\textit{metrizable, unstable}\\[3pt] photon numbers\\[1pt] a scanned spin chain\\[3pt] nonzero limiting chiral condensate, clean metal, ordered Ising\\[4pt] \textbf{$D_1$, $D_2$}};
  \node[text width=3.1cm, align=center, white] at (-4.9,1.35)
  {\textit{continuous}\\[2pt] persistently gapped\\ transverse-current response\\[1pt] disordered Ising};
  \node[white, font=\scriptsize\bfseries, align=center] at (-4.9,-0.15) {uniformly\\ stable};
  \node[text width=2.5cm, align=center, white] at (-4.9,-1.55)
  {\textit{(Part~I)}\\[2pt] swap test\\[1pt] Hong--Ou--Mandel coincidences\\[1pt] equality tests};
  % NIP^{hs} callout
  \draw[vtext, -{Stealth[length=4pt]}, line width=0.5pt] (-2.6,-4.85) -- (-2.2,-3.35);
  \node[text width=8.6cm, align=center, vtext, anchor=north] at (-1.0,-4.85)
  {dashed ring: \textbf{NIP$^{\mathrm{hs}}$}, the split interval. For observables it merges with NIP$^{\mathrm{met}}$,
    since mixtures force metrizability; it lives on in trees of preparations: \textbf{$D_3$, $D_4$}};
\end{tikzpicture}

%%% Local Variables:
%%% mode: LaTeX
%%% TeX-master: "Qbipartite_2_NIP"
%%% End:

%% file: hepta_body.tex
\begin{tikzpicture}[font=\scriptsize, x=1cm, y=1cm]
  \colorlet{vtext}{deepblue!85!black}
  \colorlet{bump}{urlpurple!85!black}
  % column backgrounds and headers
  \HeptaCol{0}{\strut\\ metrizable}{midblue!25}{NIP$^{\mathrm{met}}$}
  \HeptaCol{3.85}{hereditarily\\ separable}{midblue!18}{NIP$^{\mathrm{hs}}$}
  \HeptaCol{7.7}{first\\ countable}{midblue!11}{NIP$^{\mathrm{fc}}$}
  \HeptaCol{11.55}{not first\\ countable}{urlpurple!15}{NIP}
  % ---------------- D1: escape to infinity ----------------
  \begin{scope}[shift={(0,-0.45)}]
    \HeptaPanel{0}{0}{$D_1$}{escape to infinity}{photon numbers}
    \foreach \k/\y in {0/-1.75,1/-1.45,2/-1.18,3/-0.94,4/-0.73}{
        \draw[deepblue!70, line width=0.6pt] (1.15,\y) -- (2.5,\y);
        \fill[bump] (1.825,\y) circle (1.6pt);}
    \draw[-{Stealth[length=4pt]}, deepblue, line width=0.6pt] (1.825,-0.65) -- (1.825,-0.3);
    \node[vtext] at (2.15,-0.36) {$\infty$};
  \end{scope}
  % ---------------- D2: infinite configurations ----------------
  \begin{scope}[shift={(0,-4.2)}]
    \HeptaPanel{0}{0}{$D_2$}{infinite configurations}{a spin chain read site by site}
    \def\bits{1,0,1,1,0,1,0,0,1,1}
    \foreach \row/\len in {0/3,1/5,2/7,3/10}{
    \pgfmathsetmacro{\y}{-0.5-0.38*\row}
    \foreach \b [count=\j from 0] in \bits{
    \ifnum\j<\len
      \pgfmathsetmacro{\x}{0.45+0.27*\j}
      \ifnum\b=0 \draw[-{Stealth[length=2.2pt]}, bump, line width=0.55pt] (\x,\y-0.12) -- (\x,\y+0.12);
      \else \draw[-{Stealth[length=2.2pt]}, deepblue, line width=0.55pt] (\x,\y+0.12) -- (\x,\y-0.12); \fi
    \fi}}
    \node[vtext] at (3.3,-1.64) {$\cdots$};
  \end{scope}
  % ---------------- D3: approach from the right ----------------
  \begin{scope}[shift={(3.85,-0.45)}]
    \HeptaPanel{0}{0}{$D_3$}{approach from the right}{a particle in a box}
    \draw[deepblue, line width=0.6pt] (0.35,-1.6) -- (3.3,-1.6);
    \fill[deepblue] (1.2,-1.6) circle (1.5pt); \node[vtext] at (1.2,-1.82) {$a$};
    \HeptaBump{2.5}{0.5}{0.3}{-1.6}\HeptaBump{1.8}{0.28}{0.55}{-1.6}\HeptaBump{1.45}{0.15}{0.95}{-1.6}
    \draw[-{Stealth[length=3.5pt]}, vtext, line width=0.5pt] (3.0,-0.38) -- (1.35,-0.38);
  \end{scope}
  % ---------------- D4: approach from the left ----------------
  \begin{scope}[shift={(3.85,-4.2)}]
    \HeptaPanel{0}{0}{$D_4$}{approach from the left}{a particle in a box}
    \draw[deepblue, line width=0.6pt] (0.35,-1.6) -- (3.3,-1.6);
    \fill[deepblue] (2.45,-1.6) circle (1.5pt); \node[vtext] at (2.45,-1.82) {$a$};
    \HeptaBump{0.65}{0.5}{0.3}{-1.6}\HeptaBump{1.57}{0.28}{0.55}{-1.6}\HeptaBump{2.05}{0.15}{0.95}{-1.6}
    \draw[-{Stealth[length=3.5pt]}, vtext, line width=0.5pt] (0.65,-0.38) -- (2.3,-0.38);
  \end{scope}
  % ---------------- D6: twins over a Cantor set ----------------
  \begin{scope}[shift={(7.7,-0.45)}]
    \HeptaPanel{0}{0}{$D_6$}{twins over a Cantor set}{a spin chain read by sites and by templates}
    \draw[deepblue, line width=0.6pt] (0.35,-1.6) -- (3.3,-1.6);
    \foreach \x in {0.5,0.72,1.18,1.4,2.25,2.47,2.93,3.15}{
        \fill[deepblue] (\x,-1.6) circle (1.5pt);
        \draw[bump, densely dotted, line width=0.4pt] (\x,-1.53) -- (\x,-0.95);
        \draw[bump, fill=white, line width=0.6pt] (\x,-0.88) circle (1.8pt);}
    \node[vtext,align=center] at (1.825,-0.45) {twins, isolated\\in the duplicate};
  \end{scope}
  % ---------------- D7: approach from both sides ----------------
  \begin{scope}[shift={(7.7,-4.2)}]
    \HeptaPanel{0}{0}{$D_7$}{approach from both sides at once}{a particle in a box, half on each side}
    \draw[deepblue, line width=0.6pt] (0.35,-1.6) -- (3.3,-1.6);
    \fill[deepblue] (1.825,-1.6) circle (1.5pt); \node[vtext] at (1.825,-1.82) {$a$};
    \HeptaBump{0.55}{0.4}{0.22}{-1.6}\HeptaBump{2.7}{0.4}{0.22}{-1.6}
    \HeptaBump{1.2}{0.22}{0.4}{-1.6}\HeptaBump{2.23}{0.22}{0.4}{-1.6}
    \HeptaBump{1.55}{0.12}{0.7}{-1.6}\HeptaBump{2.0}{0.12}{0.7}{-1.6}
    \node[vtext] at (1.825,-0.33) {$\tfrac12+\tfrac12$};
    \draw[-{Stealth[length=3.5pt]}, vtext, line width=0.5pt] (0.6,-0.6) -- (1.45,-0.6);
    \draw[-{Stealth[length=3.5pt]}, vtext, line width=0.5pt] (3.05,-0.6) -- (2.2,-0.6);
  \end{scope}
  % ---------------- D5: light cones shrinking to a point ----------------
  \begin{scope}[shift={(11.55,-2.35)}]
    \HeptaPanel{0}{0}{$D_5$}{a light cone shrinks to a point}{the causal order of spacetime; templates}
    \draw[deepblue, line width=0.7pt] (0.3,-0.55) -- (3.35,-0.55);
    \node[vtext] at (1.825,-0.3) {Bob's slice};
    \foreach \h/\c in {0.95/40,0.6/60,0.32/85}{
        \fill[bump!\c] (1.825,-0.55-\h) circle (1.4pt);
        \draw[bump!\c, line width=0.5pt] (1.825-\h,-0.55) -- (1.825,-0.55-\h) -- (1.825+\h,-0.55);
        \draw[bump!\c, line width=1.6pt] (1.825-\h,-0.55) -- (1.825+\h,-0.55);}
    \fill[deepblue] (1.825,-0.55) circle (1.5pt);
  \end{scope}
\end{tikzpicture}

%%% Local Variables:
%%% mode: LaTeX
%%% TeX-master: "Qbipartite_2_NIP"
%%% End:

%% file: tree_body.tex
\begin{tikzpicture}[font=\scriptsize, x=1cm, y=1cm,
  inner/.style={align=center, inner sep=2pt, text=deepblue!85!black, font=\scriptsize\itshape},
  leaf/.style={rounded corners=3pt, draw=deepblue!70, line width=0.5pt, fill=white,
      text=black, align=center, text width=3.1cm, inner sep=4pt},
  edge/.style={draw=deepblue!70, line width=0.5pt, -{Stealth[length=3.5pt]}}]
  \node[rounded corners=3pt, draw=deepblue!80, line width=0.6pt, align=center, inner sep=5pt,
    font=\footnotesize] (root) at (0,0) {$K$ a separable Rosenthal compactum};
  \node[leaf, anchor=north] (met) at (-4.6,-1.3)
  {\textbf{metrizable}\\[1pt] NIP$^{\mathrm{met}}$\\[1pt] $D_1$ if $K$ is infinite;\\ $D_2$ under the hypotheses of\\ \cref{fact:tricho}(a)};
  \node[inner, anchor=north] (nmet) at (2.6,-1.3) {not metrizable};
  \node[leaf, dashed, anchor=north, text width=3.4cm] (hs) at (-0.4,-2.9)
    {\textbf{hereditarily separable}\\[1pt] NIP$^{\mathrm{hs}}$; contains $\hat{S}_+(2^\NN)$\\ or $\hat{S}_-(2^\NN)$\\[1pt] $D_3$ or $D_4$\\[1pt] \textit{empty for full observable compacta}};
  \node[inner, anchor=north] (nhs) at (5.2,-2.9) {not hereditarily separable};
  \node[leaf, anchor=north, text width=3.0cm] (fc) at (3.4,-4.6)
  {\textbf{first countable}\\[1pt] NIP$^{\mathrm{fc}}$; contains $\hat{D}(2^\NN)$\\ or $\hat{D}(S(2^\NN))$\\[1pt] $D_6$ or $D_7$};
  \node[leaf, anchor=north, text width=3.0cm] (nfc) at (7.0,-4.6)
  {\textbf{not first countable}\\[1pt] NIP; contains $\hat{A}(2^\NN)$\\[1pt] $D_5$};
  \draw[edge] (root.south) -- (met.north);
  \draw[edge] (root.south) -- (nmet.north);
  \draw[edge] (nmet.south) -- (hs.north);
  \draw[edge] (nmet.south) -- (nhs.north);
  \draw[edge] (nhs.south) -- (fc.north);
  \draw[edge] (nhs.south) -- (nfc.north);
\end{tikzpicture}

%%% Local Variables:
%%% mode: LaTeX
%%% TeX-master: "Qbipartite_2_NIP"
%%% End:

%% file: Qbipartite_2_NIP.bbl
\newcommand{\etalchar}[1]{$^{#1}$}
\def\cprime{$'$}
\providecommand{\bysame}{\leavevmode\hbox to3em{\hrulefill}\thinspace}
\providecommand{\MR}{\relax\ifhmode\unskip\space\fi MR }
\providecommand{\MRhref}[2]{%
  \href{http://www.ams.org/mathscinet-getitem?mr=#1}{#2}
}
\providecommand{\href}[2]{#2}
\begin{thebibliography}{ABDCBH97}

\bibitem[Aar07]{Aaronson:2007}
Scott Aaronson, \emph{The learnability of quantum states}, Proceedings of the
  Royal Society A \textbf{463} (2007), 3089--3114.

\bibitem[ABDCBH97]{Alon-BenDavid-CesaBianchi-Haussler:1997}
Noga Alon, Shai Ben-David, Nicol{\`o} Cesa-Bianchi, and David Haussler,
  \emph{Scale-sensitive dimensions, uniform convergence, and learnability},
  Journal of the ACM \textbf{44} (1997), no.~4, 615--631.

\bibitem[ADK08]{Argyros-Dodos-Kanellopoulos:2008b}
Spiros~A. Argyros, Pandelis Dodos, and Vassilis Kanellopoulos, \emph{A
  classification of separable {Rosenthal} compacta and its applications},
  Dissertationes Mathematicae \textbf{449} (2008), 1--55.

\bibitem[Ai92]{Arkhangelskii:1992}
A.~V. Arkhangelski\u~i, \emph{Topological function spaces}, Mathematics and its
  Applications (Soviet Series), vol.~78, Kluwer Academic Publishers Group,
  Dordrecht, 1992, Translated from the Russian by R. A. M. Hoksbergen.
  \MR{1144519}

\bibitem[Asa70]{Asano:1970b}
Taro Asano, \emph{Lee--{Y}ang theorem and the {G}riffiths inequality for the
  anisotropic {H}eisenberg ferromagnet}, Physical Review Letters \textbf{24}
  (1970), 1409--1411.

\bibitem[Ass83]{Assouad:1983}
Patrick Assouad, \emph{Densit\'e et dimension}, Annales de l'Institut Fourier
  \textbf{33} (1983), no.~3, 233--282.

\bibitem[BC80]{Banks-Casher:1980}
Tom Banks and A.~Casher, \emph{Chiral symmetry breaking in confining theories},
  Nuclear Physics B \textbf{169} (1980), 103--125.

\bibitem[{Ben}08]{Yaacov:2008}
Ita{\"i} {Ben Yaacov}, \emph{Continuous first order logic for unbounded metric
  structures}, Journal of Mathematical Logic \textbf{8} (2008), no.~2,
  197--223.

\bibitem[{Ben}09]{Yaacov:2009}
\bysame, \emph{Continuous and random {V}apnik--{C}hervonenkis classes}, Israel
  Journal of Mathematics \textbf{173} (2009), 309--333.

\bibitem[BFT78]{Bourgain-Fremlin-Talagrand:1978}
J.~Bourgain, D.~H. Fremlin, and M.~Talagrand, \emph{Pointwise compact sets of
  {B}aire-measurable functions}, Amer. J. Math. \textbf{100} (1978), no.~4,
  845--886. \MR{509077}

\bibitem[BG09]{Bjornberg-Grimmett:2009}
Jakob~E. Bj{\"o}rnberg and Geoffrey~R. Grimmett, \emph{The phase transition of
  the quantum {I}sing model is sharp}, Journal of Statistical Physics
  \textbf{136} (2009), 231--273, arXiv:0901.0328.

\bibitem[BK90]{Borgs-Kotecky:1990}
Christian Borgs and Roman Koteck{\'y}, \emph{A rigorous theory of finite-size
  scaling at first-order phase transitions}, Journal of Statistical Physics
  \textbf{61} (1990), 79--119.

\bibitem[BYU10]{BenYaacov-Usvyatsov:2010}
Ita{\"{\i}} Ben~Yaacov and Alexander Usvyatsov, \emph{Continuous first order
  logic and local stability}, Trans. Amer. Math. Soc. \textbf{362} (2010),
  no.~10, 5213--5259. \MR{2657678}

\bibitem[CHY15]{Cheng-Hsieh-Yeh:2015}
Hao-Chung Cheng, Min-Hsiu Hsieh, and Ping-Cheng Yeh, \emph{The learnability of
  unknown quantum measurements}, 2015, arXiv:1501.00559.

\bibitem[Deb13]{Debs:2013}
Gabriel Debs, \emph{Descriptive aspects of {Rosenthal} compacta}, Recent
  Progress in General Topology III, Atlantis Press, 2013, pp.~205--227.

\bibitem[DI26]{Duenez-Iovino:2026a}
E.~Due{\~n}ez and J.~Iovino, \emph{Classification of bipartite observables {I}:
  stable and uniformly stable observables}, preprint; arXiv:2610.11002, 2026.

\bibitem[DIMW{\etalchar{+}}26]{Duenez-Iovino-MatosWiederhold-Salvetti-Tall:2026}
E.~Due{\~n}ez, J.~Iovino, T.~Matos-Wiederhold, L.~Salvetti, and F.~D. Tall,
  \emph{Complexity of deep computations via topology of function spaces}, 2026,
  arXiv:2601.00528.

\bibitem[FHS13]{Farah-Hart-Sherman:2013}
Ilijas Farah, Bradd Hart, and David Sherman, \emph{Model theory of operator
  algebras {I}: Stability}, Bull. Lond. Math. Soc. \textbf{45} (2013), no.~4,
  825--838.

\bibitem[FR12]{Frohlich-Rodriguez:2012b}
J{\"u}rg Fr{\"o}hlich and Pierre-Fran{\c{c}}ois Rodr{\'\i}guez, \emph{Some
  applications of the {L}ee--{Y}ang theorem}, Journal of Mathematical Physics
  \textbf{53} (2012), 095218, arXiv:1205.6643.

\bibitem[FV17]{Friedli-Velenik:2017}
Sacha Friedli and Yvan Velenik, \emph{Statistical mechanics of lattice systems:
  A concrete mathematical introduction}, Cambridge University Press, 2017.

\bibitem[Gan40]{Gantmacher:1940}
V.~Gantmacher, \emph{{\"Uber} schwache totalstetige {Operatoren}}, Mat. Sbornik
  N.S. \textbf{7(49)} (1940), 301--308.

\bibitem[Gin69]{Ginibre:1969}
Jean Ginibre, \emph{Existence of phase transitions for quantum lattice
  systems}, Communications in Mathematical Physics \textbf{14} (1969),
  205--234.

\bibitem[GM22]{Glasner-Megrelishvili:2022}
Eli Glasner and Michael Megrelishvili, \emph{{Todorčević'} trichotomy and a
  hierarchy in the class of tame dynamical systems}, Transactions of the
  American Mathematical Society \textbf{375} (2022), no.~7, 4513--4548.

\bibitem[God80]{Godefroy:1980}
Gilles Godefroy, \emph{Compacts de {R}osenthal}, Pacific J. Math. \textbf{91}
  (1980), no.~2, 293--306. \MR{615679}

\bibitem[Gri64]{Griffiths:1964}
Robert~B. Griffiths, \emph{Peierls proof of spontaneous magnetization in a
  two-dimensional {I}sing ferromagnet}, Physical Review \textbf{136} (1964),
  A437--A439.

\bibitem[Gro52]{Grothendieck:1952}
A.~Grothendieck, \emph{Crit\`eres de compacit\'{e} dans les espaces
  fonctionnels g\'{e}n\'{e}raux}, Amer. J. Math. \textbf{74} (1952), 168--186.
  \MR{47313}

\bibitem[Gru88]{Gruenhage:1988b}
Gary Gruenhage, \emph{{[TITLE TO BE CONFIRMED]}}, General Topology and its
  Relations to Modern Analysis and Algebra {VI} ({P}roc. {S}ixth {P}rague
  {T}opological {S}ymposium, 1986) (Zden{\v{e}}k Frol{\'\i}k, ed.), Research
  and Exposition in Mathematics, vol.~16, Heldermann, Berlin, 1988.

\bibitem[GS84]{Gurevich-Schmitt:1984}
Yuri Gurevich and Peter~H. Schmitt, \emph{The theory of ordered abelian groups
  does not have the independence property}, Transactions of the American
  Mathematical Society \textbf{284} (1984), no.~1, 171--182.

\bibitem[HL66]{Halpern-Lauchli:1966}
J.~D. Halpern and H.~Läuchli, \emph{A partition theorem}, Transactions of the
  American Mathematical Society \textbf{124} (1966), 360--367.

\bibitem[HP11]{Hrushovski-Pillay:2011b}
Ehud Hrushovski and Anand Pillay, \emph{On {NIP} and invariant measures},
  Journal of the European Mathematical Society \textbf{13} (2011), no.~4,
  1005--1061.

\bibitem[Kad51]{Kadison:1951a}
Richard~V. Kadison, \emph{A representation theory for commutative topological
  algebra}, Memoirs of the American Mathematical Society, vol.~7, American
  Mathematical Society, 1951.

\bibitem[Kei23]{Keisler:2023}
H.~Jerome Keisler, \emph{Model theory for real-valued structures}, Beyond
  First-Order Model Theory, Vol.\ 2 (Jos{\'e} Iovino, ed.), CRC Press, 2023,
  arXiv:2005.11851, pp.~43--101.

\bibitem[Koh64]{Kohn:1964}
Walter Kohn, \emph{Theory of the insulating state}, Physical Review
  \textbf{133} (1964), A171--A181.

\bibitem[KPS86]{Knight-Pillay-Steinhorn:1986}
Julia~F. Knight, Anand Pillay, and Charles Steinhorn, \emph{Definable sets in
  ordered structures. {II}}, Transactions of the American Mathematical Society
  \textbf{295} (1986), no.~2, 593--605.

\bibitem[KR83]{Kadison-Ringrose:1983a}
Richard~V. Kadison and John~R. Ringrose, \emph{Fundamentals of the theory of
  operator algebras, vol.~i}, Academic Press, 1983.

\bibitem[Lie94]{Lieb:1994}
Elliott~H. Lieb, \emph{Flux phase of the half-filled band}, Physical Review
  Letters \textbf{73} (1994), 2158--2161.

\bibitem[LS81]{Lieb-Sokal:1981}
Elliott~H. Lieb and Alan~D. Sokal, \emph{A general {L}ee--{Y}ang theorem for
  one-component and multicomponent ferromagnets}, Communications in
  Mathematical Physics \textbf{80} (1981), no.~2, 153--179.

\bibitem[LS92]{Leutwyler-Smilga:1992}
H.~Leutwyler and A.~Smilga, \emph{Spectrum of {D}irac operator and role of
  winding number in {QCD}}, Physical Review D \textbf{46} (1992), 5607--5632.

\bibitem[LSV93]{Louveau-Shelah-Velickovic:1993}
Alain Louveau, Saharon Shelah, and Boban Veličkovi{\'{c}}, \emph{Borel
  partitions of infinite subtrees of a perfect tree}, Annals of Pure and
  Applied Logic \textbf{63} (1993), no.~3, 271--281.

\bibitem[LY52]{Lee-Yang:1952b}
T.~D. Lee and C.~N. Yang, \emph{Statistical theory of equations of state and
  phase transitions. {II}. {L}attice gas and {I}sing model}, Physical Review
  \textbf{87} (1952), 410--419.

\bibitem[Mil81]{Milliken:1981}
Keith Milliken, \emph{A partition theorem for the infinite subtrees of a tree},
  Transactions of the American Mathematical Society \textbf{263} (1981),
  137--148.

\bibitem[New74]{Newman:1974}
Charles~M. Newman, \emph{Zeros of the partition function for generalized
  {I}sing systems}, Communications on Pure and Applied Mathematics \textbf{27}
  (1974), 143--159.

\bibitem[Pei36]{Peierls:1936}
Rudolf Peierls, \emph{On {I}sing's model of ferromagnetism}, Proceedings of the
  Cambridge Philosophical Society \textbf{32} (1936), 477--481.

\bibitem[PS86]{Pillay-Steinhorn:1986a}
Anand Pillay and Charles Steinhorn, \emph{Definable sets in ordered structures.
  {I}}, Transactions of the American Mathematical Society \textbf{295} (1986),
  no.~2, 565--592.

\bibitem[Rod21]{Rodriguez:2021}
Jos{\'e} Rodr{\'i}guez, \emph{$\varepsilon$-weakly precompact sets in {B}anach
  spaces}, 2021, arXiv:2102.01546.

\bibitem[Ros74]{Rosenthal:1974}
H.~P. Rosenthal, \emph{A characterization of {B}anach spaces containing
  $l\sp{1}$}, Proc. Nat. Acad. Sci. U.S.A. \textbf{71} (1974), 2411--2413.
  \MR{358307}

\bibitem[Sau72]{Sauer:1972}
N.~Sauer, \emph{On the density of families of sets}, J. Combinatorial Theory
  Ser. A \textbf{13} (1972), 145--147. \MR{307902}

\bibitem[SF71]{Suzuki-Fisher:1971b}
Masuo Suzuki and Michael~E. Fisher, \emph{Zeros of the partition function for
  the {H}eisenberg, ferroelectric, and general {I}sing models}, Journal of
  Mathematical Physics \textbf{12} (1971), 235--246.

\bibitem[SG73]{Simon-Griffiths:1973}
Barry Simon and Robert~B. Griffiths, \emph{The $(\phi^4)_2$ field theory as a
  classical {I}sing model}, Communications in Mathematical Physics \textbf{33}
  (1973), 145--164.

\bibitem[She71]{Shelah:1971b}
Saharon Shelah, \emph{Stability, the f.c.p., and superstability; model
  theoretic properties of formulas in first order theory}, Ann. Math. Logic
  \textbf{3} (1971), no.~3, 271--362. \MR{0317926}

\bibitem[She72]{Shelah:1972b}
\bysame, \emph{A combinatorial problem; stability and order for models and
  theories in infinitary languages}, Pacific J. Math. \textbf{41} (1972),
  247--261. \MR{307903}

\bibitem[She90]{Shelah:1990a}
S.~Shelah, \emph{Classification theory and the number of nonisomorphic models},
  second ed., Studies in Logic and the Foundations of Mathematics, vol.~92,
  North-Holland Publishing Co., Amsterdam, 1990. \MR{MR1083551 (91k:03085)}

\bibitem[Sim93]{Simon:1993}
Barry Simon, \emph{The statistical mechanics of lattice gases}, vol.~I,
  Princeton University Press, Princeton, NJ, 1993.

\bibitem[Sim15a]{Simon:2015b}
Pierre Simon, \emph{A guide to {NIP} theories}, Lecture Notes in Logic,
  vol.~44, Association for Symbolic Logic, Chicago, IL; Cambridge Scientific
  Publishers, Cambridge, 2015. \MR{3560428}

\bibitem[Sim15b]{Simon:2015a}
\bysame, \emph{Rosenthal compacta and {NIP} formulas}, Fund. Math. \textbf{231}
  (2015), no.~1, 81--92. \MR{3361236}

\bibitem[Suz76]{Suzuki:1976b}
Masuo Suzuki, \emph{Relationship between $d$-dimensional quantal spin systems
  and $(d+1)$-dimensional {I}sing systems: equivalence, critical exponents and
  systematic approximants of the partition function and spin correlations},
  Progress of Theoretical Physics \textbf{56} (1976), 1454--1469.

\bibitem[SWZ93]{Scalapino-White-Zhang:1993}
Douglas~J. Scalapino, Steven~R. White, and Shoucheng Zhang, \emph{Insulator,
  metal, or superconductor: The criteria}, Physical Review B \textbf{47}
  (1993), 7995--8007.

\bibitem[Tit39]{Titchmarsh:1939}
E.~C. Titchmarsh, \emph{The theory of functions}, second ed., Oxford University
  Press, Oxford, 1939.

\bibitem[Tka11]{Tkachuk:2011}
Vladimir~V. Tkachuk, \emph{A {$C_p$}-theory problem book: Topological and
  function spaces}, Problem Books in Mathematics, Springer, 2011.

\bibitem[Tod97]{Todorcevic:1997}
Stevo Todorcevic, \emph{Topics in topology}, Lecture Notes in Mathematics, vol.
  1652, Springer Berlin, Heidelberg, 1997.

\bibitem[Tod99]{Todorcevic:1999}
Stevo Todor{\v{c}}evi\'{c}, \emph{Compact subsets of the first {B}aire class},
  J. Amer. Math. Soc. \textbf{12} (1999), no.~4, 1179--1212. \MR{1685782}

\bibitem[Val84]{Valiant:1984}
L.~G. Valiant, \emph{A theory of the learnable}, Commun. ACM \textbf{27}
  (1984), no.~11, 1134–1142.

\bibitem[VC71]{Vapnik-Chervonenkis:1971}
Vladimir~N Vapnik and Alexey~Ya Chervonenkis, \emph{On the uniform convergence
  of relative frequencies of events to their probabilities}, Theory of
  Probability \& Its Applications \textbf{16} (1971), no.~2, 264--280.

\bibitem[VC74]{Vapnik-Chervonenkis:1974}
\bysame, \emph{Theory of pattern recognition}, Nauka, Moscow, 1974, German
  Translation: Theorie der Zeichenerkennung, Akademie-Verlag, Berlin, 1979.

\bibitem[VW00]{Verbaarschot-Wettig:2000}
J.~J.~M. Verbaarschot and T.~Wettig, \emph{Random matrix theory and chiral
  symmetry in {QCD}}, Annual Review of Nuclear and Particle Science \textbf{50}
  (2000), 343--410.

\bibitem[Wie49]{Wielandt:1949}
Helmut Wielandt, \emph{{\"U}ber die {U}nbeschr{\"a}nktheit der {O}peratoren der
  {Q}uantenmechanik}, Mathematische Annalen \textbf{121} (1949), 21.

\bibitem[Wil96]{Wilkie:1996}
A.~J. Wilkie, \emph{Model completeness results for expansions of the ordered
  field of real numbers by restricted {P}faffian functions and the exponential
  function}, Journal of the American Mathematical Society \textbf{9} (1996),
  no.~4, 1051--1094.

\bibitem[Win47]{Wintner:1947}
Aurel Wintner, \emph{The unboundedness of quantum-mechanical matrices},
  Physical Review \textbf{71} (1947), 738--739.

\bibitem[Zil16]{Zilber:2016b}
Boris Zilber, \emph{The semantics of the canonical commutation relation}, 2016,
  arXiv:1604.07745.

\end{thebibliography}
